\def\draft{0}
\documentclass{article}
\usepackage{macros}

\title{Classical Obfuscation of Quantum Circuits \\ via Publicly-Verifiable QFHE}
\author{James Bartusek\thanks{\texttt{bartusek.james@gmail.com}}\\Columbia \and Aparna Gupte\thanks{\texttt{agupte@mit.edu}}\\MIT \and Saachi Mutreja\thanks{\texttt{sm5540@columbia.edu}}\\Columbia \and Omri Shmueli\thanks{\texttt{omri.shmueli1@gmail.com}}\\NTT Research}
\date{\today}

\begin{document}

\maketitle
\begin{abstract}

A classical obfuscator for quantum circuits is a classical program that, given the classical description of a quantum circuit $Q$, outputs the classical description of a functionally equivalent quantum circuit $\widetilde{Q}$ that hides as much as possible about $Q$. Previously, the \emph{only} known feasibility result for classical obfuscation of quantum circuits (Bartusek and Malavolta, ITCS 2022) was limited to ``null'' security, which is only meaningful for circuits that always reject. On the other hand, if the obfuscator is allowed to compile the quantum circuit $Q$ into a \emph{quantum state} $\ket{\widetilde{Q}}$, there exist feasibility results for obfuscating much more expressive classes of circuits: All pseudo-deterministic quantum circuits (Bartusek, Kitagawa, Nishimaki and Yamakawa, STOC 2023, Bartusek, Brakerski and Vaikuntanathan, STOC 2024), and even all unitaries (Huang and Tang, FOCS 2025).

\vspace{1mm} 

We show that (relative to a classical oracle) there exists a \emph{classical} obfuscator for all pseudo-deterministic quantum circuits. As our main technical step, we give the first construction of a compact quantum fully-homomorphic encryption (QFHE) scheme that supports public verification of (pseudo-deterministic) quantum evaluation, relative to a classical oracle. 

\vspace{1mm}

To construct our QFHE scheme, we improve on an approach introduced by Bartusek, Kitagawa, Nishimaki and Yamakawa (STOC 2023), which previously required ciphertexts that are both \emph{quantum} and \emph{non-compact} due to a heavy use of quantum coset states and their publicly-verifiable properties. As part of our core technical contribution, we introduce new techniques for analyzing coset states that can be generated ``on the fly'', by proving new cryptographic properties of the one-shot signature scheme of Shmueli and Zhandry (CRYPTO 2025). Our techniques allow us to produce QFHE ciphertexts that are purely classical, compact, and publicly-verifiable. This additionally yields the first classical verification of quantum computation protocol for BQP that simultaneously satisfies blindness and public-verifiability.
\end{abstract}


\newpage
\tableofcontents
\newpage

\section{Introduction}  \label{section:introduction}

Program obfuscation disguises the inner workings of an algorithm while leaving its behavior unchanged. This concept traces back to the roots of modern cryptography \cite{diffie2022new}, and ever since the formalization of \emph{indistinguishability} obfuscation (iO) \cite{barak2001possibility}, constructing iO has been one of cryptography's overarching goals. It took no less than a decade to come up with the first candidates for secure iO \cite{garg2016candidate}, and almost an additional decade to base classically-secure iO on well-studied assumptions \cite{jain2021indistinguishability}. Current efforts continue along this trajectory, with an active line of work devoted to developing quantum-secure iO, eventually grounded in similarly well-understood assumptions (e.g.\ \cite{gentry2015graph, bartusek2018return, chen2018ggh15, SP20, brakerski2023candidate, wee2021candidate},...).

The original notion of program obfuscation implicitly considers computation as classical, and accordingly, seeks to obfuscate classical algorithms. Now, as quantum computation increasingly emerges as a standard model of computation, it is natural to revisit the problem in this broader context: Can quantum algorithms be obfuscated as well? 

\paragraph{Ideal Program Obfuscation: From Classical to Quantum.}
Our understanding of obfuscation in the quantum setting is still at a very early stage: In particular, we still don't know how to base indistinguishability obfuscation for expressive classes of quantum circuits on (post-quantum) indistinguishability obfuscation for classical circuits.\footnote{The only existing result from indistinguishability obfuscation is limited to quantum circuits with logarithmically-many non-Clifford gates \cite{BK21}.} Instead, recent work has turned to exploring the relationship between \emph{idealized} notions of obfuscation in the classical and quantum worlds \cite{BM22,bartusek2023obfuscation, bartusek2024quantum, huang2025obfuscation}. In such an idealized model, 
an obfuscation of a circuit $C$ is considered secure if everything that can be learned from the obfuscation $\widetilde{C}$ can be simulated given only \emph{oracle access} to $C$. 

Of course, ideal obfuscation cannot be achieved in full generality, as there are well-known counterexamples demonstrating its impossibility \cite{barak2001possibility}. Nevertheless, studying the connection between classical and quantum ideal obfuscation remains valuable for at least two reasons. First, it isolates a fundamental question: Does the existence of an ideal obfuscator for classical computation imply one for quantum? Second, indistinguishability obfuscation is widely regarded as a strong heuristic for ideal obfuscation (e.g., it is the best-possible obfuscation \cite{goldwasser2007best}). Thus, establishing that ideal classical obfuscation gives rise to ideal quantum obfuscation results in 
candidates for quantum obfuscation by instantiating the black-box with (post-quantum) indistinguishability obfuscation.
 
Prior work \cite{bartusek2023obfuscation, bartusek2024quantum, huang2025obfuscation} has established that, assuming a classical ideal obfuscator for classical circuits, one can construct a \emph{quantum} ideal obfuscator for quantum circuits. In this setting, the obfuscator itself is a quantum algorithm, and on input a quantum circuit $Q$ it outputs a quantum state $\ket{\widetilde{Q}}$ serving as the obfuscated program.\footnote{We note that \cite{bartusek2024quantum, huang2025obfuscation} are able to do something slightly stronger---they obfuscate quantum circuits that also have an auxiliary quantum state. This is particularly relevant in some cryptographic applications, e.g. copy-protection.} A distinctive feature of these constructions is that the obfuscated object consists of \emph{unclonable} states: Once $\ket{\widetilde{Q}}$ is produced, it cannot be duplicated. As a consequence, it is impossible to create a single obfuscation of a quantum circuit and distribute it broadly. Moreover, these obfuscations cannot be transmitted over the classical channels that comprise the internet today.


However, the goal of constructing a classical obfuscator for quantum circuits (with security for non-null circuits\footnote{Bartusek and Malavolta~\cite{BM22} give a feasibility result for classical obfuscation of null quantum circuits, that is, circuits that always reject.}), even in the ideal setting, has so far remained out of reach. This raises our central motivating question:
\begin{quote}
\centering
    \textit{Is a quantum obfuscator necessary to obfuscate expressive quantum circuits?\\
    Or can a classically-described quantum program be obfuscated by a \emph{classical} obfuscator?}
\end{quote}
We address this question by constructing a classical obfuscator for all pseudo-deterministic circuits.

\subsection{Classical Obfuscation via Publicly-Verifiable QFHE}
Previous work \cite{bartusek2023obfuscation} identified \emph{pseudo-deterministic} quantum circuits as a particularly useful class for obfuscation. A pseudo-deterministic quantum circuit $Q$ from $n$ to $m$ qubits is one that, up to negligible error, computes a classical function: there exists a function $F: \{0,1\}^n \to \{0,1\}^m$ such that for every input $x \in \{0,1\}^n$, the output of $Q(x)$ equals $F(x)$ with probability $1 - \negl(n)$, where the probability is taken over the measurement of its output. This class is both expressive—capturing, for example, all quantum algorithms that solve problems in $\NP$, including Shor’s factoring algorithm \cite{shor1999polynomial}—and structurally well-suited for cryptographic applications like obfuscation. Indeed, pseudo-deterministic circuits have already served as the foundation for obfuscating increasingly general classes of quantum computation \cite{bartusek2024quantum, huang2025obfuscation}. Motivated by these properties, our work focuses on the obfuscation of pseudo-deterministic circuits. Crucially, we consider obfuscation through only classical means.

As also observed in previous work \cite{bartusek2023obfuscation}, a natural route toward obfuscating pseudo-deterministic circuits relies on constructing quantum fully-homomorphic encryption (QFHE) with public verification of homomorphic evaluation for such circuits.
\begin{itemize}
    \item 
    A QFHE scheme is defined by algorithms $\left( \Gen,\Enc,\Eval,\Dec \right)$, where all but the evaluation algorithm $\Eval$ can be regarded as classical. On input a ciphertext $\ct = \Enc(x)$ and a quantum circuit $Q$, the evaluation algorithm $\Eval$ produces an evaluated ciphertext $\widetilde{\ct}$ which encrypts the classical value $Q(x)$.
    
    \item
    In the publicly-verifiable variant of QFHE, $\Eval$ additionally outputs a succinct classical proof $\pi$, and correctness can be certified by an efficient classical verification algorithm $\Ver$, which checks that the homomorphic evaluation has been honestly computed, and outputs $\Ver(\ct, Q, \widetilde{\ct}, \pi) \in \{\top, \bot\}$.
\end{itemize}

It is straightforward to see that, given such a powerful QFHE, obfuscation of pseudo-deterministic quantum circuits follows naturally. This can be accomplished by obfuscating the verifier and applying the standard “verify-then-decrypt” paradigm. Concretely, the obfuscation of a circuit $Q$ consists of encrypting its classical description, i.e., $\ct_Q \gets \Enc(Q)$. On input $x \in \{0,1\}^n$ and the encrypted circuit $\ct_Q$, the evaluator computes $(\widetilde{\ct}, \pi) \gets \Eval(\ct_Q, U_x)$, where $U_x$ denotes the universal quantum circuit with input $x$ hard-wired. To recover the (classical) output $Q(x)$, the obfuscation makes use of a classical oracle implementing the circuit

$$
C(x, \widetilde{\ct}, \pi) := 
\begin{cases}
\Dec(\widetilde{\ct}) & \text{if } \Ver(\ct_Q, U_x, \widetilde{\ct}, \pi) = \top \\
\bot & \text{otherwise} 
\end{cases}
$$
\noindent
Thus, only honestly evaluated ciphertexts are ever decrypted, and the obfuscation securely reproduces the functionality of $Q$.

\paragraph{The central role of Pauli Functional Commitments.} However, it turns out that verification of QFHE evaluations is a particularly delicate problem in quantum cryptography.
\begin{itemize}
    \item
    On the one hand, by definition, evaluating QFHE ciphertexts is a quantum \emph{sampling} task. The quantum QFHE evaluation gets the classical $\ct$ and outputs a quantumly computed classical sample $\widetilde{\ct}$, which is going to have high entropy \emph{even if} the underlying evaluated circuit $Q$ is deterministic. In fact, this entropy was shown to be somewhat inherent in \cite{shmueli2022public, shmueli2022semi}. This seemingly presents a barrier, since it remains open how to classically verify quantum sampling problems with negligible soundness error (\cite{chung2022constant} has shown how to achieve inverse-polynomial error, which, if used in the template above, would allow a polynomial-time attack on the obfuscation scheme).

    \item 
    On the other hand, the special structure of certain QFHE schemes, most notably Mahadev’s construction \cite{mahadev2020classical}, enables a partial solution. As shown in \cite{bartusek2021secure}, the use of Mahadev’s measurement protocol \cite{mah18} on Mahadev's QFHE scheme allows for verification of QFHE evaluations with negligible soundness error. However, since the measurement protocol of Mahadev is privately verifiable, then so is the resulting QFHE scheme.
\end{itemize}

The technical centerpiece behind Mahadev’s measurement protocol was later abstracted by \cite{bartusek2023obfuscation} as a Pauli functional commitment (PFC). Roughly speaking, a PFC allows a classical verifier to issue a public classical commitment key $\CK$. A quantum prover can then commit to a quantum state in such a way that, later, measurements of Pauli observables (in particular, standard or Hadamard basis measurements) on the state can be verified through classical communication. Viewed through this lens, the state of the art can be summarized as follows: Mahadev \cite{mah18} constructs a privately verifiable PFC; Bartusek \cite{bartusek2021secure} demonstrates that applying Mahadev's PFCs to QFHE homomorphic evaluation yields privately verifiable evaluations; and finally, \cite{bartusek2023obfuscation} show how to construct, relative to a classical oracle, PFCs with public verification, albeit with a \emph{quantum} commitment key $\ket{\CK}$.

Measurement protocols, or in their formal abstraction as PFCs, were originally devised to allow a classical client to force a quantum server to perform some prescribed computation. Existing results on public verification, however, require the commitment keys to remain quantum, which prevents them from being generated by a classical verifier. This leads us to what we view as a fascinating technical question in its own right: can one design classical measurement protocols that are publicly verifiable? As we have explained, resolving this question would not only advance the theory of quantum cryptography but would also translate directly to our applications, enabling public verification of QFHE evaluations and, ultimately, classical obfuscation of quantum circuits.

\subsection{Results}
We summarize our main results in this section.
\begin{theorem}[Informal]
    Assuming the quantum hardness of the Learning With Errors problem, there exists classical obfuscation for (classically-described) pseudo-deterministic quantum circuits in the classical oracle model.
\end{theorem}
Along the way, as a tool for our obfuscation scheme, we build the first publicly-verifiable QFHE scheme for pseudo-deterministic circuits.
\begin{theorem}[Informal]
    Assuming the quantum hardness of the Learning With Errors problem (resp. plus an appropriate circular-security assumption), there exists a leveled (resp. unleveled) quantum fully-homomorphic encryption in the classical oracle model, with the following properties.
    \begin{itemize}
        \item Classical encryption of classical messages.
        \item Succinct, classical, and public verification of evaluation for pseudo-deterministic circuits. 
    \end{itemize}
\end{theorem}


\paragraph{Main Technical Contribution.}
Put plainly, our main technical contribution is the first publicly-verifiable quantum measurement protocol in the classical oracle model. Before our work, a quantum prover could prove the integrity of a measurement (standard-basis or Hadamard) on a quantum state only to a single verifier holding a trapdoor. 
More formally, we construct the first publicly-verifiable Pauli functional commitment scheme with \emph{classical keys}. As a central part of our construction and security proof, we construct one-shot signatures (and hash functions) with what we call \emph{delayed collapsing}. Delayed collapsing was implicitly introduced in \cite{bartusek2023obfuscation}, in the context of random coset states (and we intuitively explain how it relates to Pauli functional commitments, in Section \ref{subsec:techniques_1}). We identify delayed collapsing as a fundamental property for hash functions; a delayed collapsing hash implies both collapsing hash functions and non-collapsing collision-resistant hash functions. We prove that the one-shot signatures construction of \cite{shmueli2025one} induces not only a non-collapsing collision-resistant hash function (which was previously known), but a delayed collapsing hash function. To prove our results, we show a new stronger analysis of one-shot signatures, and prove new quantum information theoretical properties of random coset states, which may be of independent interest. See more details in our Technical Overview (Section \ref{subsec:techniques}).


\vspace{2mm}

As an additional contribution that may be of independent interest, we show that post-quantum \emph{succinct} ideal obfuscation for classical Turing machines exists in the classical oracle model (Section \ref{sec:input-succinct}). We do this via a natural ``FHE + SNARK'' approach, and combine techniques from \cite{chiesa2019snark, zhandry2019record} in order to establish post-quantum security.


Finally, we note that our quantum obfuscation scheme is in fact \emph{succinct}, meaning that the size of the obfuscated program grows only with the \emph{description size} of the plaintext quantum program, as opposed to the number of physical gates in the circuit. We emphasize that, even though our construction includes a classical oracle, the classical circuit implemented by the oracle is itself succinct. This means that succinctness is preserved when the classical oracle is heuristically instantiated with post-quantum indistinguishability obfuscation.

As an immediate corollary, we obtain the first \emph{succinct randomized encoding} scheme for pseudo-deterministic quantum computation. A succinct randomized encoding allows for compiling a circuit $Q$ and an input $x$ into a succinct representation $\widehat{Q(x)}$ in time much less than evaluating $Q$, such that $\widehat{Q(x)}$ reveals only $Q(x)$ and nothing else about $x$. Succinct randomized encodings have had several important applications in classical cryptography \cite{BGLPT15}, and we expect them to be widely useful in the quantum setting as well.

\subsection{Discussion}

\paragraph{Comparison with prior work on CVQC.} Our publicly-verifiable QFHE protocol naturally yields a new type of classical verification of quantum computation (CVQC) protocol for BQP. Indeed, given some quantum computation $Q$, the verifier simply encrypts their input $x$ as $\Enc(x)$, and the prover responds with $\Enc(Q(x))$ along with a proof $\pi$ of correct evaluation. 

This protocol has several desirable properties: The prover remains blind to the verifier's input $x$, the proof $\pi$ is publicly-verifiable (implying for example that the prover cannot break security by repeatedly querying the verifier and learning whether they accept or reject), and the client's work does not grow with the size of $Q$. In particular, we obtain the first CVQC protocol that is both \emph{blind} and \emph{publicly-verifiable}, and on top of that it is both non-interactive and succinct. We summarize a comparison with prior work in Table \ref{tab:comparison}.

\paragraph{Open problems.} Our work highlights the utility of classical-client publicly-verifiable QFHE, both for its applications in CVQC and in obfuscation of quantum computation. While we construct this primitive in a manner that supports pseudo-deterministic evaluation, we leave open some interesting questions for future work.
\begin{itemize}
    \item Can we construct classical-client publicly-verifiable QFHE for \emph{all} possible functionalities (i.e.\ quantum circuits that yield an arbitrary distribution over classical outputs)? This would be interesting to achieve even with non-negligible soundness (achieving negligible soundness for sampling problems is a well-known and difficult open problem), and even without compactness.
    \item Can we prove security of classical-client publicly-verifiable QFHE from post-quantum indistinguishability obfuscation? Unfortunately, it still remains open to establish the security of publicly-verifiable CVQC under iO even forgetting about blindness and compactness.
\end{itemize}

Another natural direction is to expand the class of quantum programs that can be classically obfuscated. While our approach builds on \cite{bartusek2023obfuscation}, an alternative approach of \cite{bartusek2024quantum,huang2025obfuscation} has yielded obfuscation for all \emph{unitaries}, but with a quantum obfuscator. The reason that we build on \cite{bartusek2023obfuscation} as opposed to \cite{bartusek2024quantum,huang2025obfuscation} is that the quantum component of the \cite{bartusek2023obfuscation} obfuscated program contains ``only'' coset states, while the quantum component of the \cite{bartusek2024quantum,huang2025obfuscation} obfuscated programs contain more complicated states, in particular magic states encoded in the so-called coset state authentication code \cite{bartusek2024quantum}. It is a fascinating open question whether this approach can also be de-quantized, which may yield classical obfuscation of unitaries and beyond.

\begin{table}
\begin{tabularx}{\textwidth}{|X||X|X|X|X|X|}
\hline
\textbf{} &  \textbf{Blind} & \textbf{Verification}  & \textbf{Succinct} & \textbf{Interactive} & \textbf{Heuristic} \\
\hline\hline
\textbf{\cite{mah18}} & {Yes*} & {\textcolor{red}{Private}} & \textcolor{red}{No} & {\textcolor{red}{Yes}} & {No} \\
\hline
\textbf{\cite{bartusek2022succinct}} & {Yes*} & {\textcolor{red}{Private}} & {Yes} & {\textcolor{red}{Yes}} & {No} \\
\hline
\textbf{\cite{BM22}} & {\textcolor{red}{No}} & {Public} & {Yes} & {No} & \textcolor{red}{Yes}  \\
\hline

\textbf{This work} & {Yes} & {Public} & Yes & {No} & \textcolor{red}{Yes}  \\
\hline
\end{tabularx}
\caption{A comparison among classical verification of quantum computation protocols for BQP. The asterisk in the Blind column refer to the fact that these protocols aren't necessarily presented as blind, but can be made blind by running under QFHE. This cannot be done for \cite{BM22} without sacrificing public verification. The Heuristic column refers to provable security (i.e.\ reduction to an assumption) vs. heuristic security (concretely, argued in the classical oracle model). Finally, we note that while the protocols listed above are the first to achieve the listed guarantees, there have been several other protocols proposed that improve along other metrics such as the assumption \cite{NZ23,MNZ24,GKNV25} and circuit class \cite{chung2022constant}.}
\label{tab:comparison}
\end{table}

\section{Technical Overview} \label{subsec:techniques}
As mentioned above, the main technical workhorse behind our results can be seen as a \emph{publicly-verifiable} measurement protocol with a classical verifier. To this end, we construct, relative to a classical oracle, \emph{publicly-verifiable} Pauli functional commitments with classical commitment keys. We next recall what PFCs  are, discuss previous work, and then present our new techniques.

\subsection{Publicly-Verifiable PFCs and Previous Work} \label{subsec:techniques_1}
A publicly-verifiable PFC with classical keys (Definition \ref{def:PFC}) is given by algorithms $(\Gen, \allowbreak \mathsf{Com}, \allowbreak\mathsf{OpenZ},\allowbreak \mathsf{OpenX}, \allowbreak\mathsf{DecZ}, \allowbreak\mathsf{DecX})$, where the classical algorithms here are $\Gen$, $\mathsf{DecZ}, \mathsf{DecX}$. Our standard scenario is a classical verifier $\Ver$ that wants to make a Pauli measurement on some unknown $n$-qubit state $\ket{\psi}$ held by a quantum prover $\P$. In this instance, the flow of a PFC goes as follows.
\begin{enumerate}
    \item 
    $\Ver$ samples classical keys $\left( \CK, \dk \right) \gets \Gen\left( 1^\lambda \right)$. We think of $\CK$ as a public commitment key in the form of an ideally-obfuscated classical circuit, and $\dk$ is a secret decoding key. $\Ver$ sends $\CK$ to the prover.

    \item  \label{overview:protocol_commit}
    $\P$ commits to its state by computing $\left( c, \calB, \calU \right) \gets \mathsf{Com}^{\mathsf{CK}}\left( \ket{\psi}_{\calB} \right)$, where $c$ is a classical commitment string, $\calB := \left( \calB_i \right)_{i \in [n]}$ is the original $n$-qubit register that held $\ket{\psi}$ and $\calU := \left( \calU_i \right)_{i \in [n]}$ is an additional quantum register with some $\poly(n)$-qubits, now entangled with $\calB$. $\P$ sends $c$ to the verifier.

    \item 
    $\Ver$ chooses some measurement basis $h \in \{ 0, 1 \}^n$, where $h_{i} = 0$ denotes measuring qubit $i$ of $\ket{\psi}$ in the standard basis and $h_i = 1$ denotes Hadamard basis. $\Ver$ sends $h$ to the prover.

    \item \label{overview:protocol_DEcommit}
    $\P$ decommits to the Pauli observables of its state by applying $u_i \gets \mathsf{OpenZ}\left( \calB_i, \calU_i \right)$ for every $h_i = 0$ and $u_i \gets \mathsf{OpenX}\left( \calB_i, \calU_i \right)$ otherwise. $\P$ sends $u := \left( u_i \right)_{i \in [n]}$ to the verifier.

    \item $\Ver$ can now privately decode the measurement outcomes by applying the corresponding decoding procedures $\mathsf{DecZ}\left( \dk, \cdot, \cdot \right)$ and $\mathsf{DecX}\left( \dk, \cdot, \cdot \right)$. For each $i \in [n]$, the corresponding algorithm takes as input the commitment $c$ and the value $u_i$, and outputs a measurement result $b_{i} \in \{ 0, 1 \}$.
\end{enumerate}
In words, what happens above is that at Step \ref{overview:protocol_commit}, the prover generates an \emph{encoding} of the state $\ket{\psi}_{\calB}$, by embedding it into a larger space $(\calB,\calU)$. Then, at Step \ref{overview:protocol_DEcommit}, the prover generates encodings of standard basis measurements $u_i \gets \mathsf{OpenZ}\left( \calB_i, \calU_i \right)$ or Hadamard basis measurements $u_i \gets \mathsf{OpenX}\left( \calB_i, \calU_i \right)$, for each $i \in [n]$. The logical values of these encoded measurements are later extracted by $\Ver$  using $\dk$.

\paragraph{PFC Security: Private vs Public Verification.}
Intuitively, security of a PFC ensures that the logical standard basis measurements cannot be swapped. A bit more formally, for each part $i \in [n]$ of the committed state on register $\left( \calB_i, \calU_i \right)$ at the end of Step \ref{overview:protocol_commit}, if we project the state onto encodings of standard basis measurements for some fixed value $b \in \{ 0, 1 \}$ (i.e., strings $u_i$ such that $\mathsf{DecZ}\left( \dk, c, u_i \right) = b$), then no adversary can map this state to encodings of standard basis measurements for value $1-b$ (i.e., strings $u_i$ such that $\mathsf{DecZ}\left( \dk, c, u_i \right) = 1-b$). In fact, this is just one way to describe the traditional notion of post-quantum binding for classical messages (equivalent to collapse-binding and sum-binding).

In a privately-verifiable PFC, the verifier's decoding key $\dk$ must remain completely inaccessible to the prover throughout the protocol, and such a PFC was (implicitly) constructed in \cite{mah18}. In a publicly-verifiable scheme, as considered by \cite{bartusek2023obfuscation}, we aim to give the prover access to the verifier's decoding functionalities without compromising the binding security stated above. A natural attempt at such a definition would require the scheme to remain secure against a prover given quantum oracle access to the classical functions $\mathsf{DecZ}\left( \dk, \cdot, \cdot \right)$, $\mathsf{DecX}\left( \dk, \cdot, \cdot \right)$.

It turns out that, for our application (as well as for \cite{bartusek2023obfuscation}), the security of publicly-verifiable PFC need not hold\footnote{In fact, there is good reason to believe that this definition is not satisfiable at all, since a Hadamard basis decoding functionality intuitively gives the committer the ability to perform a projection in the Hadamard basis and thus change their value in the standard basis.} if the prover is given \emph{uninterrupted} access to both $\mathsf{DecZ}\left( \dk, \cdot, \cdot \right)$, $\mathsf{DecX}\left( \dk, \cdot, \cdot \right)$. Instead, we consider the relaxation where the prover $\P$ is given quantum access to $\mathsf{DecZ}\left( \dk, \cdot, \cdot \right)$, $\mathsf{DecX}\left( \dk, \cdot, \cdot \right)$ only up to the collapse of the encoding register $\left( \calB, \calU \right)$. That is, we consider a game where a malicious prover $\P^*$ prepares the commitment $\left( c, \calB, \calU \right) \gets \mathsf{Com}^{\mathsf{CK}}\left( \ket{\psi}_{\calB} \right)$ as in Step \ref{overview:protocol_commit}, with additional access to both $\mathsf{DecZ}\left( \dk, \cdot, \cdot \right)$, $\mathsf{DecX}\left( \dk, \cdot, \cdot \right)$. Then, the PFC challenger "removes" access to the Hadamard decoding functionality $\mathsf{DecX}\left( \dk, \cdot, \cdot \right)$, and collapses registers $\left( \calB_i, \calU_i \right)_{i \in [n]}$ to encodings of classical measurement results $\{ b_i \}_{i \in [n]}$. Then, given access to only $\mathsf{DecZ}\left( \dk, \cdot, \cdot \right)$, the prover $\P^*$ is required to map these encodings to encodings of classical measurement results for the opposing values $\{1-b_i \}_{i \in [n]}$. To re-iterate, the prover has full access to the verification functionalities while preparing its \emph{commtiment}, but \emph{loses} access to the Hadamard decoding functionality while attempting to map between standard basis openings. To satisfy the above notion of binding it is sufficient for our scheme to formally satisfy a property called \emph{single-bit binding with public decodability}, which is made formal in Definition \ref{def:binding}.

\paragraph{Public Verification via Delayed Collapsing.}
The only previous result on publicly-verifiable PFCs is due to the work of Bartusek, Kitagawa, Nishimaki and Yamakawa \cite{bartusek2023obfuscation}, which achieves the above (in a classical oracle model), though where the commitment key $\ket{\CK}$ is \emph{quantum}, and has size proportional to the number of qubits $n$ in the state $\ket{\psi}$ to be committed.

Without going too deep into their construction, we remark that it makes use of a common tool in quantum cryptography: quantum coset states \cite{coladangelo2021hidden}. A quantum coset state is defined by a random subspace $S \subset \bbZ_{2}^{\secp}$ and shift $v \in \bbZ_2^{\secp}$, and is defined as
\[\ket{S + v} := \frac{1}{\sqrt{|S|}} \sum_{x \in S} \ket{x + v} \enspace .\]

As a major part of their security proof, \cite{bartusek2023obfuscation} prove a delicate information-theoretic statement about coset states, which we refer to as the \textbf{delayed collapsing} property in this work. In a nutshell, they consider a classical oracle $\calO$ such that quantum access to $\calO$ enables implementing the projector onto the coset state, but when access to $\calO$ is revoked, the coset state becomes indistinguishable from the coset state measured in the standard basis.

To elaborate, the ability to verify cosets is given by quantum oracle access to the classical membership check for the sets $S + v$ and $S^{\bot}$, where $S^{\bot} \subseteq \bbZ_{2}^{\secp}$ is the dual subspace of $S$ (i.e., vectors that have an inner product of $0$ with all vectors inside $S$). It is a seminal result \cite{aaronson2012quantum} and common knowledge by now that quantum access to $S + v$ and $S^{\bot}$ enable implementation of the projector $\ketbra{S+v}{S+v}$, while still preserving unclonability of the state $\ket{S+v}$. However, the delayed collapsing property required by \cite{bartusek2023obfuscation} is \emph{stronger} than unclonability of $\ket{S+v}$.

This delayed collapsing property is established in \cite[Claim 4.10]{bartusek2023obfuscation},\footnote{The statement in \cite[Claim 4.10]{bartusek2023obfuscation} is slightly different, but implies the statement we discuss here.} by building on the inner-product adversary method of \cite{aaronson2012quantum}. To be precise, the statement is as follows. A quantum adversary, given $\ket{S + v}$ and quantum oracle access to both $S + v$, $S^{\bot}$, outputs some state on registers $\calR_{1}$, $\calR_{2}$.
The first part $\calR_{1}$ holds a superposition of elements in the coset $S + v$. The second part $\calR_{2}$ is possibly entangled with the first part $\calR_{1}$, and holds arbitrary quantum information that was learned from the oracle access to $S + v$ and $S^{\bot}$.
Then, the challenger "takes away" oracle access to the membership check $S^{\bot}$, flips a coin $b \gets \{0,1\}$, and either measures $\calR_{1}$ or not depending on the bit $b$. The task is for the adversary to now guess $b$. Note that this would have been easy given access to $S^{\bot}$. That is, during the first stage, the coset state is publicly-verifiable and non-collapsing, and during the second stage, access to $S^\bot$ is revoked and the state becomes indistinguishable from collapsed.

Circling back to our publicly-verifiable PFCs, our desired security notion (again, formalized in Definition \ref{def:binding}) is reminiscent of the delayed collapsing property of coset states. Observe that the binding of PFCs can be interpreted as "revocable equivocality" as follows: The adversary has quantum access to $\mathsf{DecZ}\left( \dk, \cdot, \cdot \right)$, $\mathsf{DecX}\left( \dk, \cdot, \cdot \right)$ and makes a commitment. Having access to both oracles (and in particular the Hadamard decoding oracle $\mathsf{DecX}\left( \dk, \cdot, \cdot \right)$), it may be possible to equivocate from a state that has a superposition of encodings of measurement result $b \in \{0, 1\}$, to encodings of the opposite measurement result $1 - b$. However, the security definition ensures that once access to $\mathsf{DecX}\left( \dk, \cdot, \cdot \right)$ is removed, equivocation becomes computationally intractable.

Indeed, the above connection between delayed collapsing and public verification of PFCs was observed and formalized in \cite{bartusek2023obfuscation}. However, as mentioned, the previous work deals only with quantum (and long) keys. In our work we seek to construct a scheme with classical (and potentially short) keys, and to this end we develop new techniques.

\subsection{Our Work: One-Shot Signatures with Delayed Collapsing} \label{subsec:techniques_2}
Our work begins with the natural idea to employ \emph{one-shot signatures} (OSS) in order to generate the coset states required by $\ket{\CK}$. OSS is a powerful primitive in quantum cryptography, first defined in \cite{amos2020one} and recently constructed for the first time in \cite{shmueli2025one}. The OSS of \cite{shmueli2025one} is defined with respect to public parameters that enable anyone to sample their own \emph{unclonable} coset states, and to publicly verify these states. That is, the sampled coset states cannot be cloned even by the sampler that produced them. Moreover, the public parameters can be used to sample an unbounded polynomial number of such unclonable states.

\paragraph{From Quantum to Classical Keys (and, from long to short keys).}
In principle, the use of OSS may eliminate two unwanted properties of the \cite{bartusek2023obfuscation} PFC in a single swipe: OSS can enable commitment keys that are both classical and short. However, it turns out that the previous results on OSS are insufficient for our needs. 

The OSS of \cite{shmueli2025one} is given by three classical oracles $\left( P, P^{-1}, D \right)$. In our context here, $P$, $P^{-1}$ can be thought of as the oracles that allow for generation of coset states, and $D$ is the oracle that allows for verification of the quantum tokens. In particular, $D$ can be used to detect collapsing and distinguish between a uniform superposition over the coset and a collapsed classical state. \cite{shmueli2025one} prove that the oracles $\left( P, P^{-1}, D \right)$ allow to generate coset states that are unclonable,\footnote{In fact, an even stronger statement is proven by \cite{shmueli2025one} showing the one cannot produce two classical coset elements derived from the same token.} and are verifiable (in particular, non-collapsing) given the oracle $D$. For our PFCs, we will want the stronger, delayed collapsing property: We will further ask that once we revoke access to $D$, the generated states become indistinguishable from collapsed. More formally, we consider a game in which an adversary, given full access to $\left( P, P^{-1}, D \right)$, first samples a quantum coset state (and potentially some side information). Then, a challenger either measures the state or not in the standard basis, and eliminates access to the $D$ oracle. Then, the adversary is required to tell which of the options occured (this is formalized in Definition \ref{def:oss-collapse-binding}).

Proving the delayed collapsing property for OSS requires a careful reexamination of the \cite{shmueli2025one} proof strategy, which we do in Section \ref{subsec:binding}. Our proof technique rests upon a new and simpler proof of the delayed collapsing property of coset states that implicitly formed the heart of the \cite{bartusek2023obfuscation} proof strategy. In fact, the inner-product adversary method approach of \cite{bartusek2023obfuscation} appears ill-suited for use in the context of OSS,\footnote{This was the approach taken in the original but flawed attempt at proving security of OSS \cite{amos2020one}.} and hence we develop a different approach to establishing delayed collapsing that interfaces better with the \cite{shmueli2025one} proof technique.

\vspace{3.5mm}
\noindent
\textit{Roadmap of our proofs.}
In the remainder of the overview, we explain the high-level structure for our proofs. Recall that our main goal is to construct a scheme that has single-bit binding with public decodability (Definition \ref{def:binding}). We next go over the main steps for achieving this, how this connects to OSS with delayed collapsing, and the proof of the delayed collapsing of the OSS construction of \cite{shmueli2025one}.

\paragraph{Reducing single-bit binding PFCs to collapse binding PFCs.}
The first observation we make is that, in order to construct publicly-verifiable PFC with single-bit binding, it suffices to construct PFC with a property we refer to as collapse binding (Definition \ref{def:collapse-binding}). The scheme satisfies our notion of collapse binding if the following experiment has negligible advantage.
\begin{itemize}
    \item The verifier samples $\left( \CK, \dk \right) \leftarrow \mathsf{Gen}(1^\secp)$.
    \item The adversary, with access to $\CK, \DecZ[\dk], \DecX[\dk]$, outputs a commitment $c$ and a state on registers $\left( \mathcal{B}, \mathcal{U} \right)$.
    \item The verifier samples a bit $b \gets \{0,1\}$ and either measures the registers $\left( \mathcal{B}, \mathcal{U} \right)$ in the standard basis or not depending on the value of $b$.
    \item The adversary, given $\CK, \DecZ[\dk]$, tries to guess the value of $b$.
\end{itemize}

\noindent To show the implication, we pass through a sequence of related notions of commitment binding. There are two main steps. 
\begin{enumerate}
    \item
    First, we show that any commitment scheme satisfying collapse binding (Definition \ref{def:collapse-binding}) also satisfies \emph{single-bit} collapse binding.\footnote{Note the names slightly defer, and "single-bit collapse binding" is not the exact same as our end goal security of "single-bit binding with public decodability."} Roughly, the setting of single-bit collapse binding is similar to that of collapse binding, but rather than measuring the registers $\mathcal{B}, \mathcal{U}$ in the computational basis, only the \emph{first bit} of $(\mathcal{B}, \mathcal{U})$ is measured. 
    
    \item
    Second, it was established in \cite{bartusek2023obfuscation} that single-bit collapse binding implies a notion called "unique message" binding, which further implies a notion called string binding. Finally, string binding implies single-bit binding with public decodability.
\end{enumerate}

\paragraph{Reducing collapse binding PFCs to OSS with delayed collapsing.}
It remains to construct a publicly-verifiable PFC with collapse binding. Our construction is given in Section \ref{sec:construction-functional-commitment} and is based on the OSS construction of \cite{shmueli2025one}, which we denote by $\overline{\OSS}$. In the body of the paper, we show that our construction is indeed a publicly-verifiable PFC with collapse binding, as long as $\overline{\OSS}$ satisfies the property we informally discussed earlier, referred to as delayed collapsing. This game is formalized in \Cref{def:oss-collapse-binding}.

Delayed collapsing is captured by the advantage of query-bounded adversaries $\mathcal{A}=(\mathcal{A}_1,\mathcal{A}_2)$ in following experiment:\footnote{The following elaboration necessitates some familiarity with the the inner workings of $\overline{\OSS}$, which can be found in Section \ref{sec:oss-sz}.}
\begin{enumerate}
    \item Sample $\left( P, P^{-1}, D \right) \leftarrow \overline{\OSS}.\Setup(1^\lambda)$. Let $\{S_y\}_{y \in \{0,1\}^r}$ be the cosets defined by $P, P^{-1}$.
    
    \item Run $\calA_1^{P, P^{-1}, D}(1^\lambda)$ until it outputs a string $y \in \{0,1\}^r$ and a state on registers $\left( \calU, \calR \right)$. Here, $\calU$ is a $k$-qubit register and $\calR$ holds the internal state of $\calA$.
    
    \item Sample $b \leftarrow \{0,1\}$. If $b=0$, do nothing, and if $b=1$ measure $\calU$ in the standard basis.
    
    \item Run $\calA_2^{P, P^{-1}}\left( \calU, \calR \right)$ until it outputs a bit $b'$. The experiment outputs $1$ if $b=b'$ and $0$ otherwise.
\end{enumerate}
Our reduction from collapse binding security of PFC to the delayed collapsing of $\overline{\OSS}$ requires an \emph{additional} use of OSS,\footnote{In \cite{bartusek2023obfuscation}, this role was fulfilled by a signature token scheme with quantum keys.} though for this we only require standard unforgeability of OSS, as already shown by \cite{shmueli2025one}.

\paragraph{Reducing OSS with delayed collapsing to a coset collapsing game.}

This reduction captures the heart of our technical contribution, as it proves the delayed collapsing property of $\overline{\OSS}$. This is done in three main steps: (1) We propose a novel information-theoretic property of coset states, which we formalize as the \emph{coset collapsing game}. (2) We perform a careful analysis of $\overline{\OSS}$, which deviates from the previous security analysis of \cite{shmueli2025one} and ultimately reduces the delayed collapsing property of $\overline{\OSS}$ to the hardness of the coset collapsing game. (3) We prove the information-theoretic hardness of the coset collapsing game.

The coset collapsing game (formalized in Definition \ref{def:coset-collapsing-dual access}) is a game between a challenger and a two-stage adversary $\left( \mathcal{A}_1,\mathcal{A}_2 \right)$. There is a public subspace $S$ known to all parties and the game executes as follows:
\begin{itemize}
    \item
    The challenger samples a random subspace $T \subset S$, and grants $\mathcal{A}_1$  access to $O_{T^{\perp}}$, the membership checking oracle for $T^{\perp}$, which is a random superspace of $S^{\bot}$. 
    
    \item
    $\mathcal{A}_1^{O_{T^{\perp}}}$ generates a state on register $\mathcal{R}$, to which the challenger applies a projective measurement onto a coset of $T$ in $S$. That is, the left-over state will be in \emph{some} superposition over \emph{some} coset parallel to $T$ (i.e., a set of the form $T + x$ for some $x$).
    
    \item
    With equal probability, the challenger either
    \begin{itemize}
        \item Does nothing further to $\mathcal{R}$, or
        \item Measures $\mathcal{R}$ in the standard basis. 
    \end{itemize}

    \item
    Finally $\mathcal{A}_2$ is given $\mathcal{R}$ but no access to $O_{T^{\perp}}$.
\end{itemize}
We say that "cosets of $T$ are collapsing" if no $\left( \mathcal{A}_1,\mathcal{A}_2 \right)$ has non-negligible advantage in guessing whether the final measurement was performed or not. In the proof of Theorem \ref{lem:oss-collapsing}, we show that any adversary for the delayed collapsing of $\overline{\OSS}$ can be turned into an adversary for the coset collapsing game. As mentioned above, this proof follows the security analysis of \cite{shmueli2025one} for some of the steps, but deviates in others. For example, while \cite{shmueli2025one} work towards a reduction to the collision-resistance of some ``coset partition function'', we reduce to the (non-delayed) collapsing property of the same affine partition function.

\snote{I added details to this reduction in the previous version, would it make sense to bring some of it back? Not sure how high level it is supposed to be. Feel free to use/delete all or part of this:}

We will now give a proof sketch of the hybrids used in the reduction from the delayed collapsing property of $\overline{\OSS}$ to the hardness of the coset collapsing game.
\paragraph{Step 1: Bloating the dual.} Like \cite{shmueli2025one}, we first bloat the dual oracle $D$ and move from $(P, P^{-1}, D)$ (see Definition \ref{def:OSS}) to $(P,P^{-1},D')$, where $D'$ checks membership for a random superspaces $T^{\perp}_y$ of $S^{\perp}_y$. A query bounded adversary will not notice the difference, since points in $T^{\perp}_y \setminus S^{\perp}_y$ are hard to find.
Concretely, we bloat the dual oracle as follows \jnote{mismatch between $\vecv$ and $\vecz$ below?}\snote{fixed}
\begin{align*}
D'(y, \vecv) &= \begin{dcases}
              1 &\text{if } \vecv^\top \matA_y^{(1)} = \mathbf{0}^{n-r-s}\\
              0 &\text{otherwise.}
          \end{dcases}
\end{align*}
Here $\matA^{(1)}_y$ refers to only the last $n-r-s$ columns of the matrix $\matA_y$, which defines $S^\bot_y$.
Identically to \cite{shmueli2025one}, the superspaces $T^{\perp}_y$ are chosen such that they are sparse subsets of the entire space, and each $S^{\perp}_y$ is a sparse subset of $T^{\perp}_y$. 


\paragraph{Step 2: Simulating the dual to introduce a collapsing coset partition function.} Like in \cite{shmueli2025one}, we now eliminate the dual oracle altogether by embedding a rerandomization of the cosets defined by a smaller instance of the $(\tilde{P}, \tilde{P}^{-1}, \tilde{D})$ into the cosets defined by the actual oracles $(P, P^{-1}, D)$. This allows us to simulate \textit{without access to $\tilde{D}$} the corresponding random superspace $T_y^\perp$ of the dual $S_y^\perp$. By the previous step, the adversary cannot tell the difference. Once we have eliminated the dual oracles, the adversary effectively has access only to $\tilde{P}, \tilde{P}^{-1}$ of these smaller instances, which we are now free to simulate using a collapsing coset-partition function $Q$.
 
Concretely, by combining Lemmas 31 and 34 of \cite{shmueli2025one}, we simulate the oracles $P, P^{-1}, D'$ \footnote{The following elaboration necessitates some familiarity with the the inner workings of $\overline{\OSS}$, which can be found in Section \ref{sec:oss-sz}.} using a coset partition function $Q:\{0,1\}^{r+s}\rightarrow \{0,1\}^s$ (see Definition \ref{def:cosetpartfunc}) as follows:\\
\begin{enumerate}
\item Sample a random permutation $\Gamma: \{0,1\}^n \rightarrow \{0,1\}^n$ and for every $y$, choose a random full rank matrix $\matC_y \in \mathbb{Z}_2^{k \times n}$, and a random vector $\vecd_y \in \mathbb{Z}_2^k$. 
    \item Simulating $P(x \in \{0,1\}^n)$. 
    \begin{itemize}
        \item $(x_0 \in \{0,1\}^{r+s}, x_1 \in \{0,1\}^{n-r-s})\gets \Gamma(x)$. 
        \item $y \gets Q(x_0)$. 
        \item $\vecu \gets (\matC_y\cdot \Gamma (x)+\vecd_y)$. 
        \item Output $(y,\vecu)$. 
    \end{itemize}
    \item Simulating $P^{-1}(y\in \{0,1\}^r,\vecu\in \mathbb{Z}_2^{k})$. 
    \begin{itemize}
        \item $z \gets \matC_y^{-1}\cdot (\vecu-\vecd_y)$. 
        \item $(z_0 \in \{0,1\}^{r+s}, z_1\in \{0,1\}^{n-r-s}):=z$.
        \item If $Q(z_0)=y$, output $\Gamma^{-1}(z)$. Otherwise, output $\bot$. 
    \end{itemize}
    \item Simulating $D'(y \in \{0,1\}^r,\vecv \in \mathbb{Z}_2^k)$
    \begin{itemize}
        \item $\matA_y^{(1)\bot}:=$ last $n-r-s$ columns of $\matC_y$. 
        \item Output $1$ iff $\vecv^{T}\cdot \matA^{(1)}_y=0^{n-r-s}$. 
    \end{itemize}
\end{enumerate}
\paragraph{Step 3: Utilizing the collapsing of coset partition function.}
In this part of the hybrid argument, we utilize the fact that our coset partition function $Q$ is a collapsing hash function, which essentially means that, for any superposition of inputs (corresponding to a particular output) to $Q$ that the adversary can produce, it is computationally infeasible to tell whether a challenger measured the superposition in the computational basis, or didn't measure the superposition at all.  Now more concretely, we show that, for a fixed $y$,  $\mathcal{A}_1^{Q}$ produces the following state on registers $(\mathcal{U}, \mathcal{R})$,  
\[
    \sum_{\overline{\Gamma(x)}:Q(\overline{\Gamma(x)})=y, \widetilde{\Gamma(x)}
    }\matC_y \begin{bmatrix}
        \overline{\Gamma(x)}\\
        \widetilde{\Gamma(x)}
    \end{bmatrix}+\vecd_y
    \]
    Since $Q$ is collapsing, this state is indistinguishable from the state 
    \[
    \sum_{ \widetilde{\Gamma(x)}
    }\matC_y \begin{bmatrix}
        \overline{\Gamma(x)}\\
        \widetilde{\Gamma(x)}
    \end{bmatrix}+\vecd_y
    \]
    where $\overline{\Gamma(x)}$ is some fixed value such that $Q(\overline{\Gamma(x)})=y.$
    
\paragraph{Step 4: Unsimulating the bloated dual.}
Since $Q$ is a $(r+s,r,s)$ coset partition function, the preimage set $Q^{-1}(y)$ has size $2^s$ and is a coset of a linear space of dimension $s$.  
Using this fact,  crucially, we observe that collapsing the preimages of $Q$ essentially projects the $\mathcal{U}$ register to  a superposition over a coset of colspan($\matA_y^{(1)}$) in colspan($\matA_y$).
Therefore, using knowledge of $\matA_y $, the reduction can discard $Q$, and actively collapse to cosets of  colspan$(\matA_y^{(1)})$.
Thus, through a series of hybrids, we show that the state output by $\mathcal{A}_1$ is a superposition over a coset of a \emph{specific} subspace of colspan($\matA_y$), namely colspan($\matA^{(1)}_y$). This is the coset collapsing game with $T=\text{colspan}(\matA^{(1)}_y)$, and $S=\text{colspan}(\matA_y)$. Now, the final few hybrids of our reduction serve to randomize $T$, at which point the game is exactly reduced to the coset collapsing game.\jnote{changed several $\mathcal{A}_y$ to $\matA_y$ above, which I think was what was meant?}\snote{yes, thanks!}

\paragraph{Information-theoretic hardness of the coset collapsing game.} It remains to show hardness of the coset collapsing game. Intuitively, since $T$ is a random subspace of $S$ chosen by the challenger, it should be hard for any adversary to distinguish whether or not a measurement on a coset of $T$ in $S$ was performed, \emph{after} access to $T^{\perp}$ has been revoked. This is established in three steps.
\begin{itemize}
    \item Un-bloat the dual from $T^{\perp}$ to $S^{\perp}$. It is standard practice that access to $O_{T^{\perp}}$ for a random superspace $T^{\bot}$ can be shown to be indistinguishable from access to $O_{S^{\perp}}$. However, this only holds if $T^\perp$ is random and \emph{independent of the rest of the experiment}. In our case, note that the challenger also depends on $T$ to perform its intermediate coset measurement. Our saving grace is that we only need to argue about indistinguishability from $O_{S^{\perp}}$ \emph{before} the challenger's measurement, since the adversary's access to this oracle is revoked afterwards. In fact, this is the \textbf{only} place in our proof where we make use of the crucial fact that the adversary's access to the dual oracle is revoked after the challenge phase.
    
    \item Since $S$ is public, access to $O_{S^{\perp}}$ gives no additional information to the adversary, so we can ignore it.
    
    \item Re-viewing our game in this modified setting, the adversary's success in this game is exactly the diamond distance between the standard basis measurement channel and the coset-measurement channel over a random choice of subspace. This is something that can be directly analyzed, and we formally show this diamond distance to be exponentially small in Lemma \ref{lem:change-of-basis-coset-collapsing}.
\end{itemize}

\subsection{From PFCs to QFHE and Obfuscation} \label{subsec:techniques_3}

Once we are equipped with publicly-verifiable PFCs with classical keys, our main results follow naturally from the outline proposed in \cite{bartusek2023obfuscation}, with one notable exception. In particular, we use PFCs to upgrade a privately-verifiable QFHE scheme satisfying certain properties to a publicly-verifiable scheme. This compiler is presented in Section \ref{sec:pvQFHE}. Then, given a publicly-verifiable QFHE (in the classical oracle model) we use a straightforward ``verify-then-decrypt'' construction to obtain classical obfuscation for pseudo-deterministic quantum circuits in the classical oracle model (Section \ref{sec:obfuscation}).

The notable departure from \cite{bartusek2023obfuscation} mentioned above is necessitated by our requirement of \emph{succinctness}. While neither the privately-verifiable QFHE scheme we start with, nor the PFC-based compiler are themselves succinct, we show how to restore succinctness with a generic approach based on (post-quantum) succinct ideal obfuscation of \emph{classical} circuits, in the classical oracle model. More concretely, given a Turing machine $P$ (which may take in a long input, but has a short output), a post-quantum succinct ideal obfuscator outputs a circuit $\tilde{P}$ of size bounded by $|P|$, which is the size of the $P$'s \emph{code}, as opposed it its \emph{run-time}. We provide the first construction of this object by combining (classical) FHE and post-quantum non-interactive succinct arguments of knowledge in the quantum random oracle model \cite{chiesa2019snark}. The proof establishing ideal obfuscation of our construction requires a careful application of techniques from \cite{chiesa2019snark}, which we formalize in Section \ref{sec:input-succinct}.\\\snote{added details  on succint ideal obf:} \\
We describe some high level ideas below:

\paragraph{Succint Ideal Obfuscation}
Classically, an input succinct obfuscator can be constructed through fully homomorphic encryption (FHE) and SNARKs (succinct non-interactive arguments of knowledge) as follows: the obfuscated circuit $\tilde{P}$ consists of the FHE encryption of $P$, $\FHE.\Enc(P)$, along with an oracle $O_{\sk}$, which has the FHE encryption secret key $\sk$ hardcoded. $O_{\sk}$ takes as input an evaluated ciphertext and a SNARK of its honest evaluation. It verifies the SNARK and outputs the decryption of the ciphertect (using $\sk$) iff the verification procedure outputs 1.  

In the setting where an adversary only makes classical queries to $O_{\sk}$, prior work (\cite{classicalsuccintobf}) establishes a simulation strategy (and therefore establishes ideal obfuscation). However, in our setting, the adversary makes \emph{superposition queries} to $O_{\sk}$. In order to show a simulation strategy in this setting, we utilize the knowledge extractor for Micali's \jnote{cite} SNARK construction \cite{micalisnark} in the quantum random oracle model (QROM). In more detail, \cite{chiesa2019snark} showed that the SNARK construction of Micali, which maps any probabilistically checkable proof (PCP) into a corresponding non-interactive argument,  is unconditionally secure in the QROM. Along the way, they also prove that the Micali construction inherits proof of knowledge properties from the underlying PCP, by constructing an extractor $\mathcal{E}^*$, which, given a valid proof $\pi$ (one that passes SNARK verification), outputs a valid witness with high probability. 
Our simulation strategy, which we call the $\mathbf{FindWitness}$ procedure, is essentially the witness extractor $\mathcal{E}^*$, with slight changes required to handle superposition query access. In more detail, on input a SNARK proof $\pi$ and some evaluated ciphertext $\ct$, our  simulator runs the $\mathbf{FindWitness}$ procedure \emph{coherently} \jnote{added coherently} iff $\pi$ is a valid proof. Then, by utilizing the guarantees provided by $\mathcal{E}^*$, as well as the compressed oracle method of Zhandry \cite{zhacomporacle},  we show that the simulator can extract the intended (potentially large) input $x$ and query it's oracle, establishing ideal obfuscation security.

\section{Preliminaries}
\subsection{Notation}

\begin{definition}[Pseudo-deterministic quantum circuit]
\label{def: deterministic}
A family of psuedo-deterministic quantum circuits $\{Q_\secp\}_{\secp \in \mathbb{N}}$ is defined as follows. The circuit $Q_\secp$ takes as input a classical string $x \in \{0,1\}^{n(\secp)}$ and outputs a bit $b \gets Q_\secp(x)$. The circuit is pseudo-deterministic if for every sequence of classical inputs $\{x_\secp\}_{\secp \in \mathbb{N}}$, there exists a sequence of outputs $\{b_\secp\}_{\secp \in \mathbb{N}}$ such that \[\Pr[Q_\secp(x_\secp) \to b_\secp] = 1-\negl(\secp).\] We will often leave the dependence on $\secp$ implicit, and just refer to pseudo-deterministic circuits $Q$ with input $x$. In a slight abuse of notation, we will denote by $Q(x)$ the bit $b$ such that $\Pr[Q(x) \to b] = 1-\negl(\secp)$.

Sometimes, we will consider (potentially succinct) \emph{descriptions} of pseudo-deterministic circuits, which we write as $P_Q$, and refer to as (pseudo-deterministic) \emph{programs}. That is, while $Q$ is understood to consist of the list of gates in the circuit, $P_Q$ is understood to be an arbitrary (potentially short) description from which the entire list of gates $Q$ can be derived. Thus, in general it may be the case that $|P_Q| \ll |Q|$.
\end{definition}

\subsection{Fully-homomorphic encryption}

We will use fully-homomorphic encryption for classical circuits, defined as follows.

\begin{definition}[FHE]\label{def:FHE}
    A fully-homomorphic encryption scheme consists of the following algorithms.
    \begin{itemize}
        \item $\Gen(1^\secp,D) \to (\pk,\sk)$: The p.p.t. key generation algorithm takes as input the security parameter $1^\secp$ and a circuit depth $D$, and returns a public key $\pk$ and secret key $\sk$.
        \item $\Enc(\pk,x) \to \ct$: The p.p.t. encryption algorithm takes as input the public key $\pk$ and a plaintext $x$, and outputs a ciphertext $\ct$.
        \item $\Eval(\ct,C) \to \widetilde{\ct}$: The classical deterministic evaluation algorithm takes as input a ciphertext $\ct$ and the description of a classical deterministic computation $C$, and outputs an evaluated ciphertext $\widetilde{\ct}$.
        \item $\Dec(\sk,\ct) \to x$: The classical deterministic decryption algorithm takes as input the secret key $\sk$ and a ciphertext $\ct$ and outputs a plaintext $x$.
    \end{itemize}
    It should satisfy the following properties.
    \begin{itemize} 
        \item \textbf{Correctness.} For any $(\pk, \sk) \in \Gen(1^\secp, D)$, $\ct \in \Enc(\pk, x)$ and circuits $C$ of depth at most $D$
        \begin{align*}
            \P\left[\Dec(\sk, \ct') = C(x) : \begin{array}{c}
                 (\pk, \sk) \leftarrow \Gen(1^\secp)\\
                  \ct \leftarrow \Enc(\pk, x)\\
                  \ct' \leftarrow \Eval(\ct, C)
            \end{array} \right] = 1 - \negl(\secp).
        \end{align*}
        \item \textbf{Privacy.} For any QPT adversary $\calA = \{\calA_\secp\}_{\secp \in \bbN}$, polynomial $D$ and messages $x_0, x_1$,
        \begin{align*}
            \Bigl|
            \Pr\left[ \calA(\pk, \ct) = 1 : (\pk, \sk) \gets 
            \begin{array}{r}
                (\pk, \sk) \gets \Gen(1^\secp, D)\\
                \ct \gets \Enc(\pk, x_0)
            \end{array} \right]
            \\
            - \Pr\left[\calA(\pk, \ct) = 1 : (\pk, \sk) \gets \begin{array}{c}
                (\pk, \sk) \gets \Gen(1^\secp, D)\\
                \ct \gets \Enc(\pk, x_1)
            \end{array} \right]
            \Bigr| = \negl(\secp).
        \end{align*}
        \item \textbf{Compactness.} There exists a fixed polynomial $p$ such that any classical circuit $C$ of depth at most $D$ and it holds that
        \begin{itemize}
            \item $|\pk|, |\ct|, |\ct'| \le p(\secp, D, |x|)$ and
            \item the runtimes $\Gen(1^\secp, D, |x|)$, $\Enc(\pk, x)$, $\Dec(\sk, \ct')$ are bounded by $p(\secp, D, |x|)$.
        \end{itemize}
    \end{itemize}
\end{definition}

\subsection{One-shot signatures (OSS)}\label{sec:oss-sz}

\begin{definition}[One-Shot Signature Scheme]\label{def:OSS} 
A \emph{one-shot signature (OSS)} scheme is a tuple of algorithms $\OSS = (\mathsf{Setup}, \mathsf{Gen}, \mathsf{Sign}, \mathsf{Ver})$  satisfying the following:
\begin{itemize}
  \item $\mathsf{Setup}(1^\lambda) \to \O$: \emph{A classical probabilistic polynomial-time algorithm that given the security parameter $1^\secp$, outputs a classical deterministic circuit $\O$.}

  \item $\Gen^\O(\ell) \to (\vk, \ket{\sk})$: \emph{A quantum polynomial-time algorithm that takes as input a message length $\ell$, has oracle access to $\O$, and samples a classical public key $\vk$ and quantum secret key $\ket{\sk}$.}

  \item $\mathsf{Sign}^\O(\ket{\sk}, m) \to \sigma$: \emph{A quantum polynomial-time algorithm that has oracle access to $\O$, and, given the quantum key $\ket{\sk}$ and message $m \in \{0,1\}^\ell$ produces a classical signature $\sigma$.}

  \item $\Ver^\O(\vk, m, \sigma) \to \{\bot, \top\}$: \emph{A classical deterministic polynomial-time algorithm that has oracle access to $\O$ and verifies a message $m$ and signature $\sigma$ relative to a public key $\vk$.}
\end{itemize}

\textbf{Correctness:} For any $\ell$ and $m \in \{0,1\}^\ell$, 
\[
\Pr_{\substack{
\O \leftarrow \mathsf{Setup}(1^\lambda,1^\ell) \\
(\vk, \ket{\sk}) \leftarrow \mathsf{Gen}^\O \\
\sigma \leftarrow \mathsf{Sign}^\O(\ket{\sk}, m)
}}[
\mathsf{Ver}^\O(\vk, m, \sigma) = \top
] \geq 1 - \mathsf{negl}(\lambda).
\]

\textbf{Security (Strong unforgeability):} For any $\ell$ and QPT algorithm $\mathcal{A}$,
\[
\Pr_{\substack{
\O \leftarrow \mathsf{Setup}(1^\lambda,1^\ell) \\
(\vk, m_0, m_1, \sigma_0, \sigma_1) \leftarrow \mathcal{A}^\O
}}[
\mathsf{Ver}^\O(\vk, m_0, \sigma_0) = \top \land \mathsf{Ver}^\O(\vk, m_1, \sigma_1) = \top
\land \sigma_0 \neq \sigma_1] \leq \mathsf{negl}(\lambda).
\]
\end{definition}
Note that an adversary breaks the strong unforgeability security guarantee even if it produces two different signatures $\sigma_0 \neq \sigma_1$ of the same message $m_0 = m_1$.

We now describe the construction of one-shot signatures in \cite{shmueli2025one}. We will later use this construction in a non-black-box way in \Cref{sec:construction-functional-commitment}.

Parameters: Let $\secp \in \bbN$ be the security parameter. Let $s = 16 \secp$ and let $n, r, k \in \mathbb{N}$ be such that $r = s \cdot (\secp - 1)$, $n = r + \frac{3}{2} \cdot s$ and $k = n$.
\begin{itemize}
  \item $\OSS.\Setup(1^\secp)$
  \begin{itemize}
      \item Sample a random permutation $\Pi: \{0,1\}^n \rightarrow \{0,1\}^n$. Let $H(x)$ be the first $r$ output bits of $\Pi(x)$ and let $J(x)$ denote the remaining $n-r$ bits. Interpret $J(x) \in \mathbb{Z}_2^{n-r}$.
      \begin{align*}
          \Pi(x) = H(x) \| J(x).
      \end{align*}
      \item Sample a random function $F: \{0,1\}^r \rightarrow \{0,1\}^{k \cdot (n-r+1)}$. 
      Interpret $F(y) = (\matA_y, \vecb_y)$ where $\matA_y \in \mathbb{Z}_2^{k \times (n-r)}$ is a full-rank matrix and $\vecb \in \mathbb{Z}_2^{n}$.

      Let $S_y = \{\matA_y \vecx + \vecb_y \mid \vecx \in \bbZ_2^{n-r} \}$ be the coset defined by $(\matA_y, \vecb_y)$. Let $S_{y,b} := \{\vecu \in S_y \mid u_1 = b\}$ be the affine subspace of $S_y$ whose first coordinate entry is $b$.
      \item Define the functions $P, P^{-1}, D$ as follows
      \begin{align*}
          P(x) &= (y, \matA_y J(x) + \vecb_y) \in \{0,1\}^n \rightarrow \{0,1\}^r \times \bbZ_2^{k} \text{ where } y = H(x)\\
          P^{-1}(y, \vecu) &= \begin{dcases}
              \Pi^{-1} (y, \vecz) &\text{if }\exists \vecz \in \bbZ_2^{n-r} \text{ such that } \matA_y \vecz + \vecb_y = \vecu\\
              \bot &\text{otherwise.}
          \end{dcases}\\
          D(y, \vecv) &= \begin{dcases}
              1 &\text{if } \vecv^\top \matA_y = \mathbf{0}^{n-r}\\
              0 &\text{otherwise.}
          \end{dcases}
      \end{align*}
      Let $S_y^\bot := \{\vecv \mid D(y, \vecv) = 1\}$ be the subspace dual to $S_y$.
      \item Output $\O \coloneqq (P, P^{-1}, D)$.
  \end{itemize}
  \item $\OSS.\Gen^\O(1^\secp)$
  \begin{itemize}
      \item Parse $\O = (P, P^{-1}, D)$.\footnote{In the construction of Pauli functional commitments in \Cref{sec:construction-functional-commitment}, it will be relevant that $\OSS.\Gen$ does not actually need access to the $D$ oracle, the oracles $P$ and $P^{-1}$ suffice.}
      \item Prepare a uniform superposition over all strings $x \in \{0,1\}^n$ on register $\calX$, and using the $P$ oracle, compute $P(x):= (y, \vecu)$ onto registers $\calY, \calU$. Measure register register $\calY$ to get $y$.
      \item Use the $P^{-1}$ oracle on registers $(\calY, \calU)$ to uncompute the $\calX$ register.
      \item Output $\vk := y$ and the state on $\calU$ as $\ket{\sk}$.
  \end{itemize}
  \item $\OSS.\Sign^\O(\ket{\sk}, m \in \{0,1\})$
  \begin{itemize}
        \item Parse $\O = (P, P^{-1}, D)$.
      \item Let $\calU$ be the register that $\ket{\sk}$ lives on. Measure the first qubit $\widehat{m}$ of the $\calU$ register in the standard basis. If the result is $\widehat{m} = m$, then measure the entire $\calU$ register in the standard basis to get $\vecu$ and output $\vecu$.
      \item If $\widehat{m} \neq m$, apply the rotation $H^{\otimes n} \circ \mathsf{Phase}^{D} \circ H^{\otimes n}$ to the $\calU$ register, where $\mathsf{Phase}^D$ is the map $\ket{\vecv} \mapsto (-1)^{D(\vecv)} \ket{\vecv}$. Measure the register $\calU$ to get $\vecu$. Abort if the first bit of $\vecu$ is $0$, and otherwise output $\vecu$.
  \end{itemize}
  \item $\OSS.\Ver^\O(\vk, m, \sigma)$
  \begin{itemize}
      \item Parse $\vk = y$ and $\sigma = \vecu$.
      \item If the first bit of $\vecu$ is not $m$ and reject.
      \item Compute $P^{-1}(y, \vecu)$ and reject if the output is $\bot$. Accept otherwise.
  \end{itemize}
\end{itemize}

\subsection{Measure and re-program}

\begin{theorem}[Measure and re-program \cite{DFMS19,DFM20}\footnote{This theorem was stated more generally in \cite{DFMS19,DFM20} to consider the drop in expectation for each specific $a^* \in A$, and also to consider a more general class of quantum predicates. }]\label{thm:measure-and-reprogram}
Let $A,B$ be finite non-empty sets, and let $q \in \bbN$. Let $\calA$ be an oracle-aided quantum circuit that makes $q$ queries to a uniformly random function $H: A \to B$ and then outputs classical strings $(a,z)$ where $a \in A$. There exists a two-stage quantum circuit $\Sim[\calA]$ such that for any predicate $V$, it holds that 
\begin{align*}\Pr\left[V(a,b,z) = 1: \begin{array}{r} (a,\state) \gets \Sim[\calA] \\ b \gets B \\ z \gets \Sim[\calA](b,\state) \end{array}\right] \geq \frac{\Pr\left[V(a,H(a),z) = 1 : (a,z) \gets \calA^H\right]}{(2q+1)^2}.\end{align*}

Moreover, $\Sim[\calA]$ operates as follows.

\begin{itemize}
    \item Sample $H: A \to B$ as a $2q$-wise independent function and $(i,d) \gets (\{0,\dots,q-1\} \times \{0,1\}) \cup \{(q,0)\}$.
    \item Run $\calA$ until it has made $i$ oracle queries, answering each query using $H$. 
    \item When $\calA$ is about to make its $(i+1)$'th oracle query, measure its query registers in the standard basis to obtain $a$. In the special case that $(i,d) = (q,0)$, the simulator measures (part of) the final output register of $\calA$ to obtain $a$.
    \item The simulator receives $b \gets B$.
    \item If $d = 0$, answer $\calA$'s $(i+1)$'th query using $H$, and if $d=1$, answer $\calA$'s $(i+1)$'th query using $H[a \to b]$, which is the function $H$ except that $H(a)$ is re-programmed to $b$.
    \item Run $\calA$ until it has made all $q$ oracle queries. For queries $i+2$ through $q$, answer using $H[a \to b]$.
    \item Measure $\calA$'s output $z$.
\end{itemize}

Note that the running time of $\Sim[\calA]$ is at most $\poly(q,\log|A|,\log|B|)$ times the running time of $\calA$.

\end{theorem}

\subsection{The Compressed Oracle Technique 
\label{sec:comporacle}}
The compressed random oracle \cite{zhandry2019record} acts on query and database registers $(\calQ, \calD)$ as follows:
\def\CStO{\mathsf{CStO}}
\def\Decomp{\mathsf{Decomp}}
\begin{align*}
    \CStO = \Decomp \circ \CStO' \circ \Decomp \circ \mathsf{Increase},
\end{align*}
where $\mathsf{Increase} \ket{x, y} \otimes \ket{D} = \ket{x, y} \otimes \ket{D} \ket{(\bot, 0^n)}$ essentially adds another point in the database, and where $\Decomp$ and $\CStO'$ are defined as follows:
\begin{itemize}
    \item $\Decomp$ is defined by $\ket{x, u}_\calQ \otimes \ket{D}_\calD \mapsto \ket{x, u}_{\calQ} \otimes \Decomp_x\ket{D}_{\calD}$, where $\Decomp_x$ is defined by its action on basis states as follows. Let $D$ be such that $D(x) = \bot$:
    \begin{align*}
        \ket{D} &\mapsto \frac{1}{\sqrt{|\calY|}} \sum_{y \in \calY} \ket{D \cup \{(x,y)\}}\\
        \frac{1}{\sqrt{|\calY|}} \sum_{y \in \calY} \ket{D \cup \{(x,y)\}} &\mapsto \ket{D}\\
        \frac{1}{\sqrt{|\calY|}} \sum_{y \in \calY} (-1)^{y \cdot u} \ket{D \cup \{(x,y)\}} &\mapsto \frac{1}{\sqrt{|\calY|}} \sum_{y \in \calY} (-1)^{y \cdot u} \ket{D \cup \{(x,y)\}} \quad \text{ for } u \neq 0
    \end{align*}
    In words, $\Decomp_x$ swaps the states $\ket{D}$ and $\frac{1}{\sqrt{|\calY|}} \sum_{y \in \calY} \ket{D \cup \{(x,y)\}}$, and acts as the identity on the space orthogonal to these two states.
    \item $\CStO'$ maps $\ket{x, u}_\calQ \otimes \ket{D}_\calD \mapsto \ket{x, u \oplus D(x)} \otimes \ket{D}$.
\end{itemize}

\subsection{Useful Lemmas and Definitions}
\begin{definition}[Coset Partition Functions]\label{def:cosetpartfunc}
\cite{shmueli2025one}
For $n, \ell \in \mathbb{N}$, such that $\ell\leq n$, we say a function $Q:\{0,1\}^n\rightarrow \{0,1\}^m$ is a $(n,m, \ell)$ coset partition function if, for each $y$ in the image of $Q$, the preimage set $Q^{-1}(y)$ has size $2^{\ell}$ and is a coset of a linear space of dimension $\ell$. Different preimage sets are allowed to be cosets of different linear spaces. 
\end{definition}
\begin{lemma}[Lemma 3.6 of \cite{bartusek2023obfuscation}]\label{lem:evasive-oracles}
    For each $\secp \in \mathbb{N}$, let $\mathcal{K}_\secp$ be a set of keys, and let $\{\ket{\psi_k}, O_k^0, O_k^1, B_k\}_{k \in \mathcal{K}_\secp}$ be a set of state $\ket{\psi_k}$, classical functions $O_k^0, O_k^1$, and sets of inputs $B_k$. Suppose that the following properties holds.
    \begin{enumerate}
        \item The oracles $O_k^0, O_k^1$ are identical on inputs outside of $S_k$.
        \item For any oracle-aided unitary $U$, with $q = q(\secp)$ queries, there is some $\epsilon = \epsilon(\secp)$ such that\footnote{$\Pi[B_k]$ is the projector onto subspace spanned by $B_k$. Informally what this means is that the probability of that $U$ given query access to ${\cal O}^0_k$ outputs $b\in B_k$ is at most $\epsilon$.} $$\mathbb{E}_{k \gets \mathcal{K}} \left[ \| \Pi[B_k] U^{O_k^0} \ket{\psi_k}\|^2\right] \le \epsilon$$
    \end{enumerate}
    Then, for any oracle-aided unitary $U$ with $q = q(\secp)$ queries and for every distinguisher $D$,
    \begin{align*}
        \left|\Pr_{k \gets \mathcal{K}} \left[D\left(k, U^{O_k^0} \ket{\psi_k}\right) = 0 \right] - \Pr_{k \gets \mathcal{K}} \left[D\left(k, U^{O_k^1} \ket{\psi_k}\right) = 0 \right]\right| \le 4 q \sqrt{\epsilon}.
    \end{align*}
\end{lemma}

\paragraph{Small-range distribution.} \cite{zhandry2021construct} 
Given a distribution $D$ on $\mathcal{Y}$, define $\mathsf{SR}^{D}_{r}(\mathcal{X})$ as the following distribution
on functions from $\mathcal{X}$ to $\mathcal{Y}$:

\begin{itemize}
    \item For each $i \in [r]$, choose a random value $y_i \in \mathcal{Y}$ according to the distribution $D$.
    \item For each $x \in \mathcal{X}$, pick a random $i \in [r]$ and set $O(x) = y_i$.
\end{itemize}




\begin{lemma}
    [Corollary~7.5 of \cite{zhandry2021construct}]\label{lem:zhandry-small-range}
The output distributions of a quantum algorithm making $q$ quantum queries to an oracle either drawn from
$\mathsf{SR}^{D}_{r}(\mathcal{X})$ or $D^{\mathcal{X}}$ are $\ell(q)/r$-close, where $\ell(q) = \pi^{2} (2q)^{3}/3 < 27q^{3}$.
\end{lemma}

\begin{lemma}[\cite{zhandry2013note}, Corollary~2 of \cite{zhandry2019record}]\label{lem:collision-resistance-RO}
    After making $q$ quantum queries to a random oracle, the probability of finding a collision is at most $O(q^3/2^n)$.
\end{lemma}

\section{Pauli Functional Commitments with Classical Keys}\label{sec:pfc}

\subsection{Definition}
We recall the definition of Pauli functional commitments from \cite{bartusek2023obfuscation}. The difference here is that, while \cite{bartusek2023obfuscation} allow the commitment key to be a quantum state, we require the commitment key to be classical.

\begin{definition}[Syntax]\label{def:PFC} A Pauli functional commitment consists of six algorithms $(\mathsf{Gen},\mathsf{Com}, \mathsf{OpenZ}, \mathsf{OpenX},$ $\mathsf{DecZ}, \mathsf{DecX})$ with the following syntax.

\begin{itemize}
    \item \textbf{Key Generation:} $\mathsf{Gen}(1^\lambda) \rightarrow (\CK, \dk)$ is a p.p.t.~algorithm that takes as input the security parameter $1^\lambda$ and outputs a classical deterministic circuit $\CK$ and a classical decoding key $\dk$.

    \item \textbf{Commitment:} $\mathsf{Com}^{\mathsf{CK}}(b) \rightarrow (\calU, c)$ is a QPT algorithm that takes as input a bit $b$ and has oracle access to $\mathsf{CK}$. It outputs registers $(\calU, \calC)$ and then measures $\calC$ in the standard basis to obtain a classical string $c \in \{0,1\}^*$ and a left-over state on register $\calU$. We then write
    \[
    \mathsf{Com}^{\mathsf{CK}} := \ketbra{0}{0} \otimes \Com^\CK(0) + \ketbra{1}{1} \otimes \Com^\CK(1)\]    
    to refer to the map that applies $\mathsf{Com}^{\mathsf{CK}}(\cdot)$ classically controlled on a single-qubit register $\mathcal{B}$ to produce a state on registers $(\mathcal{B}, \calU, \calC)$, and then measures $\calC$ in the standard basis to obtain a classical string $c$ along with a left-over quantum state on registers $(\mathcal{B}, \calU)$.

    \item \textbf{Opening in $\Z, \X$ bases:} q.p.t. measurements on registers $(\mathcal{B}, \calU)$ that outputs a classical string $u$:
    \begin{align*}
        \mathsf{OpenZ}(\calB, \calU) \rightarrow u \quad \text{ and } \quad
        \mathsf{OpenX}(\calB, \calU) \rightarrow u.
    \end{align*}

    \item \textbf{Decoding in $\Z, \X$ bases:} classical deterministic polynomial-time algorithm that takes as input the decoding key $\dk$, a string $c$, and a string $u$, and outputs either a bit $b$ or a $\bot$ symbol:
    \begin{align*}
        \mathsf{DecZ}(\dk,c, u) \rightarrow \{0,1\} \cup \{\bot\} \quad \text{and} \quad \mathsf{DecX}(\dk,c, u) \rightarrow \{0,1\} \cup \{\bot\}.
    \end{align*}
\end{itemize}
\end{definition}

\begin{definition}[Correctness]\label{def:pfc-correctness}
A Pauli functional commitment scheme $(\mathsf{Gen}, \mathsf{Com}, \mathsf{OpenZ}, \mathsf{OpenX}, \mathsf{DecZ}, \mathsf{DecX})$ is correct if for any single-qubit (potentially mixed) state on register $\mathcal{B}$, it holds that
\[
\mathsf{TV}\left(\mathsf{Z}(\mathcal{B}), \mathsf{PFCZ}(1^\lambda, \mathcal{B})\right) = \mathsf{negl}(\lambda), \quad \text{and} \quad
\mathsf{TV}\left(\mathsf{X}(\mathcal{B}), \mathsf{PFCX}(1^\lambda, \mathcal{B})\right) = \mathsf{negl}(\lambda),
\]
\textit{where the distributions are defined as follows.}

\begin{itemize}[leftmargin=2em]
    \item $\mathsf{Z}(\mathcal{B})$ measures $\mathcal{B}$ in the standard basis.
    \item $\mathsf{X}(\mathcal{B})$ measures $\mathcal{B}$ in the Hadamard basis.
    \item $\mathsf{PFCZ}(1^\lambda, \mathcal{B})$ samples $(\CK, \dk) \leftarrow \mathsf{Gen}(1^\lambda)$, $(\mathcal{B}, \calU, c) \leftarrow \mathsf{Com}^{\CK}(\mathcal{B})$, $u \leftarrow \mathsf{OpenZ}(\mathcal{B}, \calU)$, and outputs $\mathsf{DecZ}(\dk,c, u)$.
    \item $\mathsf{PFCX}(1^\lambda, \mathcal{B})$ samples $(\CK, \dk) \leftarrow \mathsf{Gen}(1^\lambda)$, $(\mathcal{B}, \calU, c) \leftarrow \mathsf{Com}^{\mathsf{CK}}(\mathcal{B})$, $u \leftarrow \mathsf{OpenX}(\mathcal{B}, \calU)$, and outputs $\mathsf{DecX}(\dk,c, u)$.
\end{itemize}
\end{definition}

\begin{definition}[Single-bit Binding with public decodability]\label{def:binding}
A Pauli functional commitment $(\mathsf{Gen},\allowbreak \mathsf{Com}, \mathsf{OpenZ}, \allowbreak\mathsf{OpenX}, \allowbreak \mathsf{DecZ}, \allowbreak \mathsf{DecX})$ satisfies single-bit binding with public decodability if the following holds. Given $\dk$, $c$, and $b \in \{0,1\}$, define the following projective measurement
\begin{align*}
\Pi_{\dk,c, b} := \sum_{d : \mathsf{DecZ}(\dk,c, d) = b} |d\rangle \langle d|.
\end{align*}

Consider any adversary $\mathcal{A} = \{(\mathcal{A}_1, \mathcal{A}_2)_\lambda\}_{\lambda \in \mathbb{N}}$, where each $\calA_1$ is an oracle-aided quantum operation, each $\calA_2$ is an oracle-aided unitary, and each $(\mathcal{A}_1, \mathcal{A}_2)_\lambda$ make at most $\poly(\lambda)$ oracle queries. Then for any $b \in \{0,1\}$,
\begin{align*}
\mathbb{E} \left[
\left\| 
\Pi_{\dk,c, 1-b} \, \calA_2^{\mathsf{CK}, \DecZ[\dk]} \, \Pi_{\dk,c, b} \, |\psi\rangle 
\right\| : (|\psi\rangle, c) \leftarrow \calA_1^{\mathsf{CK}, \mathsf{DecZ}[\dk], \mathsf{DecX}[\dk]}(1^\lambda)
\right]
= \mathsf{negl}(\lambda),
\end{align*}
where the expectation is taken over $(\CK, \dk) \leftarrow \mathsf{Gen}(1^\lambda)$. Here, $\DecZ[\dk]$ is the oracle implementing the classical functionality $\DecZ(\dk,\cdot,\cdot)$, and $\DecX[\dk]$ is the oracle implementing the classical functionality $\DecX(\dk,\cdot,\cdot)$.
\end{definition}

Our application to publicly-verifiable QFHE will in fact rely on the following notion of string binding, which was shown in \cite{bartusek2023obfuscation} to be implied by single-bit binding.

\begin{definition}[String binding with public decodability]\label{def:string-binding} A Pauli functional commitment $(\Gen,\allowbreak\Com,\allowbreak\OpenZ,\allowbreak\OpenX,\allowbreak\DecZ,\allowbreak\DecX)$ satisfies \emph{string binding with public decodability} if the following holds for any polynomial $m = m(\secp)$ and two disjoint sets $W_0,W_1 \subset \{0,1\}^m$ of $m$-bit strings. Given a set of $m$ decoding keys $\bdk = (\dk_1,\dots,\dk_m)$, $m$ strings $\bc = (c_1,\dots,c_m)$, and $b \in \{0,1\}$, define \[\Pi_{\bdk,\bc,W_b} \coloneqq \sum_{w \in W_b} \left(\bigotimes_{i \in [m]}\Pi_{\dk_i,c_i,w_i} \right).\] 
 
Consider any adversary $\{(\calA_1, \calA_2)_\secp\}_{\secp \in \bbN}$, where each $\calA_1$ is an oracle-aided quantum operation, each $\calA_2$ is an oracle-aided unitary, and each $(\calA_1, \calA_2)_\secp$ make at most $\poly(\secp)$ oracle queries.  Then,

\begin{align*}\E\left[\bigg\|\Pi_{\bdk,\bc,W_1}\calA_2^{\bCK,\DecZ[\bdk]}\Pi_{\bdk,\bc,W_0}\ket{\psi}\bigg\|: (\ket{\psi},\bc) \gets \calA_1^{\bCK,\DecZ[\bdk],\DecX[\bdk]}(1^\secp)\right] = \negl(\secp), \end{align*}
where the expectation is over $\{\CK_i,\dk_i \gets \Gen(1^\secp)\}_{i \in [m]}$. Here, $\bCK$ is the collection of oracles $\CK_1,\dots,\CK_m$, $\DecZ[\bdk]$ is the collection of oracles $\DecZ[\dk_1],\dots,\DecZ[\dk_m]$, and $\DecX[\bdk]$ is the collection of oracles $\DecX[\dk_1],\dots,\DecX[\dk_m]$. 
\end{definition}

\begin{lemma}[Lemma 4.5 of \cite{bartusek2023obfuscation}]
    Any Pauli functional commitment that satisfies single-bit binding with public decodability (\Cref{def:binding}) also satisfies string binding with public decodability (\Cref{def:string-binding}).
\end{lemma}

Since string-binding when $m = 1$ is exactly the definition of single-bit binding, we have the following corollary.
\begin{corollary}\label{cor:string-single-equiv}
    The single-bit binding definition (\Cref{def:binding}) and string binding definition (\Cref{def:string-binding}) are equivalent.
\end{corollary}

\subsection{Construction}\label{sec:construction-functional-commitment}
In this section, we will present our construction of a Pauli Functional Commitment scheme. Let $n = n(\lambda) \ge \lambda$. Our construction makes use of the following ingredients:
\begin{enumerate}
    \item A general one-shot signature scheme $\OSS= (\OSS.\Setup, \OSS.\Gen, \OSS.\Sign, \OSS.\Ver)$.
    \item The particular construction of one-shot signatures from \cite{shmueli2025one}, which we will denote as $\overline{\OSS}$.
    \item PRF $F$.
\end{enumerate}

\noindent
Now, we are ready to present the construction.
\begin{itemize}
  \item $\mathsf{Gen}(1^\lambda;R)$: 
  \begin{itemize}
        \item Parse randomness $R = (R_1, R_2)$.
        \item Run $\O \leftarrow \OSS.\Setup(1^\secp;R_1)$.
        \item Sample a random PRF key $k$ using randomness $R_2$.
        \item Define oracle $\CK_{P}(\vk, x)$ as follows:
            \begin{itemize}
            \item Run $\overline{\O} \leftarrow \overline{\OSS}.\Setup(1^\secp; F_k(\vk))$ and parse $\overline{\O} = (P, P^{-1}, D)$.
            \item Output $P(x)$.
            \end{itemize}
        \item Define oracle $\CK_{P^{-1}}(\vk, \overline{\vk}, u)$ as follows:
            \begin{itemize}
            \item Run $\overline{\O} \leftarrow \overline{\OSS}.\Setup(1^\secp; F_k(\vk))$ and parse $\overline{\O} = (P, P^{-1}, D)$.
            \item Output $P^{-1}(\overline{\vk}, u)$.
            \end{itemize}
        \item Define the oracle $\CK_D(\vk, \sigma, \overline{\vk}, u)$ as follows:
            \begin{itemize}
            \item Run $\overline{\O} \leftarrow \overline{\OSS}.\Setup(1^\secp; F_k(\vk))$ and parse $\overline{\O} = (P, P^{-1}, D)$.
            \item If $\OSS.\Ver^\O(\vk, 0, \sigma) = \bot$, output $\bot$ and abort.
            \item Otherwise, output $0$ if $D(\overline{\vk}, u) = 1$ and $1$ if $D(\overline{\vk}, u) = 0$.
            \end{itemize}
        \item Output $\CK = (\O, \CK_{P}, \CK_{P^{-1}}, \CK_D)$ and $\dk = (R_1, k)$.
  \end{itemize}

  \item $\mathsf{Com}^{\mathsf{CK}}(b)$:
  \begin{itemize}
    \item Parse $\CK = (\O, \CK_{P}, \CK_{P^{-1}}, \CK_{D})$.
    \item Run $(\vk, \ket{\sk}) \leftarrow \OSS.\Gen^{\O}(1^\secp)$. Coherently apply $\OSS.\Sign^\O(\ket{\sk}, 0)$ to get a superposition over signatures $\sigma$ on bit $0$ on register $\calS$. This signature register $\calS$ will allow us to (coherently) access the oracle $\CK_D$.
    \item Run $(\overline{\vk}, \ket{\overline{\sk}}) \leftarrow \overline{\OSS}.\Gen^{\CK_P(\vk, \cdot ),\CK_{P^{-1}}(\vk, \cdot, \cdot)}(1^\secp)$,\footnote{Note that $\overline{\OSS}.\Gen$ does not need the entire $\O$ oracle in the OSS construction of SZ25, the oracles $P$ and $P^{-1}$ suffice.} where $\overline{\vk} = y$ and $\ket{\overline{\sk}} = \ket{S_y}$ on register $\calS'$ for some $y$ as determined by $P$. Measure the first qubit of $\ket{S_y}$ in the standard basis to get result $b'$. If $b = b'$ then the state has collapsed to $\ket{S_{y, b}}$ and we continue.

    Otherwise, perform a rotation from $\ket{S_{y, b'}}$ to $\ket{S_{y, b}}$ by applying the operation $$(\H^{\otimes n} \otimes \mathsf{I}) \circ \mathsf{Phase}^{\CK_D(\vk, \cdot, \cdot)} \circ (\H^{\otimes n} \otimes \mathsf{I})$$ to register $(\calS', \calS)$, where $\mathsf{Phase}^{\CK_D(\vk, \cdot, \cdot)}$ is the map
    \begin{align*}
        \ket{u}_{\calS'} \ket{\sigma}_{\calS} \mapsto (-1)^{\CK_D(\vk, \sigma, u)} \ket{u}_{\calS'} \ket{\sigma}_{\calS}.
    \end{align*}
    \item Reverse the $\OSS.\Sign^{\O}(\cdot, 0)$ operation on $\calS$ to recover $\ket{\sk}$.
    \item Sample $\sigma \leftarrow \OSS.\Sign^{\O}(\ket{\sk}, 1)$, and output $c= (\vk, \sigma, \overline{\vk})$ along with the final state on register $\mathcal{U} := \calS'$.
  \end{itemize}

  \item $\mathsf{OpenZ}(\mathcal{B}, \mathcal{U})$: Measure all registers in the standard basis.

  \item $\mathsf{OpenX}(\mathcal{B}, \mathcal{U})$: Measure all registers in the Hadamard basis.

  \item $\mathsf{DecZ}(\dk, c, d)$:
  \begin{itemize}
    \item Parse $\dk = (R, k)$, $c = (\vk, \sigma, \overline{\vk})$ and $d = (b, u)$, where $\vk = y$, $b \in \{0,1\}$ and $u \in \{0,1\}^n$.
    \item Run $\O \leftarrow \OSS.\Setup(1^\secp; R)$.
    \item Run $\overline{\O} \leftarrow \overline{\OSS}.\Setup(1^\secp; F_k(\vk))$ and parse $\overline{\O} = (P, P^{-1}, D)$.
    \item Check that $\OSS.\Ver^{\O}(\vk, 1, \sigma) = \top$. Otherwise output $\bot$.
    \item If $u \in S_{y, b}$, output $b$, and otherwise output $\bot$. This can be checked given oracle access to $\overline{\O}$, by checking that $P^{-1}(y, u) \neq \bot$.
  \end{itemize}

  \item $\mathsf{DecX}(\dk, c, d)$:
  \begin{itemize}
    \item Parse $\dk = (R, k)$, $c = (\vk, \sigma, \overline{\vk})$ and $d = (b', u)$, where $\overline{\vk} = y$, $b' \in \{0,1\}$ and $u \in \{0,1\}^n$.
    \item Run $\O \leftarrow \OSS.\Setup(1^\secp; R)$.
    \item Run $\overline{\O} \leftarrow \overline{\OSS}.\Setup(1^\secp; F_k(\vk))$ and parse $\overline{\O} = (P, P^{-1}, D)$.
    \item Check that $\OSS.\Ver^\O(vk, 1, \sigma) = \top$, and if not output $\bot$.
    \item Compute the bit $r$ as follows and abort if $r = \bot$
    \begin{align*}
        r:= \begin{dcases}
            0 &\text{if } u \in S_{y}^\perp\\
            1 &\text{if } u \oplus (1, 0, \ldots, 0) \in S_y^\perp\\
            \bot &\text{otherwise.}
        \end{dcases}
    \end{align*}
    This sets $r = 0$ if $u \in S_y^\perp$ and $r = 1$ if $u \in (S_{y, 0})^\perp \setminus S_y^\perp$. Output $b' \oplus r$. This can be computed using $\overline{\O}$.
    \end{itemize}
\end{itemize}

\subsection{Correctness}
\begin{theorem}\label{thm:pfc-correctness}
    The Pauli Functional commitment scheme described in \Cref{sec:construction-functional-commitment} satisfies correctness as in \Cref{def:pfc-correctness}.
\end{theorem}
The proof of this statement is almost identical to that of Theorem 4.6 of \cite{bartusek2023obfuscation}, except that we are using one-shot signatures instead of signature tokens. We defer the proof to \Cref{sec:appendix-proofs-pfc}.

\subsection{Binding}\label{subsec:binding}
\begin{theorem}
    The construction in \Cref{sec:construction-functional-commitment} satisfies single-bit binding with public decodability as in \Cref{def:binding}.
\end{theorem}

\begin{proof}
It will be convenient for us to work with a different but sufficient (in fact, equivalent) definition of binding. By \Cref{lem:collapsing-single-bit}, it suffices to show that the scheme satisfies the following definition of binding.

\begin{definition}[Collapse-binding]\label{def:collapse-binding}
For any adversary $\mathcal{A} := \{(\calA_1, \calA_2)_\secp\}_{\secp \in \mathbb{N}}$, where each of $\calA_1$ and $\calA_2$ are oracle-aided quantum operations that make at most $\poly(\secp)$ oracle queries, define the collapse-binding experiment $\mathsf{CollapseBindingExpt}^{\mathcal{A}}(1^\secp)$ as follows.
\begin{itemize}
  \item Sample $(\CK, \dk) \leftarrow \mathsf{Gen}(1^\secp)$.

  \item Run $\calA_1^{\mathsf{CK}, \mathsf{DecZ}[\dk], \mathsf{DecX}[\dk]}(1^\lambda)$ until it outputs a commitment $c$ and a state on registers $(\mathcal{B}, \mathcal{U}, \mathcal{R})$. Here $\calB$ is a single-qubit register, $\calU$ is the opening register and $\calR$ holds the internal state of $\calA$.

  \item Sample $b \leftarrow \{0,1\}$. If $b = 0$, do nothing, and if $b=1$ measure $(\mathcal{B}, \mathcal{U})$ in the computational basis.

  \item Run $\calA_2^{\mathsf{CK}, \mathsf{DecZ}[\dk]}(\mathcal{B}, \mathcal{U}, \mathcal{R})$ until it outputs a bit $b'$. 
  
  \item The experiment outputs 1 if $b = b'$.
\end{itemize}

We say that adversary $\mathcal{A}$ is \emph{valid} if the state on $(\mathcal{B}, \mathcal{U})$ output by $\calA_1$ is in the image of $|0\rangle\langle0| \otimes \Pi_{\dk, c,0} + |1\rangle\langle1| \otimes\Pi_{\dk, c,1}$. Then, we say that a Pauli functional commitment $(\mathsf{Gen}, \mathsf{Com}, \mathsf{OpenZ}, \mathsf{OpenX}, \mathsf{DecZ}, \mathsf{DecX})$ satisfies publicly-decodable collapse-binding if it holds that for all valid adversaries $\mathcal{A}$,
\[
\left| \Pr\left[ \mathsf{CollapseBindingExpt}^{\mathcal{A}}(1^\lambda) = 1 \right] - \frac{1}{2} \right| = \mathsf{negl}(\lambda).
\]
\end{definition}

We will now show, through a sequence of hybrids, that our construction satisfies the collapse-binding with public decodability property.
For any valid $q$-query\footnote{The bound $q$ on the number of queries made by $\calA$ will be important in Hybrid 3, when we apply the small-range distribution argument of~\cite{zhandry2021construct}.} adversary $\calA$, let $\Hyb_i^{\calA}$ be the advantage achieved by $\calA$ in Hybrid $i$.

\paragraph{Hybrid 0.}
This hybrid is the collapse-binding game defined above.

\paragraph{Hybrid 1.} We now move to a hybrid where we \red{replace the PRF $F_k$ in the scheme in \Cref{sec:construction-functional-commitment} with a random oracle $H: \calX \to \calR$}, so that we are effectively choosing a random instance of the one-shot signature oracles $\overline{\O}$ independently for every choice of verification key $\vk$.

Hybrid 0 and Hybrid 1 are indistinguishable because of the post-quantum security of the PRF, so for every valid $\calA$, $|\Hyb_0^\calA - \Hyb_1^\calA| \le \negl(\lambda)$.

\paragraph{Hybrid 2.} Observe that once $\calA_1$ has committed to $c^* = (\vk^*, \sigma^*_1, \overline{\vk})$, where $\sigma^*_1$ is a valid signature of $1$ under $\vk^*$, we can replace the $\CK_D$ oracle given to the second-stage adversary $\calA_2$ with $\CK_D\{\vk^*\}$, which works as follows.
\newline

\noindent
$\CK_D\{\vk^*\}(\vk, \sigma, \overline{\vk}, u)$:
\begin{itemize}
    \item If $\vk = \vk^*$, output $\bot$.
    \item Otherwise, output $\CK_D(\vk, \sigma, \overline{\vk}, u)$.
\end{itemize}
\noindent
Formally, the experiment in this hybrid is as follows:
\begin{itemize}
  \item Sample $(\CK, \dk) \leftarrow \mathsf{Gen}(1^\secp)$.

  \item Run $\calA_1^{\mathsf{CK}, \mathsf{DecZ}[\dk], \mathsf{DecX}[\dk]}(1^\lambda)$ until it outputs a commitment $c^* = (\vk^*, \sigma_1^*, \overline{\vk})$ and a state on registers $(\mathcal{B}, \mathcal{U}, \mathcal{R})$. Here $\calB$ is a single-qubit register, $\calU$ is the opening register and $\calR$ holds the internal state of $\calA$.

  \item Sample $b \leftarrow \{0,1\}$. If $b = 0$, do nothing, and if $b=1$ measure $(\mathcal{B}, \mathcal{U})$ in the computational basis.

  \item Run $\calA_2^{\red{\CK\{\vk^*\}}, \mathsf{DecZ}[\dk]}(\mathcal{B}, \mathcal{U}, \mathcal{R})$ until it outputs a bit $b'$, where $\red{\CK\{\vk^*\}} = (\O, \CK_P, \CK_{P^{-1}}, \red{\CK_D\{\vk^*\}})$.
  
  \item The experiment outputs 1 if $b = b'$.
\end{itemize}

\begin{claim}
    For every valid adversary $\calA$, $|\Hyb_1^\calA - \Hyb_2^\calA| = \negl(\secp)$.
\end{claim}
\begin{proof}
By \Cref{lem:evasive-oracles}, if an adversary is able to distinguish between these two hybrids, there must be an adversary $\calA'$ that produces an input on which $\CK$ and $\CK\{\vk^*\}$ differ. The only points on which $\CK_D$ and $\CK_D\{\vk^*\}$ differ are of the form $(\vk^*, \sigma_0, \cdot, \cdot)$, where $\sigma_0$ is a valid signature of $0$ under $\vk^*$. By the strong unforgeability of the one-shot signatures scheme $\OSS$, once the adversary has committed to a valid signature $\sigma^*_1$ of $1$ under $\vk^*$, it cannot produce a valid $\sigma_0$, and therefore cannot notice the difference between the oracles $\CK_D$ and $\CK_D\{\vk^*\}$. 
\end{proof}

\paragraph{Hybrid 3.} This hybrid is the same as Hybrid 2, except that the we define the punctured $\CK_D$ oracle differently, as follows.
\noindent
$\CK_D\{H(\vk^*)\}(\vk, \sigma, \overline{\vk}, u)$:
\begin{itemize}
    \item \red{If $H(\vk) = H(\vk^*)$, output $\bot$.}
    \item Otherwise, output $\CK_D(\vk, \sigma, \overline{\vk}, u)$.
\end{itemize}
\noindent
Formally, the experiment in this hybrid is as follows:
\begin{itemize}
  \item Sample $(\CK, \dk) \leftarrow \mathsf{Gen}(1^\secp)$.
  \item Run $\calA_1^{\mathsf{CK}, \mathsf{DecZ}[\dk], \mathsf{DecX}[\dk]}(1^\lambda)$ until it outputs a commitment $c^* = (\vk^*, \sigma_1^*, \overline{\vk})$ and a state on registers $(\mathcal{B}, \mathcal{U}, \mathcal{R})$. Here $\calB$ is a single-qubit register, $\calU$ is the opening register and $\calR$ holds the internal state of $\calA$.
  \item Sample $b \leftarrow \{0,1\}$. If $b = 0$, do nothing, and if $b=1$ measure $(\mathcal{B}, \mathcal{U})$ in the computational basis.
  \item Run $\calA_2^{\red{\CK\{H(\vk^*)\}}, \mathsf{DecZ}[\dk]}(\mathcal{B}, \mathcal{U}, \mathcal{R})$ until it outputs a bit $b'$, where $\red{\CK\{H(\vk^*)\}} =\allowbreak (\O, \CK_P, \CK_{P^{-1}}, \allowbreak\red{\CK_D\{H(\vk^*)\}})$.
  \item The experiment outputs 1 if $b = b'$.
\end{itemize}
\begin{claim}
    For every valid adversary $\calA$, $|\Hyb_2^\calA - \Hyb_{3}^\calA| = \negl(\secp)$.
\end{claim}
\begin{proof}
By \Cref{lem:evasive-oracles}, if an adversary is able to distinguish between these two hybrids, there must be an adversary $\calA'$ that produces an input on which $\CK\{\vk^*\}$ and $\CK\{H(\vk^*)\}$ differ. The only points on which $\CK_D$ and $\CK_D\{\vk^*\}$ differ are of the form $(\vk, \cdot, \cdot, \cdot)$, where $H(\vk) = H(\vk^*)$ but $\vk \neq \vk^*$. By the collision-resistance of random oracles (\Cref{lem:collision-resistance-RO}), no query-efficient adversary should be able to find such a collision with greater than negligible probability.
\end{proof}

\paragraph{Hybrid 4.} This hybrid is the same as Hybrid 3, except that use a small-range distribution argument to move to a hybrid where we \red{replace the random oracle $H$ introduced in Hybrid 1 with a random function $H_\SR$ of bounded range}. This will allow us to rephrase the collapse-binding game in terms of the oracles of the inner one-shot signature scheme $\overline{\OSS}$, which we then later analyze.

Let $R = 300 q^3/\epsilon$, where $\epsilon(\secp) = \Hyb_3^\calA$. We sample $R$ independent and uniformly random values $r_i \gets \calR$, for $i \in [R]$, and for each $x \in \calX$, we sample a random $i_x \gets [R]$ and set $H_\SR(x) := r_{i_x}$.

\begin{claim}\label{clm:small-range}
    For every valid $q$-query  adversary $\calA$, $\Hyb_4^\calA \ge \frac{4}{5}\Hyb_3^\calA$.
\end{claim}
\begin{proof}
By \Cref{lem:zhandry-small-range}, we know that the output distributions of $\calA$ in Hybrids 3 and 4 are at most $27q^3/R$-far in statistical distance, and by our choice of $R$, this means that $|\Hyb_3^\calA - \Hyb_4^\calA| \le \epsilon/5$, which in turn implies that $\Hyb_4^\calA \ge 4\epsilon/5$.
\end{proof}

We now show that no query-efficient adversary can win the game in Hybrid 4 with non-negligible advantage. Towards that end, consider the following game, which is similar in spirit to the collapse-binding game in \Cref{def:collapse-binding}, but rephrased in terms of the underlying oracles of $\overline{\OSS}$ in the \cite{shmueli2025one} construction.

\begin{definition}[OSS delayed collapsing experiment]\label{def:oss-collapse-binding}
Let $\overline{\OSS}$ be the one-shot signature scheme in \Cref{sec:oss-sz}. For any adversary $\calA = \{(\calA_1, \calA_2)_\lambda\}_{\lambda \in \bbN}$, where each of $\calA_1$ and $\calA_2$ are quantum operations that make at most $\poly(\lambda)$ many oracle queries, define the experiment $\mathsf{OSSDelayedCollapsingExpt}^{\calA}(1^\lambda)$ as follows.
\begin{itemize}
    \item Sample $P, P^{-1}, D \leftarrow \overline{\OSS}.\Setup(1^\lambda)$. Let $\{S_y\}_{y \in \{0,1\}^r}$ be the cosets defined by $P, P^{-1}$.
    \item Run $\calA_1^{P, P^{-1}, D}(1^\lambda)$ until it outputs a string $y \in \{0,1\}^r$ and a state on registers $(\calU, \calR)$. Here $\calU$ is $k$-qubit register and $\calR$ holds the internal state of $\calA$.
    \item Sample $b \leftarrow \{0,1\}$. If $b=0$, do nothing, and if $b=1$ measure $\calU$ in the computational basis.
    \item Run $\calA_2^{P, P^{-1}}(\calU, \calR)$ until it outputs a bit $b'$. The experiment outputs $1$ if $b=b'$ and $0$ otherwise.
\end{itemize}
The adversary is valid if the state on $\calU$ output by $\calA_1$ is in the image of $\Pi_{y} = \sum_{u \in S_y} |u\rangle \langle u|$. We say that the adversary $\calA$ wins the experiment with advantage
\begin{align*}
\left|\Pr[\mathsf{OSSDelayedCollapsingExpt}^\calA(1^\lambda) = 1] - \frac{1}{2}\right|.
\end{align*}
\end{definition}

Given an adversary $\calA = (\calA_1, \calA_2)$ that wins the game in Hybrid 4 with non-negligible advantage, we give an adversary $\calB = (\calB_1, \calB_2)$ that wins the OSS delayed collapsing game with non-negligible advantage. 

\noindent
The first-stage adversary $\calB_1$ works as follows:
\begin{itemize}
    \item $\calB_1$ receives oracles $P^*, (P^{-1})^*, D^*$ from the challenger.
    \item Let $\CK_{4} = (\O_4, \CK_{P, 4}, \CK_{P^{-1}, 4}, \CK_{D,4}), \DecZ_{4}, \DecX_{4}$ be oracles as defined in Hybrid 4 (using $H_\SR$ in place of the PRF), and let $\calR_\SR \subseteq \calR$ be the range of $H_\SR$, such that $|\calR_{\SR}| = R$. 
    \item Sample a random $r^* \gets \calR_\SR$.
    \item Define oracles $\CK^*_P, \CK_{P^{-1}}^*, \CK_D^*$ and $\DecZ^*, \DecX^*$ to be identical to $\CK_{P, 4}, \CK_{P^{-1}, 4}, \CK_{D, 4}$ and $\DecZ_4, \allowbreak \DecX_4$, except that on inputs $\vk$ such that $H_\SR(\vk) = r^*$, we use oracles $P^*, (P^{-1})^*, D^*$ instead of the oracles you get by running $\overline{\OSS}.\Setup(1^\secp; H_\SR(\vk))$. For example, $\CK_P^*(\vk, x)$ is defined as follows:
    \begin{itemize}
        \item[$\circ$] If $H_\SR(\vk) = r^*$, output $P^*(x)$.
        \item[$\circ$] Otherwise, output $\CK_P(\vk, x)$.
    \end{itemize}
    The oracles $\CK_{P^{-1}}^*, \CK_D^*$ and $\DecZ^*, \DecX^*$ are similarly defined.
    \item Run $\calA_1$ with $\CK^* = (\O_4, \CK^*_P, \CK^*_{P^{-1}}, \CK^*_D)$.
    \item Receive $c = (\widetilde{\vk}, \sigma, \overline{\vk})$ and registers $\calB, \calU$ from $\calA_1$.
    \item If $H_\SR(\widetilde{\vk}) = r^*$, send $y = \overline{\vk}$ and $\calU$ to the challenger.\footnote{For valid commitments the state on $(\calB, \calU)$ is supported on basis states $\ket{b, u}$ where $b = u[1]$, so measuring $\calU$ (or not) is the same as measuring both $\calB, \calU$.}
    \item Otherwise output a random $y$, where $(y, \calU) \leftarrow \OSS.\Gen^{P^*, (P^{-1})^*, D^*}(1^\secp)$.
\end{itemize}
The second stage $\calB_2$ adversary works as follows.
\begin{itemize}
    \item $\calB_2$ receives the oracles $P^*, (P^{-1})^*$ from the challenger.
    \item If $H_\SR(\widetilde{\vk}) \neq r^*$ then output $b' \gets \{0,1\}$ uniformly at random.
    \item If $H_\SR(\widetilde{\vk}) = r^*$, run $\calA_2$ with oracles $\CK^*_D\{H(\widetilde{\vk})\}$ and $\DecZ^*$, where $\CK^*_D\{H(\widetilde{\vk})\}(\vk, \sigma, \overline{\vk}, u)$ is defined just like $\CK_D\{H(\widetilde{\vk})\}$ from Hybrid 4
    \begin{itemize}
        \item[$\circ$] If $H(\vk) = H(\widetilde{\vk})$, then output $\bot$ and abort.
        \item[$\circ$] Otherwise, output $\CK^*_{D}(\vk, \sigma, \overline{\vk}, u)$.
    \end{itemize}
    \item Output the output of $\calA_2$.
\end{itemize}

Observe that the oracles $\CK^*, \DecZ^*, \DecX^*$ and $\CK^*\{\cdot\}$ given to $\calA_1, \calA_2$ will have the same distribution as in Hybrid 4. Since $|\calR_\SR| \le R$, $\calB_1$ correctly guesses $r^* = H_\SR(\widetilde{\vk})$ with probability at least $1/R$. So $\calB$ wins the OSS delayed collapsing game in \Cref{def:oss-collapse-binding} with advantage at least
\begin{align*}
    \frac{1}{R} \cdot \Hyb_4^\calA \ge \Omega\left(\frac{\epsilon^2}{q^3}\right),
\end{align*}
where $\epsilon = \Hyb_3^\calA$, and the inequality follows from \Cref{clm:small-range}. Finally, \Cref{lem:oss-collapsing} shows that no query-efficient adversary can win the OSS delayed collapsing game with non-negligible advantage, which completes the proof.
\end{proof}

\subsubsection{Hardness of the OSS delayed collapsing game}
\begin{theorem}\label{lem:oss-collapsing}
Let $\lambda \in \mathbb{N}$ denote the security parameter. Let $q,s,n,r$ be polynomials in $\lambda$, and let $s \leq n-r$ such that $\frac{k^9\cdot q^7\cdot \frac{1}{\epsilon^4}}{\sqrt{2^{n-r-s}}}\leq o(1)$.  Then, for every $q$ query adversary $\calA$,
\begin{align*}
\left|\Pr[\mathsf{OSSDelayedCollapsingExpt}^\calA(1^\lambda) = 1] - \frac{1}{2}\right| = \negl(\lambda).
\end{align*}
\end{theorem}

\begin{proof}
Suppose for contradiction that there exists an adversary $\calA=(\calA_1,\calA_2)$ with non negligible advantage in the OSS delayed collapsing game. Then, 
\begin{align*}
\left|\Pr[\mathsf{OSSDelayedCollapsingExpt}^\calA(1^\lambda) = 1] - \frac{1}{2}\right| \ge \epsilon(\lambda).
\end{align*}
for some non-negligible function $\epsilon$. Consider the following sequence of hybrids.

\paragraph{H$_1$:} 
This corresponds to $\mathsf{OSSDelayedCollapsingExpt}^\calA(1^\lambda)$

\paragraph{H$_2$:}
In this hybrid, we use the standard small-range distribution technique to simulate the oracles $P, P^{-1}, D$ using a polynomially bounded number of underlying cosets $(\matA_y, \vecb_y)$ sampled using a function $F':\{0,1\}^r\rightarrow \{0,1\}^{k \cdot (n-r+1)}$ as follows:
\begin{itemize}
    \item Let $R:=(300\cdot q^3)\cdot \frac{2^7\cdot k^2}{\epsilon}$
    \item for every $y \in \mathbb{Z}_2^r, $ sample $i_y \leftarrow [R]$.
    \item Sample a random function $F':\{0,1\}^r\rightarrow \{0,1\}^{k\cdot (n-r+1)}$. For every $i \in [R]$, interpret $F(i)=(\matA_i\in \mathbb{Z}_2^{k \times (n-r)},\vecb_i\in \mathbb{Z}_2^k)$.
    \item For $y \in \mathbb{Z}_2^r, $ define $F'(y):= (\matA_{i_y}, \vecb_{i_y})$.
\end{itemize}
\begin{claim}
    For any $q$ query quantum algorithm $\mathcal{A}$, 
    \[
    |\Pr[\textbf{H}^{\mathcal{A}}_2=1]-\Pr[\textbf{H}^{\mathcal{A}}_1=1]|\leq \frac{\epsilon}{8}
    \]
\end{claim}
\begin{proof}
    From Theorem A.6 in \cite{ananth2022pseudorandom}, it follows that for every quantum algorithm making at most $q$ queries, the distinguishing advantage  between $F$ and $F'$ is bounded by $\frac{300 \cdot q^3}{R}<\frac{\epsilon}{8}$.
\end{proof}

\paragraph{H$_3$: Bloating the dual}
Same as $\textbf{H}_2$, except the oracle $D$ is replaced with $D'$, which on input $y$,  accepts an evasive superspace of the dual subspace corresponding to the coset of $y$,
\begin{align*}
D'(y, \vecv) &= \begin{dcases}
              1 &\text{if } \vecv^\top \matA_y^{(1)} = \mathbf{0}^{n-r-s}\\
              0 &\text{otherwise.}
          \end{dcases}
\end{align*}, 
where $\matA_y^{(1)}\in \mathbb{Z}_2^{n-r-s}$ denotes the last $n-r-s$ columns of $\matA_y \in \mathbb{Z}_2^{k \times (n-r)}$.
\begin{claim}\label{claim:bloatdual}
For any $q$ query quantum algorithm $\cal{A}$, 
    \[
    |\Pr[\textbf{H}_2^{\calA}(1^{\lambda})=1]-\Pr[\textbf{H}_3^{\calA}(1^{\lambda})=1]\leq \frac{\epsilon}{8}.
    \]
\end{claim}
\begin{proof}
    By Lemma 19 in \cite{shmueli2025one}, due to $k-(k-(n-r))-s=n-r-s$ for every $i \in [R]$, changing $D$ to check for membership in $\matA_y^{(1)\bot}$ instead of $\matA_y^{\bot}$ is $O\Big(\frac{q\cdot s}{\sqrt{2^{n-r-s}}}\Big)$ indistinguishable, for any $q$ query algorithm. Since we use the above indistinguishability $R$ times, it follows that the distinguishing advantage in the current hybrid is  $:=\Pr[\textbf{H}_2^{\calA}(1^{\lambda})=1]-R\cdot O\Big(\frac{q\cdot s}{\sqrt{2^{n-r-s}}}\Big)\geq \frac{7 \epsilon}{8}-O\Big(\frac{k^2\cdot q^4\cdot s\cdot \frac{1}{\epsilon}}{\sqrt{2^{n-r-s}}}\Big)\geq \frac{3\epsilon}{4}$, when $\frac{k^9\cdot q^7\cdot \frac{1}{\epsilon^4}}{\sqrt{2^{n-r-s}}}\leq o(1)$. This implies the distinguishing advantage $\leq \frac{\epsilon}{8}$.
\end{proof}
\paragraph{H$_4$:} In this hybrid, we simulate the oracles $P, P^{-1}, D'$ from $\textbf{H}_3$ using a coset partition function (Definition \ref{def:cosetpartfunc}) that operates on a smaller input domain.  
Formally, given oracle access to any $(r + s, r, s)$ coset partition function $Q:\{0,1\}^{r+s}\rightarrow \{0,1\}^r$, the oracles $P,P^{-1}, D'$ are simulated as follows:
\begin{enumerate}
\item Sample a random permutation $\Gamma: \{0,1\}^n \rightarrow \{0,1\}^n$ and for every $y$, choose a random full rank matrix $\matC_y \in \mathbb{Z}_2^{k \times n}$, and a random vector $\vecd_y \in \mathbb{Z}_2^k$. 
    \item Simulating $P(x \in \{0,1\}^n)$. 
    \begin{itemize}
        \item $(x_0 \in \{0,1\}^{r+s}, x_1 \in \{0,1\}^{n-r-s})\gets \Gamma(x)$. 
        \item $y \gets Q(x_0)$. 
        \item $u \gets (\matC_y\cdot \Gamma (x)+\vecd_y)$. 
        \item Output $(y,u)$. 
    \end{itemize}
    \item Simulating $P^{-1}(y\in \{0,1\}^r,u\in \{0,1\}^{k})$. 
    \begin{itemize}
        \item $z \gets \matC_y^{-1}\cdot (u-\vecd_y)$. 
        \item $(z_0 \in \{0,1\}^{r+s}, z_1\in \{0,1\}^{n-r-s}):=z$.
        \item If $Q(z_0)=y$, output $\Gamma^{-1}(z)$. Otherwise, output $\bot$. 
    \end{itemize}
    \item Simulating $D'(\vecv,y)$
    \begin{itemize}
        \item $\matA_y^{(1)\bot}:=$ last $n-r-s$ columns of $\matC_y$. 
        \item Output $1$ iff $\vecv^{T}\cdot \matA^{(1)}_y=0^{n-r-s}$. 
    \end{itemize}
\end{enumerate}
\begin{claim}
 
    \[
    \Pr[\textbf{H}_4^{\calA}(1^{\lambda})=1]=\Pr[\textbf{H}_3^{\calA}(1^{\lambda})=1]
    \]    
\end{claim}
\begin{proof}
    We will argue that $\Gamma, \{\matC_y,\vecd_y\}_{y}$ and $ Q $ perfectly simulate $P,P^{-1}, D'$. Observe that simulating $P,P^{-1},D'$ using $Q$ is implicitly setting the following parameters:
    \begin{itemize}
       
        \item $\matA_y$: The preimage set $Q^{-1}(y)$ is a coset, which can be described as the set $\{\matB_y\cdot r+r_y\}$ as $r$ ranges over $\mathbb{Z}_2^{s}$. Thus $\matA_y:=\matC_y\cdot K_y$, where $K_y = \begin{pmatrix}\matB_y\in \mathbb{Z}^{(r+s)\times s}\\ &\ddots && \mathbf{I}_{n-r-s}
        \end{pmatrix}$. Since $K_y$ is full rank and $\matC_y$ is random, this implies that $\matA_y$ is random.
        \item $\vecb_y:= \vecd_y+\matC_y \cdot (r_y||0^{n-r-s})$. This follows by construction, now since $d_y$ is random, this implies that $b_y$ is also random. 
         \item $\Pi$: Define an augmented function $Q':\{0,1\}^n\rightarrow \{0,1\}^n$.  On input $\mathbf{z}$, the first $r$ bits of $Q'(\mathbf{z})$ are set to $y=Q(\mathbf{z})$. Define the function $\bar{J}(z)$ that outputs the unique vector in $\mathbb{Z}_2^{n-r}$ such that $\mathbf{z}=K_y\cdot \bar{J}(z)+r_y$. Then define $Q'(\mathbf{z})=(Q(\mathbf{z},\bar{J}(\mathbf{z}))$. This implies that $\Pi= Q'\circ \Gamma$. 
    \end{itemize}
\end{proof}
\paragraph{H$_5$:}
This hybrid describes the following experiment
\begin{itemize}
    \item Simulate $P, P^{-1}, D'$ given oracle access to $Q:\{0,1\}^{r+s}\rightarrow \{0,1\}^r$ as in $\textbf{H}_3$. 
    \item Run $\calA_1^{Q}(1^\lambda)$ until it outputs a string $y \in \{0,1\}^r$ and a state on registers $(\calU, \calR)$. Here $\calU$ is $k$-qubit register and $\calR$ holds the internal state of $\calA$.
    \item Initialize a register $T$ in the state $\ket{0}^n$ and  apply the map $(u,0)\rightarrow (u,\matC_y^{-1}(u-\vecd_y))$  on the joint system $\calU\otimes T$.
    \item Measure the first $r+s$ bits of register $T$ in the computational basis.
    \item Apply the map $(x,0) \rightarrow (x,\matC_y\cdot x +\vecd_y)$ on the resulting state in register $T \otimes \calU$.
    \item Uncompute the register $T$.
    \item Sample $b \leftarrow \{0,1\}$. If $b=0$, do nothing, and if $b=1$ measure $(\calB, \calU)$ in the computational basis.
    \item Run $\calA_2^{Q}(\calU, \calR)$ until it outputs a bit $b'$. The experiment outputs $1$ if $b=b'$ and $0$ otherwise.
\end{itemize}
\begin{claim}
For all QPT algorithms $\cal{A}$,
  \[
    |\Pr[\textbf{H}_5^{\calA}(1^{\lambda})=1]-\Pr[\textbf{H}_4^{\calA}(1^{\lambda})=1]\leq \text{negl}(\lambda).
    \]

\end{claim}
\begin{proof}
   Observe that $\mathcal{P}(x) = \left(y = Q\left(\overline{\Gamma(x)}\right), \vecu = \matC_y \begin{bmatrix}
        \overline{\Gamma(x)}\\
        \widetilde{\Gamma(x)}
    \end{bmatrix} + \vecd_y\right)$ where $\overline{\Gamma(x)}$ picks out the first $r+s$ bits and $\widetilde{\Gamma(x)}$ picks out the last $(n-r-s)$ bits. Now, for a fixed $y$, the state on register $\mathcal{U}$ has the following form:
    \[
    \sum_{\overline{\Gamma(x)}:Q(\overline{\Gamma(x))}=y, \widetilde{\Gamma(x)}
    }\matC_y \begin{bmatrix}
        \overline{\Gamma(x)}\\
        \widetilde{\Gamma(x)}
    \end{bmatrix}+\vecd_y
    \]
    Measuring the first $r+s$ bits of register $T$ in the computational basis and then applying the map $(x,0) \rightarrow (x,C_y(x)+d_y)$ on the resulting state in register $T \otimes \calU$ yields the following state:
      \[
    \sum_{ \widetilde{\Gamma(x)}
    }\matC_y \begin{bmatrix}
        \overline{\Gamma(x)}\\
        \widetilde{\Gamma(x)}
    \end{bmatrix}+\vecd_y
    \], where $\overline{\Gamma(x)}$ is such that $Q(\overline{\Gamma(x)})=y.$\\
    Now, \cite{shmueli2025one} construct collision resistant coset partition functions by composing  $2$ to $1$ functions together. More concretely, let $H:\{0,1\}^n \rightarrow \{0,1\}^{n-1}$ be a random $2$ to $1$ function. Given this, the coset partition function $H^{\ell}:\{0,1\}^{n\ell}\rightarrow \{0,1\}^{(n-1)\ell}$ is the following:
    \[H^{\ell}(x_1, \dots x_{\ell}):=H(x_1), \dots H(x_{\ell})\]
    In particular, constructing $Q$ in this manner, we can conclude that, since it is a composition of $2$ to $1$ functions which are known to be collapsing from \cite{zhandrycollapse} and composition preserves the collapsing property,  
    for any computationally bounded adversary $\cal{A}$, the above state is indistinguishable from  \[
    \sum_{\overline{\Gamma(x)}:Q(\overline{\Gamma(x))}=y, \widetilde{\Gamma(x)}
    }\matC_y \begin{bmatrix}
        \overline{\Gamma(x)}\\
        \widetilde{\Gamma(x)}
    \end{bmatrix}+\vecd_y
    \] and the claim then follows immediately. 
\end{proof}

\paragraph{H$_6$:} In this hybrid, we ``unsimulate'' the oracles, and revert to defining the oracles $P, P^{-1},D'$ as in $\textbf{H}_3$.
This hybrid describes the following experiment:
\begin{itemize}
    \item Sample $P, P^{-1}, D'\leftarrow \overline{\OSS}.\Setup(1^\lambda)$. 
  \item Run $\calA_1^{P,P^{-1},D'}(1^\lambda)$ until it outputs a string $y \in \{0,1\}^r$ and a state on registers $(\calU, \calR)$. Here $\calU$ is $k$-qubit register and $\calR$ holds the internal state of $\calA$.
    \item Initialize a register $T$ in the state $\ket{0}^n$ and  apply the map $(u,0)\rightarrow (u,(\matA_y^{-1}(u-\vecb_y))$ on the joint system $\calU\otimes T$.
    \item Uncompute $\calU$
    and measure the first $s$ bits of register $T$ in the computational basis, and then apply the map $(x,0) \rightarrow (x,\matA_y\cdot x+\vecb_y)$ on the resulting state in register $T \otimes \calU$. 
    \item Uncompute the register $T$. 
    \item Sample $b \leftarrow \{0,1\}$. If $b=0$, do nothing, and if $b=1$ measure $(\calB, \calU)$ in the computational basis.
    \item Run $\calA_2^{P, P^{-1}}(\calU, \calR)$ until it outputs a bit $b'$. The experiment outputs $1$ if $b=b'$ and $0$ otherwise.
\end{itemize}
\begin{claim}
    \[
    \Pr[\textbf{H}_6^{\calA}(1^{\lambda})=1]=\Pr[\textbf{H}_5^{\calA}(1^{\lambda})=1]
    \] 
\end{claim}
\begin{proof}
    This hybrid distributes identically to the previous hybrid by the simulation of $P,P^{-1},D'$ by $Q$ in $\textbf{H}_4$. 
\end{proof}
\paragraph{H$_{7}$:} Same as \textbf{H$_{6}$} except,  $\forall y\in [R], \matA_y$ is sampled as follows: sample a random function $F: \{0,1\}^r \rightarrow \{0,1\}^{k \cdot (n-r+1)}$ and a random permutation $M:\{0,1\}^{n-r}\rightarrow \{0,1\}^{n-r}$.  Interpret $F(y) = (\matA'_y, \vecb_y)$ where $\matA'_y \in \mathbb{Z}_2^{k \times (n-r)}$ is a full-rank matrix and $\vecb \in \mathbb{Z}_2^{n}$, and interpret $M(y)=M_y$, where $M_y\in \mathbb{Z}_2^{(n-r) \times (n-r)}$ is a full-rank matrix. Set $\matA_y= \matA'_y \circ M_y$.
\begin{claim}
     \[
    \Pr[\textbf{H}_7^{\calA}(1^{\lambda})=1]=\Pr[\textbf{H}_6^{\calA}(1^{\lambda})=1]
    \] 
\end{claim}
\begin{proof}
    Ths distribution of $\matA_y$ is identical in both hybrids, $\textbf{H}_7$ and $\textbf{H}_8$.
\end{proof}
\paragraph{H$_{8}$:} Same as $\textbf{H}_7$ except we change how $\Pi: \{0,1\}^n \rightarrow \{0,1\}^n$ is sampled in $\overline{\OSS}$:
\begin{enumerate}
   \item Sample a random permutation $\Pi':\{0,1\}^n \rightarrow \{0,1\}^n$. 
    \begin{itemize}
        \item  Let $\sigma: \{0,1\}^{n}\rightarrow \{0,1\}^{n}$ be the following permutation: $\sigma(x)$: Parse input as $(y \in \{0,1\}^{r}, z \in \{0,1\}^{n-r})$, and output $(y, M_y^{-1}(z))$.  
    \end{itemize}
    \item $\Pi(x)=\Pi'\circ \sigma(x)$. 
\end{enumerate}
\begin{claim}
 \[
    \Pr[\textbf{H}_8^{\calA}(1^{\lambda})=1]=\Pr[\textbf{H}_7^{\calA}(1^{\lambda})=1]
    \]    
\end{claim}
\begin{proof}
    The 2 hybrids are distributed identically.
\end{proof}

\begin{claim}
  If there exists an adversary $(\mathcal{A}_1, \mathcal{A}_2)$ that wins  the experiment in \textbf{H$_{8}$} with non negligible advantage, then there exists an adversary which wins the coset-collapsing game with non negligible advantage.
\end{claim}
\begin{proof}
    First, observe that the oracles $P,P^{-1},D'$ in $\textbf{H}_8$ are the following:
    \begin{itemize}
   \item Sample a random permutation $\Pi':\{0,1\}^n \rightarrow \{0,1\}^n$. 
    \begin{itemize}
        \item  Let $\sigma: \{0,1\}^{n}\rightarrow \{0,1\}^{n}$ be the following permutation: $\sigma(x)$: Parse input as $(y \in \{0,1\}^{r}, z \in \{0,1\}^{n-r})$, and output $(y, M_y^{-1}(z))$.  
    \end{itemize}
    \item $\Pi(x)=\Pi'\circ \sigma(x)$. Let $H(x)$ be the first $r$ output bits of $\Pi(x)$, and let $J'(x)$ denote the remaining $n-r$ bits.
    
    \item Sample a random function  $F: \{0,1\}^r \rightarrow \{0,1\}^{k \cdot (n-r+1)}$ and a random permutation $M:\{0,1\}^{n-r}\rightarrow \{0,1\}^{n-r}$.  For every $y \
    \in [R]$, Interpret $F(y) = (\matA'_y, \vecb_y)$ where $\matA'_y \in \mathbb{Z}_2^{k \times (n-r)}$ is a full-rank matrix and $\vecb \in \mathbb{Z}_2^{n}$, and interpret $M(y)=M_y$, where $M_y\in \mathbb{Z}_2^{(n-r) \times (n-r)}$ is a full-rank matrix. Set $\matA_y= \matA'_y \circ M_y$. Let  $S_y = \{\matA'_y \vecx + \vecb_y \mid \vecx \in \bbZ_2^{n-r} \}$ be the coset defined by $(\matA'_y, \vecb_y)$.
    \item Now,  the functions $P,P^{-1},D'$ are defined as follows:
    \begin{align*}
          P(x) &= (y, \matA'_y\cdot  J(x) + \vecb_y) \in \{0,1\}^n \rightarrow \{0,1\}^r \times \bbZ_2^{k} \text{ where } y = H(x)\\
          P^{-1}(y, \vecu) &= \begin{dcases}
              \Pi' (y, \vecz) &\text{if }\exists \vecz \in \bbZ_2^{n-r} \text{ such that } \matA'_y \vecz + \vecb_y = \vecu\\
              \bot &\text{otherwise.}
          \end{dcases}\\
          D'(y, \vecv) &= \begin{dcases}
              1 &\text{if } \vecv^\top \matA^{(1)}_y = \mathbf{0}^{n-r-s}\\
              0 &\text{otherwise.}
          \end{dcases}
      \end{align*}
\end{itemize}
The experiment in $\textbf{H}_8$ is the following:
\begin{itemize}
\item Run $\calA_1^{P,P^{-1},D'}(1^\lambda)$ until it outputs a string $y \in \{0,1\}^r$ and a state on registers $(\calU, \calR)$. Here $\calU$ is $k$-qubit register and $\calR$ holds the internal state of $\calA$.
    \item Initialize a register $T$ in the state $\ket{0}^n$ and  apply the map $(u,0)\rightarrow (u,(\matA^{-1}_y(u-\vecb_y))$ on the joint system $\calU\otimes T$.
    \item Measure the first $s$ bits of register $T$ in the computational basis, and then apply the map $(x,0) \rightarrow (x,\matA_y\cdot x+\vecb_y)$ on the resulting state in register $T \otimes \calU$. 
    \item Uncompute the register $T$. 
    \item Sample $b \leftarrow \{0,1\}$. If $b=0$, do nothing, and if $b=1$ measure $(\calB, \calU)$ in the computational basis.
    \item Run $\calA_2^{P, P^{-1}}(\calU, \calR)$ until it outputs a bit $b'$. The experiment outputs $1$ if $b=b'$ and $0$ otherwise.
    \end{itemize}
    Now, consider the following experiment. 
    \begin{definition}[Coset-collapsing game with Dual Access]
\label{def:coset-collapsing-dual access}
This is a game between a challenger and a two-stage adversary $(\calA_1, \calA_2)$. There is a public subspace $S \subset \F_2^n$ that is known to all parties.
    \begin{enumerate}
         \item The challenger samples a random subspace $T \subset S$, and sends to $A_1$ the oracle $O_{T^\perp}$. Note that $T^\perp$ is a random superspace of $S^{\perp}$. 
        \item The adversary $\calA_1^{O_{T^\perp}}$ outputs a quantum state on two registers $\mathcal{R}, \mathcal{U}$, and sends $\mathcal{U}$ to the challenger and $\mathcal{R}$ to the adversary $\calA_2$.
        \item The challenger does the following:
        \begin{enumerate}
            
            \item \label{itm:subspace-check} Challenger checks in superposition $O_S(\mathcal{U}) = \mathsf{Accept}$; if the check fails, the adversary automatically loses the game.
            \item Perform a measurement that projects $\mathcal{U}$ to a coset of $T$ in $S$. Concretely, let $\co(T)$ be the set of coset representatives of $T$. Then, the measurement is described by the projectors $\{\Pi_u\}_{\co(T)}$, where $\Pi_u := \sum_{t \in T} \ketbra{t+u}{t+u}$ for $u \in \co(T)$.
            \item The challenger samples random $b \leftarrow \{0,1\}$. If $b=0$, do nothing and send the register $\mathcal{U}$ to $\calA_2$. If $b=1$, measure the $\mathcal{U}$ register in the computational basis and then send it to $\calA_2$.
        \end{enumerate}
        \item The second stage adversary, $\calA_2$, does not receive access to $O_{T^{\perp}}$, and outputs a bit $b'$ and wins if $b' = b$.
    \end{enumerate}
\end{definition}

Now by construction, $\textbf{H}_8$ is exactly the coset-collapsing game with dual access when  $T= \text{Colspan}(\matA^{(1)}_y)$, $S= \text{Colspan}(\matA'_y)$. Assume for contradiction that there exists an adversary $(\mathcal{B}_1, \mathcal{B}_2)$ which has distinguishing advantage $\epsilon$ in $\textbf{H}_8$, where $\epsilon$ is a non negligible function.  
    Now consider an adversary $(\mathcal{A}_1, \mathcal{A}_2)$ for the coset collapsing game with dual access, with knowledge of the public subspace $S= \text{ColSpan}(\matA'_y)$  for $y \in [R]$ and access to a membership oracle for $T^{\bot}=\text{Colspan}(\matA^{(1)\perp}_y)$. 
    With probability $1/R$, $\calA$ guesses correctly the value of $y$ that $\calB$ commits to in the game. Conditioned on guessing correctly,
    $(\mathcal{A}_1, \mathcal{A}_2)$ can run $(\mathcal{B}_1, \mathcal{B}_2)$, so the distinguishing advantage of $(\mathcal{A}_1, \mathcal{A}_2)$ is $\frac{1}{R}\cdot \epsilon$, which is still non negligible. 
    We will establish next that the coset collapsing game with dual access is indistinguishable from the coset-collapsing game (Definition  \ref{def:coset-collapsing}), and prove that the coset-collapsing game is information theoretically hard (Theorem \ref{thm:coscollapse}).Assuming this, we can complete the proof by invoking Theorem \ref{thm:coscollapse}.
\end{proof}

\end{proof}

\subsubsection{Coset-collapsing game} \label{coset_collapsing_subsec}
We will first consider the following hybrids.\\ 
\noindent$\textbf{H}_0:$ This is the Coset collapsing game with dual access (Definition \ref{def:coset-collapsing-dual access}).\\
\noindent$\textbf{H}_1:$ Same as $\textbf{H}_0$, but $\mathcal{A}_1$ is given access to $O_{S^{\perp}}$ instead of $O_{T^\perp}$. 
\begin{claim} For any $q$ query algorithm $(\mathcal{A}_1, \mathcal{A}_2)$, 
    \[\Pr[1 \gets \textbf{H}^{(\mathcal{A}_1, \mathcal{A}_2)}_0]-\Pr[1 \gets \textbf{H}^{(\mathcal{A}_1, \mathcal{A}_2)}_1]\leq \text{negl}(\lambda)\]
\end{claim}
\begin{proof}
   We unbloated the dual in $\textbf{H}_1$, and the proof is identical to the proof of Claim \ref{claim:bloatdual}. 
\end{proof}
\noindent$\textbf{H}_2:$ Same as $\textbf{H}_1$, but $(\mathcal{A}_1, \mathcal{A}_2)$ are not given access to $O_{S^{\perp}}$.
\begin{claim}
    \[\Pr[1 \gets \textbf{H}^{(\mathcal{A}_1, \mathcal{A}_2)}_1]=
\Pr[1 \gets \textbf{H}^{(\mathcal{A}_1, \mathcal{A}_2)}_2]\]
\end{claim}
\begin{proof}
    This claim follows because $(\mathcal{A}_1, \mathcal{A}_2)$ know the description of $S$, since its a public subspace.
\end{proof}
Now, $\textbf{H}_2$ is the coset collapsing game, which we define more concretely below.
\begin{definition}[Coset-collapsing game]
\label{def:coset-collapsing}
This is a game between a challenger and a two-stage adversary $(\calA_1, \calA_2)$. There is a public subspace $S \subset \F_2^n$ that is known to all parties. Let $n \ge k$.
    \begin{enumerate}
        \item The adversary $\calA_1$ outputs a quantum state on two registers $\mathcal{R}, \mathcal{U}$, and sends $\mathcal{U}$ to the challenger and $\mathcal{R}$ to the adversary $\calA_2$.
        \item The challenger does the following:
        \begin{enumerate}
            \item Sample a random subspace $T \subset S$ of $\dim(T) = k$.
            \item \label{itm:subspace-check} Challenger checks in superposition $O_S(\mathcal{U}) = \mathsf{Accept}$; if the check fails, the adversary automatically loses the game.
            \item Perform a measurement that projects $\mathcal{U}$ to a coset of $T$ in $S$. Concretely, let $\co(T)$ be the set of coset representatives of $T$. Then, the measurement is described by the projectors $\{\Pi_u\}_{\co(T)}$, where $\Pi_u := \sum_{t \in T} \ketbra{t+u}{t+u}$ for $u \in \co(T)$.
            \item The challenger samples random $b \leftarrow \{0,1\}$. If $b=0$, do nothing and send the register $\mathcal{U}$ to $\calA_2$. If $b=1$, measure the $\mathcal{U}$ register in the computational basis and then send it to $\calA_2$.
        \end{enumerate}
        \item Adversary $\calA_2$ outputs a bit $b'$ and wins if $b' = b$.
    \end{enumerate}
\end{definition}

\begin{theorem}
\label{thm:coscollapse}
    Every (potentially unbounded) quantum adversary wins the coset-collapsing game in \Cref{def:coset-collapsing} with advantage at most ${2^k}/{2^n}$.
\end{theorem}

\begin{proof}
    Since the subspace $S$ is public and known to the adversary, we can assume that the subspace-check in \Cref{itm:subspace-check} of \Cref{def:coset-collapsing} always passes, and can perform a change-of-basis so that $S = \F_2^{n'}$ where $n' = \dim(S)$.

    After this change of basis, if the challenger samples $b = 1$, the channel it applies is simply $\Phi_1$ as defined in \Cref{lem:change-of-basis-coset-collapsing}, and if $b = 0$, the challenger applies the $\Phi_2$ channel. Then, the theorem follows from \Cref{lem:change-of-basis-coset-collapsing}.
\end{proof}

\def\co{\mathsf{co}}
\begin{lemma}\label{lem:change-of-basis-coset-collapsing}
    The diamond distance between the following two channels $\Phi_1, \Phi_2$ is $$\| \Phi_1 - \Phi_2 \|_\diamond \le \frac{2^k}{2^n}.$$
    \begin{itemize}
        \item $\Phi_1$ is defined as follows: Sample a random subspace $T \subset \mathbb{F}^n_2$ of $\dim(T) = k$, and project onto the cosets of $T$. Concretely, let $\co(T)$ be the set of coset representatives of $T$, and let $\# T$ be the number of subspaces $T$ such that $\dim(T) = k$.
        \begin{align*}
            \Phi_1(\rho) := \frac{1}{\# T} \sum_T \sum_{u \in \co(T)} \Pi_u \rho \Pi_u,
        \end{align*}
        where $\Pi_u := \sum_{t \in T} \ketbra{t+u}{t+u}$ for any $u \in \co(T)$. 
        \item $\Phi_2$ is defined as simply measuring in the computational basis
        \begin{align*}
            \Phi_2(\rho) := \sum_{x \in \F_2^n} \Pi_x \rho \Pi_x,
        \end{align*}
        where $\Pi_x := \ketbra{x}{x}$.
    \end{itemize}
\end{lemma}

\begin{proof}
First let us define the following pairs of channels:
\begin{itemize}
    \item Let $\Phi'_1$ (resp. $\Phi'_2$) be channels defined as follows: append an EPR pair $\ket{\EPR}_{\calC, \calD} \propto \sum_{x} \ket{x}_{\calC} \otimes \ket{x}_{\calD}$ on registers $\calC, \calD$ and teleport the state on input register $\calA$ to $\calD$. Then, apply $\Phi_1$ (resp. $\Phi_2$) on register $\calD$ and output the resulting state. Trace out registers $\calA, \calC$. 
    \item Let $\Phi_1''$ (resp. $\Phi_2''$) be channels defined as follows: append an EPR pair on registers $\calC, \calD$. Apply the channel $\Phi_1$ (resp. $\Phi_2$) to register $\calD$. Perform the "teleportation step": apply a Bell basis measurement on registers $\calA, \calC$ and then apply the $\Z$-gate teleportation corrections on register $\calD$. Trace out registers $\calA, \calC$.
\end{itemize}

Since $\Phi_1', \Phi_2'$ simply teleports the input state and then applies $\Phi_1, \Phi_2$, it is clear that $$\|\Phi_1 - \Phi_2\|_\diamond = \|\Phi'_1 - \Phi'_2\|_\diamond.$$

The teleportation step consists of a Bell basis measurement on registers $\calA, \calC$ followed by a Pauli correction on register $\calD$. \Cref{clm:measurement-Pauli-commute} shows that $\Phi_1$ and $\Phi_2$ absorb the Pauli $\X$ corrections and commute with the Pauli $\Z$ corrections, which implies that $$\|\Phi_1 - \Phi_2\|_\diamond = \|\Phi''_1 - \Phi''_2\|_\diamond.$$

\begin{claim}\label{clm:measurement-Pauli-commute}
    Let $\Phi_1', \Phi_2'$ and $\Phi_1'', \Phi_2''$ be channels as defined above. Then, $\Phi_1' = \Phi_1''$ and $\Phi_2' = \Phi_2''$.
\end{claim}
\begin{proof}
    Observe that $\Phi_1, \Phi_2$ are both channels of the form $\Phi(\rho) = \sum_i p_i \Pi_i \rho \Pi_i$ for some $\{(p_i, \Pi_i, S_i)\}_{i \in \mathcal{I}}$ such that
    $$\Pi_i = \sum_{x \in S_i} \ketbra{x}{x},$$
    where $\calI$ is a finite index set, $\sum_i p_i = 1$ are probability masses, $\Pi_i$ are projections\footnote{Note that it is not true in the case of $\Phi_1$ that $\sum_i \Pi_i = I$.} and $S_i \in \{0,1\}^n$ are (not necessarily disjoint) sets that cover all of $\{0,1\}^n$. When convenient, we will alternate between the following two notations: $\Pi_i = \Pi_{S_i}$.

    \begin{itemize}
        \item For $\Phi_1$, the index set is $\calI = \{(T, u)\}$ where $T$ is a $k$-dimensional subspace of $\F_2^n$ and $u \in \co(T)$ is a representative of one of the cosets of $T$. The probabilities $p_i = \frac{1}{\# T}$ and $S_{T, u} = T+u$.
        \item For $\Phi_2$, the index set is particular simple $\calI = \{x \mid x \in \F_2^n\}$ and $S_x = \{x\}$.
    \end{itemize}

    The channels $\Phi_1, \Phi_2$ have the further property that the probability distribution $\{p_i\}_{i \in \calI}$ is uniform, and for every $\barx \in \{0,1\}^n$, shifting by $\barz$ simply rotates the sets $S_i$. That is, $\{S_i\}_{i \in \calI} = \{S_i + \barx\}_{i \in \calI}$.

    For any Pauli corrections $\overline{x}, \overline{z}$, and for any input state $\rho$, we have that
    \begin{align*}
        \Phi(\X^{\barx} \Z^{\barz} \rho \X^{\barx} \Z^{\barz}) &= \sum_i p_i \, \Pi_i \, \X^{\barx} \Z^{\barz} \, \rho \, \X^{\barx} \Z^{\barz} \Pi_i = \sum_i p_i \Pi_{S_i + \barx} \Z^{\barz} \rho \Z^{\barz} \Pi_{S_i + \barx}\\
        &= \Z^{\barz} \sum_i p_i \Pi_i \rho \Pi_i \Z^{\barz} = \Z^{\barz} \Phi(\rho) \Z^{\barz}.
    \end{align*}
\end{proof}

Since these channels are first applying $\Phi_1$ (resp. $\Phi_2$) to half of the EPR pair and the performing some additional measurements as part of the teleportation step (which can only decrease the distance), we only need to worry about how $\Phi_1, \Phi_2$ act on halves of EPR pairs. That is, for every bipartite state $\rho_{\calA, \calB}$,
\begin{align*}
    \|(\Phi''_1 \otimes I) \rho_{\calA, \calB} - (\Phi''_2 \otimes I) \rho_{\calA, \calB}\|_1
    &\le \| (I \otimes \Phi_1) \ketbra{\EPR}{\EPR} - (I \otimes \Phi_2) \ketbra{\EPR}{\EPR}\|_1.
\end{align*}
Therefore,
\begin{align*}
    \|\Phi_1 - \Phi_2 \|_\diamond &=  \|\Phi''_1 - \Phi''_2 \|_\diamond \\
    &= \max_{\rho_{\calA, \calB}} \|(\Phi''_1 \otimes I) (\rho_{\calA, \calB}) - (\Phi''_2 \otimes I) (\rho_{\calA, \calB})\|_1\\
    &\le \| (I \otimes \Phi_1) \ketbra{\EPR}{\EPR} - (I \otimes \Phi_2) \ketbra{\EPR}{\EPR}\|_1.
\end{align*}
For channels $\Phi$ of the form in the proof of \Cref{clm:measurement-Pauli-commute},
\begin{align*}
    I \otimes \Phi (\ketbra{\EPR}{\EPR}) &= \sum_i \sum_{x, x'} p_i (I \otimes \Pi_i) \ketbra{x,x}{x',x'} (I \otimes \Pi_i)\\
    &= \sum_{i} \sum_{x, x' \in S_i} p_i \ketbra{x, x}{x', x'}\\
    &= \CNOT(\Phi(\ketbra{+}{+}^{\otimes}) \otimes \ketbra{0^n}{0^n}).
\end{align*}
Since the map $\CNOT(\cdot \otimes \ketbra{0^n}{0^n})$ preserves the trace distance, we have that
\begin{align*}
    \| (I \otimes \Phi_1) \ketbra{\EPR}{\EPR} - (I \otimes \Phi_2) \ketbra{\EPR}{\EPR}\|_1 = \| \Phi_1(\ketbra{+}{+}^{\otimes n}) - \Phi_2(\ketbra{+}{+}^{\otimes n})\|_1,
\end{align*}
and \Cref{lem:diamond-distance-plus-state} completes the proof.
\end{proof}

\begin{lemma}\label{lem:diamond-distance-plus-state}
\begin{align*}
    \left\|\Phi_1\left(\ketbra{+}{+}^{\otimes n}\right) - \Phi_2\left(\ketbra{+}{+}^{\otimes n}\right)\right\|_1 \le \frac{2^k}{2^n}.
\end{align*}
\end{lemma}
\begin{proof}
For a fixed subspace $T$ of $\dim(T) = k$,
\begin{align*}
\sum_{u \in \co(T)} \Pi_u \ketbra{+}{+}^{\otimes n} \Pi_u &= \sum_{u \in \co(T)} \left(\sum_{t \in T} \ketbra{t+u}{t+u}\right) \left(\frac{1}{2^n}\sum_{x, x' \in \F_2^n} \ketbra{x}{x'}\right) \left(\sum_{t' \in T} \ketbra{t'+u}{t'+u}\right)\\
&= \frac{1}{2^n} \sum_{u \in \co(T)} \sum_{t, t' \in T} \sum_{x, x'} \ket{t+u}\!\!\braket{t+u|x}\!\!\braket{x'|t'+u} \!\!\bra{t'+u}\\
&= \frac{1}{2^n} \sum_{u \in \co(T)} \sum_{t, t' \in T} \ketbra{t+u}{t'+u}.
\end{align*}
Grouping the cosets of $T$ together, this is a block diagonal matrix containing $2^{n-k}$ blocks, where each block is $\frac{1}{2^n} \mathbf{1} \mathbf{1}^\top 2^k \times 2^k$, where $\mathbf{1}\mathbf{1}^\top$ is the all ones matrix. 

Average over all possible subspaces $T$, we can argue that by symmetry, every diagonal entry must be the same value, say $\alpha$, and every off-diagonal entry must also be the same value $\beta$
\begin{align*}
    \Phi_1 \left(\ketbra{+}{+}^{\otimes n}\right) = \sum_{x} \alpha \ketbra{x}{x} + \sum_{x \neq x'} \beta \ketbra{x}{x'}.
\end{align*}
By an averaging argument, we get that $\alpha = \frac{1}{2^n}$ and $\beta = \frac{2^k - 1}{2^{2n} - 2^n}$.

Channel $\Phi_2$ is simply the full standard basis measurement, so
\begin{align*}
    \Phi_2 \left(\ketbra{+}{+}^{\otimes n}\right) &= \frac{1}{2^n} \sum_{x \in \F_2^n} \ketbra{x}{x}.
\end{align*}
Since $\mathbf{1}_d \mathbf{1}_d^\top - I_d$ has eigenvalues $(0, -1, -1, \ldots, -1)$, we know that $\|\mathbf{1}_d \mathbf{1}_d^\top - I_d\|_1 = d-1$, and
\begin{align*}
    \left\|\Phi_1 \left(\ketbra{+}{+}^{\otimes n}\right) - \Phi_2 \left(\ketbra{+}{+}^{\otimes n}\right)\right\|_1 &= \beta \cdot (2^n -1) = \frac{2^k - 1}{2^n}.
\end{align*}
\end{proof}
\section{Post-quantum Succinct Ideal Obfuscation}\label{sec:input-succinct}

\begin{definition}[Post-quantum succinct ideal obfuscation]\label{def:succinct-obf}
    A (post-quantum) succinct ideal obfuscator in the oracle model consists of the following algorithms.
    \begin{itemize}
        \item $\Obf(1^\secp,P,T) \to (O[P], \pp)$: The PPT obfuscator takes as input the security parameter $1^\secp$, the description of a classical Turing machine $P$ with single-bit output, and a time bound $T = \poly(\secp)$, and outputs the description of a \emph{circuit} $O[P]$ and public parameters $\pp$.
        \item $\Eval^{O[P]}(\pp, x) \to y$: The PPT evaluation algorithm has oracle access to $O[P]$, takes an input $x$, along with public parameters $\pp$, and produces an output $y$.
    \end{itemize}
    It should satisfy the following properties.
    \begin{itemize}
        \item \textbf{Correctness.} For any $1^\secp,P,T,x$, and $(O[P], \pp) \in \Obf(1^\secp,P,T)$, it holds that $\Eval^{O[P]}(\pp, x) = P_T(x)$, where $P_T(x)$ is the bit on the output tape of $P$ after $T$ steps.
        \item \textbf{Succinctness.} There exists a fixed polynomial $p$ such that for any program $P$, time bound $T$, and $(O[P], \pp) \gets \Obf(1^\secp,P,T)$, it holds that $|(O[P], \pp)| = p(\secp,|P|)$.
        \item \textbf{Ideal obfuscation.} There exists a QPT simulator $\calS$ such that for any QPT adversary $\calA$, program $P$, time bound $T$, and auxiliary information $z$,
        \[\left\{\calA^{O[P]}(\pp, z) : O[P],\pp \gets \Obf(1^\secp,P,T)\right\} \approx \left\{\calS^{P_T}(z)\right\}.\]
    \end{itemize}
\end{definition}

Classically, the above follows from FHE and SNARKs in the oracle model in a straightforward manner. In this section, we prove that the same construction is also post-quantum secure. Before describing the construction, we will introduce some relevant prelimnaries. 
\subsection{SNARKs}
\begin{definition}\label{def:snark}
    A succinct non-interactive argument of knowledge (SNARK) for NP in the quantum random oracle model (QROM) consists of the following algorithms. Let $H : \{0,1\}^{2\secp} \to \{0,1\}^\secp$ be a fixed-size random oracle.
    \begin{itemize}
        \item $\Prove^H(1^\secp,C,x) \to (z,\pi)$: The p.p.t. prove algorithm takes as input the security parameter $1^\secp$, the description of a classical deterministic computation $C$, and an input $x$, and produces an output $z$ and a proof $\pi$.
        \item $\Ver^H(1^\secp,C,z,\pi) \to \{\top,\bot\}$: The deterministic classical verification algorithm takes as input the security parameter $1^\secp$, the description of a classical deterministic computation $C$, an output $z$, and a proof $\pi$, and either accepts or rejects.
    \end{itemize}
    It should satisfy the following properties.
    \begin{itemize}
        \item \textbf{Completeness.} For any $\secp,C,x,z$ such that $C(x) = z$,
        \[\Pr_H\left[\Ver^H(1^\secp,C,z,\pi) = \top : (z,\pi) \gets \Prove^H(1^\secp,C,x)\right] = 1.\]
        \item \textbf{Succinctness.} There exists a fixed polynomial $p$ such that for any time $T$ computation $C$, the running time of $\Ver^H(1^\secp,C,z,\pi)$ is $p(\secp,\log T,|C|,|z|)$. Note that $|x|$ and the time $T$ might be much larger than $p(\secp,\log T,|C|,|z|)$.
        \item \textbf{Argument of Knowledge.} The argument $(\Prove, \Ver)$ is an argument of knowledge with extraction probability $\kappa$ if there exists a polynomial time quantum extractor $\mathcal{E}$ such that, for every input $x$ and $t$ query quantum oracle algorithm $\Tilde{P}$, if, over a random oracle $H$, for $\pi:= \Tilde{P}^{H}$, it holds that $\Ver^{H}(z,\pi)=1$ with probability $\mu$, the probability that $\mathcal{E}^{\Tilde{P}}(1^{t},1^{\lambda})$ outputs a valid witness for $\pi$ is at least $\kappa(t, \mu, \lambda)$. Here the notaion $\mathcal{E}^{\Tilde{P}}$ denotes that $\mathcal{E}$ has black box access to $\Tilde{P}$. 
    \end{itemize}
\end{definition}
\subsection{Probabilistically Checkable Proofs (PCPs)}
\begin{definition}
    A probabilistically checkable proof (PCP) for a relation $\cal{R}$ with soundness error $\epsilon$, proof length $\ell$ and alphabet $\Sigma$ is a pair of poly-time algorithms $(\mathbf{P}, \mathbf{V})$ for which the following holds:
    \begin{itemize}
        \item \textbf{Completeness:} For every instance-witness pair $(x,w)\in \cal{R}$, $\mathbf{P}(x,w)$ outputs a proof string $\Pi: [\ell]\rightarrow \Sigma$ such that $\Pr[\mathbf{V}^{\Pi}(x)=1]=1.$
        \item \textbf{Soundness:} For every instance $x \notin \cal{L}(\cal{R})$ and proof string $\Pi:[\ell]\rightarrow \Sigma$, $\Pr[\mathbf{V}^{\Pi}(x)=1]\leq \epsilon.$
        \item \textbf{Proof of Knowledge:} The PCP $(\mathbf{P}, \mathbf{V})$ has knowledge error $k$ if there exists a polynomial time extractor $\mathbf{E}$ such that for every witness $x$ and proof string $\Pi:[\ell]\rightarrow \Sigma$, if $\Pr[\mathbf{V}^{\Pi}(x)=1]>k$, then $\mathbf{E}(x,\Pi)$ outputs  a valid witness $w$. 
    \end{itemize}
\end{definition}
\subsection{The construction of Micali}
Let $(\textbf{P,V})$ be the corresponding PCP with prover $\textbf{P}$ and verifier $\textbf{V}$.
Before describing the construction formally, we will define some notation relevant to Merkle trees. 
Given a domain $Z$ and $\ell=2^d$, a \emph{Merkle tree} on a list $\mathbf{v}=(v_i)_{i \in \ell}\in Z^{\ell}$ with respect to the function $H: Z \times Z \rightarrow Z$ is a list of values $(z_{k,i})_{k \in \{0,\dots d\}, i \in [2^k]}\in Z^{2\ell}$, where 
\[
\forall i \in \ell, z_{d,i}=v_i
\]
\[
\forall k \in \{0,\dots d-1\}, \forall i \in [2^k], z_{k,i}=H(z_{k+1,2i-1}, z_{k+1,2i})
\]
The root of the merkle tree, denoted $\mathsf{rt}$ is $z_{0,1}$.
Given $i \in [\ell],$ where $[\ell]=2^d$, and $d$ is the depth of the tree, the authentication path for the $i^{th}$ leaf is 
\[
ap=(a_k)_{k \in d},
\] where $a_k$ is the sibling of the $k^{th}$ node on the path from $v_i$ (the $i^{th}$ node) to $\mathsf{rt}$. \\
To check whether an authentication path is valid, we define 
\[
\text{CheckPath}^{H}( \mathsf{rt},i,v,ap)=1 \text{ iff } (\bf{x}, \bf{y})\leftarrow \text{ Expand}^{H}( \mathsf{rt},i,v,ap) \text{ has } y_0= \mathsf{rt},
\] 
where $\text{ Expand}^{H}( \mathsf{rt},i,v,ap)$ outputs all the input-output pairs, $(\bf{x}, \bf{y})$, that arise when validating an authentication path. 
\paragraph{Extracting Trees:} We will recall the algorithm $\mathsf{Extract}$ that is used to obtain a Merkle tree rooted at a chosen entry in the database corresponding to the compresssed oracle for $H$. This algorithm is used in \cite{chiesa2019snark}.
In more detail, $\mathsf{Extract}$ receives a database $D: Z\times Z\rightarrow Z$, a root value $\mathsf{rt}\in Z$, and a height bound $d$, and outputs a rooteed binary tree $T=(V,E)$ of depth at most $d$ where $V \subseteq Z \times \{0,1\}^{\leq d}$. \\

\noindent Extract$(D,\mathsf{rt},d):$
\begin{enumerate}
    \item Initialize the vertex set  with the root $V:=\{(\mathsf{rt}, \emptyset)\}$ and the edge set to be empty $E:=\emptyset$.
    \item While there exists an unmarked vertex $(u,s)\in V$ such that $u \in \text{Im}(D)$ and $s \in \{0,1\}^{<d}$:
    \begin{itemize}
        \item Mark the vertex $(u,s)$.
        \item If $\exists x\neq x':D(x)=D(x')=u$, return $\bot$.
        \item Let $x_0, x_1$ be the unique values such that $D(x_0,x_1)=u$.
        \item Update the vertex set $V :=V \cup \{(x_0, s||0), (x_1, s||1)\}$. 
        \item Update the edge set $E :=E \cup \{((u,s),(x_0, s||0)),((u,s), (x_1, s||1))\}$. 
    \end{itemize}
Given a tree $T:= \mathsf{Extract}(D, \mathsf{rt},d) $we can define a string $\mathsf{leaves}(T)$ where for each $i \in \{0,1\}^d$, if there exists a vertex $(x,i)\in T$ then $\mathsf{leaves}(T)_i:=x$, else $\mathsf{leaves}(T)_i:=\bot$.
 \end{enumerate}
   The formal construction of the Micali SNARK is now as follows:
 
\begin{itemize}
    \item $\SNARK.\Prove(x,\mathrm{w})\rightarrow \pi:$ 
    \begin{itemize}
        \item Prover $P$ uses $H$ to commit to the proof string $\Pi\leftarrow \textbf{P}(x,\mathrm{w})$ via a Merkle tree, obtaining a corresponding root $\mathsf{rt}$. 
        \item Randomness $r$ is derived by applying $H$ to $\mathsf{rt}$.
        \item $P$ simulates $\textbf{V}(\Pi, x,r)$ to deduce the queried locations  $(v_i)_{i \in [q]}$ in $\Pi$.
        \item Output $\pi=\{\mathsf{rt}, (v_i)_{i \in [q]}, (ap_i)_{i \in [q]}\} $.
    \end{itemize}
    \item $\SNARK.\Verify(\pi,x)\rightarrow 0/1:$
    \begin{itemize}
        \item Using $H$ (by making queries to it), check that, for $i \in [q]$, 
        \[
        \text{CheckPath}^{H}( \mathsf{rt},i,v,ap)=1
        \]
        \item Let $r= H(\mathsf{rt})$.
        \item Run $\textbf{V}(x,r)$:
        \begin{itemize}
            \item For every query $j_i \in [\ell]$, where $j_i$ is the is the location of the $i^{th}$ query made by $\textbf{V}$, $V$ runs $\text{CheckPath}^{H}( \mathsf{rt},ap_i,v,j)$, and returns $v_i$ if $\text{CheckPath}^{H}( \mathsf{rt},ap_i,v,j)=1$. 
        \end{itemize}
        \item Output  the same answer as $\textbf{V}(x,r)$. 
    \end{itemize}
    
\end{itemize} 
\begin{theorem}\label{thm:micalisound}\cite{chiesa2019snark}
    Let $(\mathbf{P},\mathbf{V})$ be a PCP for relation $\cal{R}$ that has soundness error $\epsilon$, proof length $\ell$ and query complexity $q$. The Micali construction when based on $(\mathbf{P},\mathbf{V})$ is a non-interactive argument for $\cal{R}$ that has soundness error $O(t^2\epsilon+t^3/2^{\lambda})$ against quantum attackers $\Tilde{P}$ that make at most $t-O(q \log \ell)$ queries to $H$.
\end{theorem}
Now we will describe the knowledge extractor for the Micali SNARK in the QROM for any $t'$ query algorithm $\tilde{P}$.\\

\noindent $\mathcal{E*}^{\Tilde{P}}(1^{t'}, 1^{\lambda})$:
\begin{enumerate}
\item Set $t=t'+O(q\log \ell)$.
    \item Compute the quantum state $\ket{\text{Sim}^*(\tilde{P})}$ by simulating $\tilde{P}$, and recording its queries to the random oracle $H$ with compressed oracles. 
    \item Measure the database register to get a database $D$.
     \item For each $\mathsf{rt}\in D$:
    \begin{itemize}
        \item Run $T\leftarrow \text{Extract}(D,\mathsf{rt},d)$.  
        \item Set $\Pi_T$ equal to $\text{leaves}(T)$ on all points where $\text{leaves}(T)\neq \bot$, and $0$ otherwise. 
        \item Run $x\leftarrow \textbf{E} (\Pi_{T})$.
     \item Output $x$.
    \end{itemize}

\end{enumerate}

\begin{lemma}\label{lem:cmsextract}\cite{chiesa2019snark}  The Micali construction, when based on a PCP of knowledge with knowledge error $k$, has the following guarantee: 
\begin{itemize}
    \item For any quantum attacker $\tilde{P}$ that makes at most $t'$ queries to $H$, the following holds with probability $\Omega(\mu-t^2 k-t^3/2^{\lambda})$:
    \begin{itemize}
    \item For all inputs $(ct', \pi)$ such that  $\SNARK.\Verify(\pi)=1$ with probability $\geq \mu$, $\mathcal{E*}^{\tilde{P}}(1^{t'},1^{\lambda})$ extracts a valid witness $x$.
    \end{itemize}
\end{itemize}

\end{lemma}
\subsection{Construction}
Let $U_T(\cdot, \cdot)$ be the universal circuit that takes as input a Turing machine description $P$ and an input $x$ and outputs $P_T(x)$ (i.e.\ the output of the Turing machine after running for $T$ steps). Let $U_{T,x} := U_T(\cdot, x)$ be the same universal circuit except that the input is hardcoded and fixed to $x$, so it only takes as input a Turing machine description $P$.

Let $H$ be a random oracle. We describe the construction of succinct ideal obfuscation in \Cref{fig:input-succinct-obf}.

\begin{figure}[h]
\centering
\fbox{%
\parbox{0.96\textwidth}{%
\textbf{Succinct Ideal Obfuscation}
\begin{itemize}
    \item $\Obf^H(1^\secp,P,T)$:
    \begin{itemize}
        \item $(\pk, \sk) \leftarrow \FHE.\Gen(1^\secp)$
        \item $\ct \leftarrow \FHE.\Enc(\pk, P)$
        \item Define the circuit $O^H_{\ct, \sk}(\ct', \pi)$ to take as input a ciphertext $\ct'$ and proof $\pi$. It runs $\SNARK.\Verify^H((\ct, \ct'), \pi)$ to check that $\pi$ is a valid proof of the following NP statement:
        \begin{align*}
            \exists x : \ct' = \FHE.\Eval(\pk, U_{T,x}, \ct).
        \end{align*}
        If $\SNARK.\Verify^H((\ct, \ct'), \pi)$ accepts, output $y \leftarrow \FHE.\Dec(\sk, \ct')$ and otherwise output $\bot$.
        \item Output the obfuscation $(O[P], \pp)$ where $O[P] := O^H_{\ct, \sk}$ and $\pp := \ct$.
    \end{itemize}
    \item $\Eval^{O[P], H}(\pp, x)$:
    \begin{itemize}
        \item Parse $\pp = \ct$ and evaluate $\ct' \leftarrow \FHE.\Eval(\pk, U_x, \ct)$.
        \item Prove that $\ct'$ was honestly generated: run $\pi \leftarrow \SNARK.\Prove^H((\ct, \ct'), x)$ to get a SNARK $\pi$ of the statement
        \begin{align*}
            \exists x: \ct' = \FHE.\Eval(\pk, U_{T,x}, \ct).
        \end{align*}
        \item Output $O[P](\ct', \pi)$.
    \end{itemize}
\end{itemize}
}
}
\caption{Succinct Ideal Obfuscation using Classical Oracles and FHE}\label{fig:input-succinct-obf}
\end{figure}

\subsection{Correctness}
\begin{theorem}
    The construction in \Cref{fig:input-succinct-obf} satisfies correctness as defined in \Cref{def:succinct-obf}.
\end{theorem}
\begin{proof}
    Correctness of the succinct obfuscation scheme in \Cref{fig:input-succinct-obf} follows from the completeness of the SNARK and the correctness of the FHE scheme.
\end{proof}

\subsection{Succinctness}
\begin{theorem}
    The construction in \Cref{fig:input-succinct-obf} satisfies succinctness as defined in \Cref{def:succinct-obf}.
\end{theorem}
\begin{proof}
    For a Turing machine with description $P$, its obfuscation consists of a circuit $O[P]$ that's accessible as an oracle and a FHE ciphertext $\ct \leftarrow \FHE.\Enc(\pk, P)$. The ciphertext $\ct$ encrypts the description of $P$, so its size is bounded $|\ct| \le \poly(\secp, |P|)$, which is succinct. The circuit $O[P](\ct', \pi)$ consists of two circuits: (1) the SNARK verification circuit and (2) the FHE decryption circuit. By the succinctness of the SNARK scheme and compactness of the FHE scheme, both these circuits have size bounded by $\poly(\secp, |P|)$. Therefore the construction in \Cref{fig:input-succinct-obf} satisfies succinctness.
\end{proof}
\subsection{Security}
\begin{theorem}
    The construction in \Cref{fig:input-succinct-obf} satisfies ideal obfuscation security as defined in \Cref{def:succinct-obf}.
\end{theorem}
\begin{proof}
We will start by formally describing the simulator $\cal{S}$.\\

Before describing our simulator, we need some terminology. 
We define a procedure $\mathsf{FindWitness}$. Very roughly, $\mathsf{FindWitness}$ is the witness extractor from Lemma 7.4 in \cite{chiesa2019snark}, with the crucial difference being that  the database register is not measured by $\mathsf{FindWitness}$. \\\begin{figure}[h]
\centering
\fbox{%
\parbox{0.96\textwidth}{%
\textbf{FindWitness}
\begin{itemize}
    \item $\textbf{Input}$:
    \begin{itemize}
        \item Database $D$. 
    \end{itemize}
    \item $\textbf{Procedure:}$
    \begin{itemize}
   \item  Run $x \leftarrow \textbf{E}(0^{\ell})$, where $0^{\ell}$ denotes the all $0$s proof, and $\ell=2^d$ for some $d \in \mathbb{N}$ is the proof length. If $x$ is a valid witness, output $x$, otherwise continue. 
    \item For each $\mathsf{rt}\in D$:
    \begin{itemize}
        \item Run $T\leftarrow \text{Extract}(D,\mathsf{rt},d)$.  
        \item Set $\Pi_T$ equal to $\text{leaves}(T)$ on all points where $\text{leaves}(T)\neq \bot$, and $0$ otherwise. 
        \item Run $x\leftarrow \textbf{E} (\Pi_{T})$.
     \item Output $x$.
    \end{itemize}
    \end{itemize}
    
\end{itemize}
}
}
\caption{FindWitness}\label{fig:findwitness}
\end{figure}


Now we are ready to define our simulator.  On (superposition) query $(ct', \pi)$ to $O[P]$, simulator $\cal{S}$ does the following: 
\begin{itemize}
    \item $\cal{S}$ outputs $\pp= \FHE.\Enc(pk,0)$. 
   \item  The joint state of the distinguisher $\cal{D}$ and $H$ is given by a superposition over $\ket{ct',\pi}\ket{y}\otimes \ket{D}\ket{\text{aux}}$, where $\ket{ct',\pi}, \ket{y}$ are the query input and query output registers respectively, and $\ket{D},\ket{\text{aux}}$ are the database $D$ (recording the  queries to $H$) and $\cal{D}$'s auxiliary registers respectively. $\ket{D}$ will update according to the Ctso procedure.  
   \item On query $(ct', \pi)$ to $O[P]$, $\cal{S}$ performs the following operation on the state $\ket{ct'}\ket{\pi}\otimes \ket{D}$:
   \begin{itemize}
    \item Initilialize ancilla register $\ket{0}_b $. Run $\SNARK.\Verify^{H}(\pi)$ and XOR the output in register $b$.  If output is $0$, output  $\bot$. 
     \item Perform the map $\ket{ct'}\ket{\pi}\otimes \ket{D}\rightarrow \ket{ct'}\ket{\pi}\otimes \ket{D}\ket{ \mathsf{FindWitness}(D)}_{\cal{W}}$. 
    \end{itemize}
    \item Evaluate $P$ on the input from register $\cal{W}$, and XOR in the query output register.  
    \item Uncompute the $\cal{W}$ register by evaluating $\mathsf{FindWitness}(D)$ and XOR-ing in the $\cal{W}$ register.
\end{itemize}

    \noindent\textbf{Hybrid 0}: Sample $(O[P],\text{pp})\leftarrow \Obf^H(1^\secp,P)$.\\
    
    \noindent \textbf{Hybrid 1}: Same as \textbf{Hybrid 0}, but the random oracle $H$ is simulated as a compressed standard oracle for a database $D$. 
    \begin{claim}
        $\Pr[\textbf{Hybrid 0}=1]=\Pr[\textbf{Hybrid 1}=1]$.
    \end{claim}
    \begin{proof}
        This follows directly from Lemma 4 in \cite{zhandry2019record}.
    \end{proof}
    \noindent\textbf{Hybrid 2}: Same as \textbf{Hybrid 1},  but we add an abort condition for $\mathsf{FindWitness}(D)$ outputting invalid witnesses. More precisely, 
    \begin{itemize}
        \item On query $(ct', \pi)$ to $O[P]$, the joint state of the distinguisher $\cal{D}$ and $H$ is $\ket{ct',\pi}\ket{y}\otimes \ket{D}\ket{\text{aux}}$
        \begin{itemize}
            \item Initilialize ancilla register $\ket{0}_b $. Run $\SNARK.\Verify^{H}(\pi)$ and XOR the output in register $b$.  If output is $0$, output  $\bot$.\item Measure the database register $D$, and run $\mathsf{FindWitness}(D)$. 
            \item If $x \leftarrow \mathsf{FindWitness}(D)$ is not a valid witness, output $\bot$.
            \item Otherwise, compute  $y \leftarrow \FHE.\Dec(sk,ct')$ and XOR into the query output register of $\cal{D}$. 
        \end{itemize}
    \end{itemize}
    \begin{claim}\label{claim:extfailcms}
        $\Pr[\textbf{Hybrid 2}=1]-\Pr[\textbf{Hybrid 1}=1]\leq \text{negl}(\lambda)$
    \end{claim}
    \begin{proof}
Assume otherwise, then there must exist some query made by $\cal{D}$ with noticeable amplitude on $\ket{ct,\pi}\otimes \ket{D}$ such that $\mathsf{FindWitness}(D)$ does not output a valid witness, but $\SNARK.\Verify^{H}(\pi)=1$. Formally, define the following set $B$ of inputs of the form 
\[B= \{(ct', \pi, D, b):  x\leftarrow\mathsf{FindWitness}(D) \text{ is not valid but } b=1 \}\]
Let $\ket{\psi}$ be the joint state of $\cal{D}$ and $H$ right before measuring the database $D$. Then our assumption implies that
\[
||\Pi[B]\ket{\psi}||^2=\epsilon(\lambda)
\]
for some non negligible function $\epsilon$.
For a database $D$, define $B[D]:=\{(.,D): (.,D)\in B\}$. We can conclude that there exists a database $D'$ such that 
\[
||\Pi[B[D']]\ket{\psi}||^2=\epsilon'(\lambda),
\]where $\epsilon'(.)$ is also a non negligible function.
Now define the state 
\[
\ket{\psi'}=\bra{x}\mathsf{FindWitness}(D')\ket{\psi}, \ket{\psi^*}:=\frac{\ket{\psi'}}{||\ket{\psi'}||}
\]
So $\ket{\psi^*}$ is the result of running $\mathsf{FindWitness}(D')$ coherently, and post selecting on the witness being invalid. 
From the gentle measurement lemma, we can conclude that, 
the probability that $\mathsf{FindWitness}(D')$ finds  an invalid witness on $\Pi[B[D']]\ket{\psi^*}$ is non negligible as well,  but this is a contradiction to Lemma \ref{lem:cmsextract} .  
    \end{proof}
   \noindent \textbf{Hybrid 3}: Same as \textbf{Hybrid 2}, but queries to $O[P]$ are answered  as follows:
   \begin{itemize}
      \item On query $(ct', \pi)$ to $O[P]$, $\cal{S}$ performs the following operation on the state $\ket{ct'}\ket{\pi}\otimes \ket{D}$:
   \begin{itemize}
     \item Initilialize ancilla register $\ket{0}_b $. Run $\SNARK.\Verify^{H}(\pi)$ and XOR the output in register $b$.  If output is $0$, output  $\bot$.
       \item  Perform the map $\ket{ct'}\ket{\pi}\otimes \ket{D}\rightarrow \ket{ct'}\ket{\pi}\otimes \ket{D}\ket{ \mathsf{FindWitness}(D)}_{\cal{W}}$.
     \item If $x\leftarrow \mathsf{FindWitness}(D)$ is a not a valid witness, that is, if $P(x)\neq \FHE.\Dec(sk,ct')$, output $\bot$. 
     \item Otherwise, evaluate $P$ on the input from register $\cal{W}$, and XOR in the query output register.
     
    \end{itemize}
    \item Uncompute the $\cal{W}$ register by evaluating $\mathsf{FindWitness}(D)$ and XOR-ing in the $\cal{W}$ register.
    \end{itemize}
   \begin{claim}
       $\Pr[\textbf{Hybrid 2}=1]=\Pr[\textbf{Hybrid 3}=1]$
   \end{claim}

   
    \begin{proof}
        First observe that $\mathsf{FindWitness}$ is essentially the extractor $\cal{E}*$ for the Micali SNARK, with the only (but crucial) difference being that the database $D$ is not measured. Since $\mathsf{Extract}$ aborts if there exists a $u : D(x)=D(x')=u, x\neq x'$, by construction, this implies that, for a fixed database $D$,  $\mathsf{FindWitness}(D)$ has a unique solution. Thus, when all invalid solutions are aborted, it follows from the correctness of $\cal{E}*$ that queries are answered identically in $\textbf{Hybrid 2}$ and  $\textbf{Hybrid 3}$. 
    \end{proof}
\noindent    \textbf{Hybrid 4:} Same as \textbf{Hybrid 3}, but the abort condition for $\mathsf{FindWitness}$ outputting invalid witnesses is removed, and queries to $O[P]$ are answered by $\cal{S}$ as follows: \\
   \begin{itemize}
      \item On query $(ct', \pi)$ to $O[P]$, $\cal{S}$ performs the following operation on the state $\ket{ct'}\ket{\pi}\otimes \ket{D}$:
   \begin{itemize}
    \item Initilialize ancilla register $\ket{0}_b $. Run $\SNARK.\Verify^{H}(\pi)$ and XOR the output in register $b$.  If output is $0$, output  $\bot$.
       \item  Perform the map $\ket{ct'}\ket{\pi}\otimes \ket{D}\rightarrow \ket{ct'}\ket{\pi}\otimes \ket{D}\ket{ \mathsf{FindWitness}(D)}_{\cal{W}}$.
      
     \item Evaluate $P$ on the input from register $\cal{W}$, and XOR in the query output register.
     
    \end{itemize}
    \item Uncompute the $\cal{W}$ register by evaluating $\mathsf{FindWitness}(D)$ and XOR-ing in the $\cal{W}$ register.
    \end{itemize}
    \begin{claim}
        $\Big|\Pr[\textbf{Hybrid 3}=1]-\Pr[\textbf{Hybrid 4}=1]\Big|\leq \text{negl}(\lambda)$
    \end{claim}
    \begin{proof}
    The proof of this claim is identical to Claim \ref{claim:extfailcms}.
    \end{proof}
    \textbf{Hybrid 5}: This is the ideal world.
    \begin{itemize}
        \item Queries to $O[P]$ are answered as in the description of $\cal{S}$ as follows:
    
    \begin{itemize}
        \item Initilialize ancilla register $\ket{0}_b $. Run $\SNARK.\Verify^{H}(\pi)$ and XOR the output in register $b$.  If output is $0$, output  $\bot$.
       \item  Perform the map $\ket{ct'}\ket{\pi}\otimes \ket{D}\rightarrow \ket{ct'}\ket{\pi}\otimes \ket{D}\ket{ \mathsf{FindWitness}(D)}_{\cal{W}}$.
      
     \item Evaluate $P$ on the input from register $\cal{W}$, and XOR in the query output register.
   
   \end{itemize}
   \item Uncompute the $\cal{W}$ register by evaluating $\mathsf{FindWitness}(D)$ and XOR-ing in the $\cal{W}$ register.
   \item  Output $\text{pp} \leftarrow \FHE.\Enc(pk,0^{|P|})$.
    \end{itemize}
    \begin{claim}
        $\Big|\Pr[\textbf{Hybrid 5}=1]-\Pr[\textbf{Hybrid 4}=1]\Big|\leq \text{negl}(\lambda)$
    \end{claim}
    \begin{proof}
This follows directly from the security of the FHE scheme.         
    \end{proof}
\end{proof}
\section{Publicly-verifiable QFHE}\label{sec:pvQFHE}

We define publicly-verifiable QFHE in the oracle model. We consider both leveled and unleveled QFHE.

\subsection{Definition}
\begin{definition}[Publicly-verifiable QFHE]\label{def:PVQFHE}
A leveled publicly-verifiable quantum fully-homomorphic encryption scheme for pseudo-deterministic circuits consists of the following algorithms $(\Gen,\Enc,\Eval,\Ver,\Dec)$.
\begin{itemize}
    \item $\Gen(1^\lambda,D,N,S) \to (\pk, \sk,\PP)$: The key generation algorithm is a PPT algorithm that takes as input the security parameter $1^\secp$, an upper bound on circuit depth $D$, an upper bound on input length $N$, and an upper bound on circuit description size $S$, and returns a public key $\pk$, a secret key $\sk$, and the description of a classical deterministic circuit $\PP$.
    \item $\Enc(\pk, x) \to \ct$: The encryption algorithm is a PPT algorithm that takes as input the public key $\pk$ and a classical plaintext $x$, and outputs a ciphertext $\ct$.
    \item $\Eval^{\PP}(\ct, Q) \to (\widetilde{\ct},\pi)$: The evaluation algorithm is a QPT algorithm that has oracle access to $\PP$, takes as input a ciphertext $\ct$, a quantum circuit $Q$, and outputs an evaluated ciphertext $\widetilde{\ct}$ and proof $\pi$.
    \item $\Ver^{\PP}(\ct,P_Q,\widetilde{\ct},\pi) \to \{\top,\bot\}$: The verification algorithm is a PPT algorithm that has oracle access to $\PP$, takes an input ciphertext $\ct$, a quantum program $P_Q$ (i.e.\ a potentially succinct description of the quantum circuit $Q$), an output ciphertext $\widetilde{\ct}$ and a proof $\pi$, and accepts ($\top$) or rejects ($\bot$).
    \item $\Dec(\sk, \ct) \to x$: The decryption algorithm is a PPT algorithm that takes as input the secret key $\sk$ and a classical ciphertext $\ct$, and outputs plaintext $x$.
\end{itemize}

These algorithms should satisfy the following properties. Recall that pseudo-deterministic circuits $Q$ are defined to have a single bit of output, without loss of generality.

\begin{itemize}
    \item \textbf{Correctness.} For any polynomials $D = D(\secp),N = N(\secp), S = S(\secp)$, any pseudo-deterministic program $P_Q$ where $Q$ has circuit depth at most $D$ and $|P_Q| \leq S$, and any input $x$ with $|x| \leq N$,
    \begin{align*}
        \Pr\left[\begin{array}{l}\Ver^{\PP}(\ct,P_Q,\widetilde{\ct},\pi) = \top ~~ \wedge \\ \Dec(\sk,\widetilde{\ct}) = Q(x) \end{array}: \begin{array}{r}(\pk,\sk,\PP) \gets \Gen(1^\secp,D,N,S) \\ \ct \gets \Enc(\pk,x) \\ (\widetilde{\ct},\pi) \gets \Eval^{\PP}(\ct,Q) \end{array}\right] = 1-\negl(\secp).
    \end{align*}
    \item \textbf{Privacy.} For any QPT adversary $\calA = \{\calA_\secp\}_{\secp \in \bbN}$, polynomials $D=D(\secp),N = N(\secp),S = S(\secp)$, and messages $x_0, x_1$,
    \begin{align*}
    &\bigg| \Pr\left[\calA^\PP(\pk,\ct) = 1 : \begin{array}{r}(\pk,\sk,\PP) \gets \Gen(1^\secp,D,N,S) \\ \ct \gets \Enc(\pk,x_0)  \end{array}\right]\\ &- \Pr\left[\calA^\PP(\pk,\ct) = 1 : \begin{array}{r}(\pk,\sk,\PP) \gets \Gen(1^\secp,D,S) \\ \ct \gets \Enc(\pk,x_1)\end{array}\right]\bigg| = \negl(\secp).
    \end{align*}
    \item \textbf{Soundness.}  For any QPT adversary $\calA = \{\calA_\secp\}_{\secp \in \bbN}$, polynomials $D=D(\secp),N = N(\secp),S = S(\secp)$, pseudo-deterministic program $P_Q$ where $Q$ has circuit depth at most $D$ and $|P_Q| \leq S$, and input $x$ with $|x| \leq N$, 
    \[\Pr\left[\begin{array}{l} \Ver^{\PP}(\ct,P_Q,\widetilde{\ct},\pi) = \top ~~ \wedge \\ \Dec(\sk,\widetilde{\ct}) \neq Q(x)\end{array} : \begin{array}{r}(\pk,\sk,\PP) \gets \Gen(1^\secp,D,N,S) \\ \ct \gets \Enc(\pk,x) \\  (\widetilde{\ct},\pi) \gets \calA^{\PP}(\ct)\end{array}\right] = \negl(\secp).\] 
    \item \textbf{Compactness.} There exists a fixed polynomial $p$ such that the following holds. Fix any polynomials $D = D(\secp), N = N(\secp), S = S(\secp)$, any pseudo-deterministic progam $P_Q$ where $Q$ has circuit depth at most $D$ and $|P_Q| \leq S$, and any $x$ with $|x| \leq N$, and let $(\pk,\sk,\PP) \gets \Gen(1^\secp,D,N,S), \ct \gets \Enc(\pk,x)$, and $(\widetilde{\ct},\pi) \gets \Eval^\PP(\ct,Q)$. Then it holds that 
    \begin{itemize}
        \item $|\pk|,|\sk|,|\PP|,|\ct|,|\widetilde{\ct}|,|\pi| \leq p(\secp,D,N,S)$, and 
        \item the run-times of $\Gen(1^\secp,D,N,S),\Enc(\pk,x),$ $\Ver^\PP(\ct,P_Q,\widetilde{\ct},\pi)$, and $\Dec(\sk,\widetilde{\ct})$ are bounded by $p(\secp,D,N,S)$.
    \end{itemize}

\end{itemize}
\end{definition}

An \emph{unleveled} publicly-verifiable QFHE scheme is defined exactly as above, except that $\Gen$ no longer takes a depth parameter $D$ and all properties are defined without the parameter $D$.

\subsection{Upgradable privately-verifiable QFHE}\label{subsec:priv} 

We will construct publicly-verifiable QFHE by \emph{upgrading} a privately-verifiable scheme that satisfies certain properties. These properties are encapsulated in the following definition.

\begin{definition}[Upgradable privately-verifiable QFHE]\label{def:priv}
    A leveled privately-verifiable QFHE scheme that is amenable to being upgraded to public-verifiability via our compiler has the following syntax. Let $H : \{0,1\}^* \to \{0,1\}$ be a random oracle. Let $\ell(\secp,|Q|)$ be the ``state size'' parameter and $p(\secp,N,D)$ be the ``compactness'' parameter.

\begin{itemize}
    \item $\Gen(1^\secp,D) \to (\pk,\sk)$: The PPT key generation algorithm takes as input the security parameter $1^\secp$ and an upper bound on circuit depth $D$ and outputs a public key $\pk$ and secret key $\sk$.
    \item $\Enc(\pk,x) \to \ct$: The PPT encryption algorithm takes as input a public key $\sk$ and an input $x$ and outputs a ciphertext $\ct$.
    \item $\VerGen(1^\secp,\ct,Q) \to \pp,\sp$: The PPT verification parameter generation algorithm takes as input the security parameter $1^\secp$, a ciphertext $\ct$, and quantum circuit $Q$, and outputs public parameters $\pp$ and secret parameters $\sp$.
    \item $\Eval^H(\pp,\ct,Q) \to \pi$: The QPT evaluation algorithm operates as follows.
    \begin{itemize}
        \item Let $\ell = \ell(\secp,|Q|)$, initialize an $\ell$-qubit state $\ket{\psi_{\ct,Q}}$ on register $\calM = (\calM_1,\dots,\calM_\ell)$, and initialize two fresh registers $\calY,\calZ$. 
        \item Apply a measurement $M_{\pp}$ on $(\calM,\calZ,\calY)$ that is classically-controlled on $\calM$ to obtain a string $y$.
        \item Compute $T = H(y)$, where $T$ is parsed as a subset $T \subseteq [\ell]$.\footnote{In a slight abuse of notation, we define an $\ell$-bit output $H(y) \coloneqq H(1,y)\|\dots\|H(\ell,y)$.}
        \item Apply a measurement $M_{\pp,y,T}$ on $(\calM,\calZ)$ that is classically-controlled on $\calM$ to obtain a string $z$.
        \item For all $i \in [\ell]$, measure $\calM_i$ in the standard basis if $i \in T$ or in the Hadamard basis if $i \notin T$ to obtain a string $m \in \{0,1\}^\ell$.
        \item Return $\pi = (m,y,z)$.
    \end{itemize}
    \item $\Ver^H(\sp,\ct,Q,\pi) \to \{\bot,\widetilde{\ct}\}$: The PPT verification algorithm takes as input secret parameters $\sp$, a ciphertext $\ct$, a quantum circuit $Q$, and a proof $\pi$, and outputs either reject $(\bot)$ or a ciphertext $\widetilde{\ct}$.
    \item $\Dec(\sk,\ct) \to x$: The PPT decryption algorithm takes as input the secret key $\sk$ and a ciphertext $\ct$ and outputs a plaintext $x$.
\end{itemize}

It should satisfy the following properties. We leave probabilities taken over the sampling of the random oracle $H$ implicit.

\begin{itemize}
    \item \textbf{Privacy.} For any QPT adversary $\{\calA_\secp\}_{\secp}$, depth $D = D(\secp)$, and messages $x_0,x_1,$
    \begin{align*}
    &\bigg| \Pr\left[\calA(\pk,\ct) = 1 : \begin{array}{r}(\pk,\sk) \gets \Gen(1^\secp,D) \\ \ct \gets \Enc(\pk,x_0)  \end{array}\right]\\ &- \Pr\left[\calA(\pk,\ct) = 1 : \begin{array}{r}(\pk,\sk) \gets \Gen(1^\secp,D) \\ \ct \gets \Enc(\pk,x_1)\end{array}\right]\bigg| = \negl(\secp).
    \end{align*}
    \item \textbf{Completeness.} For any $(\pk,\sk) \in \Gen(1^\secp,D)$, $\ct \in \Enc(\pk,x)$, and $Q$,
    \[\Pr\left[\Dec(\sk,\Ver^H(\sp,\ct,Q,\pi)) = Q(x) : \begin{array}{r}\pp,\sp \gets \VerGen(\ct,Q) \\ \pi \gets \Eval^H(\pp,\ct,Q)\end{array}\right] = 1-\negl(\secp).\]
    \item \textbf{Soundness.} For any $(\pk,\sk) \in \Gen(1^\secp,D)$, $\ct \in \Enc(\pk,x)$, $Q$, and QPT adversary $\{\calA_\secp\}_\secp$, 
    \[\Pr\left[\widetilde{\ct} \neq \bot \wedge \Dec(\sk,\widetilde{\ct}) \neq Q(x) : \begin{array}{r} \pp,\sp \gets \VerGen(\ct,Q) \\ \pi \gets \calA^H(\pp) \\ \widetilde{\ct} \gets \Ver^H(\sp,\ct,Q,\pi)\end{array}\right] = \negl(\secp).\]
        \item \textbf{Incompatible standard-basis measurements.} First, we introduce some notation.

    \begin{itemize}
        \item Given a subset $I \subset [\ell]$, define $\overline{I} \coloneqq [\ell] \setminus I$.
        \item Given any $m \in \{0,1\}^\ell$ and a subset $I \subseteq [\ell]$, we define $m_{I \to *} \in \{0,1,*\}^\ell$ to be equal to $m_i$ at indices $i \in \overline{I}$ and equal to $*$ at indices $i \in I$.
        \ \item Each choice of $(\cdot,\sp) \in \VerGen(\ct,Q)$ defines\footnote{The way that $S$ is defined based on $\sp$ is specified in \cite[Section 5.3]{bartusek2023obfuscation}, though the details will not be important to us here.} a set of ``standard basis positions'' $S \subseteq [\ell]$, which we denote by $S = S[\sp]$ (not to be confused with the positions $T \subseteq [\ell]$ that are \emph{measured} in the standard basis). 
        \item We say that two sets of strings $M_0,M_1 \in \{0,1,*\}^{\ell}$ are \emph{incompatible} if for any $m_0 \in M_0, m_1 \in M_1$, there exists an $i \in [\ell]$ such that $m_{0,i},m_{1,i} \neq *$ and $m_{0,i} \neq m_{1,i}$. 
        \item Each choice of $(\cdot,\sp) \in \VerGen(\ct,Q)$ defines\footnote{As above, this definition is specified in \cite[Section 5.3]{bartusek2023obfuscation}, and the details will not be important to us here. } two incompatible sets of strings $M_0,M_1 \in \{0,1,*\}^{\ell}$.
    \end{itemize}

    Now, we need the following properties.

    \begin{itemize}
        \item For any $(\pk,\sk) \in \Gen(1^\secp,D)$, $\ct \in \Enc(\pk,x)$, Q, $(\pp,\sp) \in \VerGen(\ct,Q)$, and proof $\pi = (m,y,z)$, it holds that 
        \[\Ver^H(\sp,\ct,Q,(m,y,z)) = \Ver^H(\sp,\ct,Q,(m_{S \cap \overline{T} \to *},y,z)).\]  That is, $\Ver$ ignores all bits of $m$ that are measured in the Hadamard basis, but correspond to ``standard basis positions''.


        \item For any $(\pk,\sk) \in \Gen(1^\secp,D)$, $\ct \in \Enc(\pk,x)$, $Q$, and QPT adversary $\{\calA_\secp\}_\secp$, 
    \[\Pr\left[\widetilde{\ct} \neq \bot \wedge m_{\overline{S} \cup \overline{T} \to *} \notin M_0: \begin{array}{r} \pp,\sp \gets \VerGen(\ct,Q) \\ \pi = (m,y,z) \gets \calA^H(\pp) \\ T \coloneqq H(y) \\ \widetilde{\ct} \gets \Ver^H(\sp,\ct,Q,\pi)\end{array}\right] = \negl(\secp).\] 
    That is, it is hard to find an accepting proof $\pi$ such that $m_{\overline{S} \cap \overline{T} \to *}$ (where only the standard basis positions that are measured in the standard basis remain) is in $M_0$.
    \item For any $(\pk,\sk) \in \Gen(1^\secp,D)$, $\ct \in \Enc(\pk,x)$, $Q$, and QPT adversary $\{\calA_\secp\}_\secp$,
    \[\Pr\left[\widetilde{\ct} \neq \bot \wedge \Dec(\sk,\widetilde{\ct}) \neq Q(x) \wedge m_{\overline{S} \cup \overline{T} \to *} \notin M_1 : \begin{array}{r} \pp,\sp \gets \VerGen(\ct,Q) \\ \pi = (m,y,z) \gets \calA^H(\pp,\sp) \\ T \coloneqq H(y) \\ \widetilde{\ct} \gets \Ver^H(\sp,\ct,Q,\pi)\end{array}\right] = \negl(\secp).\]
    That is, it is hard to find an accepting but faulty proof $\pi$ such that $m_{\overline{S} \cap \overline{T} \to *}$ (where only the standard basis positions that are measured in the standard basis remain) is in $M_1$, \emph{even given the secret parameters} $\sp$.
    \end{itemize}

    \item \textbf{Parallelizable $\VerGen$.} There exists a fixed polynomial $k(\secp)$ and a deterministic classical algorithm $\ParVerGen(\ct,Q,\sp,i) \to \pp_i$, where $|\pp_i| = k(\secp)$, such that for any $\ct,Q$, the operation $\VerGen(1^\secp,\ct,Q)$ is equivalent to:
    \begin{itemize}
        \item Sample $\sp \gets \{0,1\}^\secp$
        \item Run $\pp_i \coloneqq \ParVerGen(\ct,Q,\sp,i)$ for $i \in [\ell]$, and set $\pp \coloneqq (\pp_1,\dots,\pp_\ell)$.
    \end{itemize}

    \item \textbf{Compactness.} Fix any polynomials $N = N(\secp), D = D(\secp)$, any input $x$ such that $|x| \leq N$, and any pseudo-deterministic quantum circuit $Q$ with depth at most $D$, and let $(\pk,\sk) \gets \Gen(1^\secp,D),\ct \gets \Enc(\pk,x), (\pp,\sp) \gets \VerGen(\ct,Q), \pi \gets \Eval^H(\pp,\ct,Q), \widetilde{\ct} \gets \Ver^H(\sp,\ct,Q,\pi)$, and $Q(x) \gets \Dec(\sk,\ct)$. Then it holds that 
    \begin{itemize}
        \item $|\pk|,|\sk|,|\ct|,|\widetilde{\ct}| \leq p(\secp,N,D)$, and 
        \item The run-times of $\Gen(1^\secp,D),\Enc(\pk,x),\Dec(\sk,\ct)$ are bounded by $p(\secp,N,D)$.
    \end{itemize}
  
\end{itemize}

\end{definition}

An \emph{unleveled} upgradable privately-verifiable QFHE scheme is defined exactly as above, except that $\Gen$ no longer takes the depth parameter $D$ and all properties are defined without the parameter $D$.

\begin{theorem}[\cite{bartusek2023obfuscation}, Section 5.3]
    Assuming LWE, there exists a leveled upgradable privately-verifiable QFHE in the QROM. Assuming LWE plus an appropriate circular-security assumption \cite{bra18}, there exists unleveled upgradable privately-verifiable QFHE in the QROM.
\end{theorem}

\subsection{Construction}
\newpage

\begin{figure}[H]
\centering
\fbox{%
\parbox{0.96\textwidth}{%
\textbf{Publicly-verifiable QFHE}
\begin{itemize}
    \item $\Gen(1^\secp,N,D,S)$:
    \begin{itemize}
        \item Sample $\pk,\sk \gets \Priv.\Gen(1^\secp)$.
        \item Sample PRF keys $k_\PFC,k_\Priv,k_H \gets \{0,1\}^\secp$.
        \item Sample $\O \gets \OSS.\Setup(1^\secp).$
        \item Define $\CK(\pk_\OSS,i,v)$ as follows.
        \begin{itemize}
            \item Compute $(\CK_{\pk,i},\dk_{\pk,i}) \coloneqq \PFC.\Gen(1^\secp,F_{k_\PFC}(\pk_\OSS,i))$.
            \item Output $\CK_{\pk,i}(v)$. 
        \end{itemize} 
        \item Define $\PrivGen(\ct,Q,\pk_\OSS,c,\sigma,i)$ as follows.
        \begin{itemize}
            \item If $\OSS.\Ver^\O(\pk_\OSS,(\ct,Q,c),\sigma) = \top$, then continue, and otherwise return $\bot$.
            \item Define $\sp \coloneqq F_{k_\Priv}(\ct,Q,\pk_\OSS,c,\sigma)$.
            \item Output $\pp_i \coloneqq \Priv.\ParVerGen(\ct,Q,\sp,i)$.
        \end{itemize}
        \item Let $\H\coloneqq F_{k_H}$.
        \item Define $\PrivVer(\ct,Q,\pk_\OSS,c,\sigma,u,y,z)$ as follows.
        \begin{itemize}
            \item If $\OSS.\Ver^\O(\pk_\OSS,(\ct,Q,c),\sigma) = \top$, then set $\sp \coloneqq F_{k_\Priv}(\ct,Q,\pk_\OSS,c,\sigma)$, $S \coloneqq S[\sp]$, and otherwise return $\bot$.
            \item Parse $u = (u_1,\dots,u_\ell)$, $c = (c_1,\dots,c_\ell)$, and define $T \coloneqq \H(y)$.
            \item For each $i \in [\ell]$, compute $(\CK_{\pk,i},\dk_{\pk,i}) \coloneqq \PFC.\Gen(1^\secp,F_{k_\PFC}(\pk_\OSS,i))$.
            \item For each $i \in T$, compute $m_i \coloneqq \PFC.\DecZ(\dk_{\pk,i},c_i,u_i)$, and return $\bot$ if $m_i = \bot$.
            \item For each $i \in \overline{T} \setminus S$, compute $m_i \coloneqq \PFC.\DecX(\dk_{\pk,i},c_i,u_i)$, and return $\bot$ if $m_i = \bot$.
            \item For each $i \in \overline{T} \cap S$, set $m_i = *$.
            \item Let $m = (m_1,\dots,m_\ell)$, and output $\Priv.\Ver^\H(\sp,\ct,Q,(m,y,z))$.
        \end{itemize}
        \item Output $\pk,\sk,\PP \coloneqq \SObf.\Obf(1^\secp,(\O,\CK,\PrivGen,\H,\PrivVer),T)$, where $T$ is the size of the circuit required to implement the concatenation of programs $(\O,\CK,\PrivGen,\H,\PrivVer)$. 
    \end{itemize}
    \item $\Enc(\pk,x)$: Output $\ct \gets \Priv.\Enc(\pk,x)$.
    \item $\Eval^\PP(\ct,Q)$
    \begin{itemize}
        \item Let $\ell = \ell(\secp,|Q|)$ and sample $(\pk_\OSS,\ket{\sk_\OSS}) \gets \OSS.\Gen^\O(|\ct| + |Q| + \ell \cdot t(\secp))$.
        \item Initialize $\ket{\psi_{\ct,Q}}$ on register $\calM = (\calM_1,\dots,\calM_\ell)$.
        \item For each $i \in [\ell]$ apply $\PFC.\Com^{\CK(\pk_\OSS,i,\cdot)}: \calM_i \to \calM_i,\calU_i,c_i$, and set $c \coloneqq (c_1,\dots,c_\ell)$.
        \item Compute $\sigma \gets \OSS.\Sign^\O(\ket{\sk_\OSS},(\ct,Q,c))$ and, for $i \in [\ell]$, compute $\pp_i \gets \PrivGen(\pk_\OSS,\ct,Q,c,\sigma,i)$. Set $\pp \coloneqq (\pp_1,\dots,\pp_\ell)$.
        \item Initialize fresh registers $\calY,\calZ$, apply $M_\pp$ on $(\calM,\calY,\calZ)$ to obtain $y$, compute $T = \H(y)$, and apply $M_{\pp,y,T}$ on $(\calM,\calZ)$ to obtain $z$.
        \item For each $i \in T$, compute $u_i \gets \PFC.\OpenZ(\calM_i,\calU_i)$ and for each $i \in \overline{T}$, compute $u_i \gets \PFC.\OpenX(\calM_i,\calU_i)$. Set $u \coloneqq (u_1,\dots,u_\ell)$.
        \item Output $\widetilde{\ct} \coloneqq \PrivVer(\ct,Q,\pk_\OSS,c,\sigma,u,y,z)$ and $\pi \coloneqq (\pk_\OSS,c,\sigma,u,y,z)$. 
    \end{itemize}
    \item $\Ver^\PP(\ct,Q,\widetilde{\ct},\pi)$: Output $\top$ iff $\PrivVer(\ct,Q,\pi) = \widetilde{\ct}$.
    \item $\Dec(\sk,\ct)$: Output $\Priv.\Dec(\sk,\ct)$.
\end{itemize}
}
}
\caption{Our construction of publicly-verifiable QFHE in the classical oracle model.}\label{fig:QFHE}
\end{figure}

Ingredients:
\begin{itemize}
    \item PRF $F_k: \{0,1\}^* \to \{0,1\}^\secp$.
    \item OSS $(\OSS.\Setup, \OSS.\Gen,\OSS.\Sign, \OSS.\Ver)$ in the classical oracle model (Definition \ref{def:OSS}). 
    \item PFC $(\PFC.\Gen, \PFC.\Com,\PFC.\OpenZ,\PFC.\OpenZ,\PFC.\DecZ,\PFC.\DecX)$ in the classical oracle model (see Definition \ref{def:PFC}), where single-qubit commitments are size $t(\secp)$.
    \item Upgradable privately-verifiable QFHE $(\Priv.\Gen, \Priv.\Enc,\Priv.\VerGen,\Priv.\Eval^H,\Priv.\Ver^H,\Priv.\Dec)$, in the QROM with state size parameter $\ell(\secp,|Q|)$ and compactness parameter $p(\secp,N,D)$ (see Definition \ref{def:priv}).
    \item Post-quantum succinct ideal obfuscation $(\SObf.\Obf,\SObf.\Eval)$ in the classical oracle model (see Definition \ref{def:succinct-obf}). In an abuse of notation, we will apply $\SObf.\Obf$ to a program with multiple bits of output, meaning that we break the program into multiple single-bit output programs, and obfuscate each.
\end{itemize}

\subsection{Proofs}

\begin{theorem}\label{thm:QFHE-correctness}
    The QFHE protocol given in Figure \ref{fig:QFHE} satisfies correctness (see Definition \ref{def:PVQFHE}).
\end{theorem}

\begin{proof}
    This is a straightforward adaptation of the completeness part of \cite[Theorem 5.12]{bartusek2023obfuscation}, where we use the correctness of $\OSS$ (along with the other primitives) rather than the signature token.
\end{proof}

\begin{theorem}\label{thm:QFHE-privacy}
    The QFHE protocol given in Figure \ref{fig:QFHE} satisfies privacy (see Definition \ref{def:PVQFHE}).
\end{theorem}

\begin{proof}
    This follows directly from the privacy of the upgradable privately-verifiable QFHE.
\end{proof}

\begin{theorem}\label{thm:QFHE-privacy}
    The QFHE protocol given in Figure \ref{fig:QFHE} satisfies compactness (see Definition \ref{def:PVQFHE}).
\end{theorem}
    
\begin{proof}
    This follows directly from the compactness of the upgradable privately-verifiable QFHE and succinctness of the post-quantum succinct ideal obfuscator. In particular, note that the output size of the program obfuscated by $\Gen$ is bounded by a fixed polynomial, and so $\PP$ consists of a bounded polynomial number of obfuscated programs (each with bounded circuit size).
\end{proof}

\begin{theorem}\label{thm:QFHE-security}
    The QFHE protocol given in Figure \ref{fig:QFHE} satisfies soundness (see Definition \ref{def:PVQFHE}).
\end{theorem}

\begin{proof}
    We follow the soundness part of the proof of \cite[Theorem 5.12]{bartusek2023obfuscation}, but with a couple of tweaks.

    Suppose there exists a program $P_{Q^*}$, input $x$, and adversary $\calA_1^{\PP}$ that violates the soundness of Definition \ref{def:PVQFHE}, where we have dropped the indexing by $\secp$ for convenience. Our first step will be to appeal to the security of the succinct obfuscation (Definition \ref{def:succinct-obf}) to replace $\calA_1$ with an adversary $\calA_2^{\O,\CK,\PrivGen,\H,\PrivVer}$ that has direct access to the large-input functionalities. This only has a negligble affect on the output of $\calA_1$.

    Next, we replace the PRFs $F_{k_\PFC},F_{k_\Priv},F_{k_H}$ with random oracles $\H_\PFC,\H_\Priv,\H$. Since $\calA_2$ only has polynomially-bounded oracle access to these functionalities, this has a negligible effect on the output of $\calA_2$ \cite{6375347}. $\calA_2$ can then be used to define an oracle algorithm $\calA_3^{\H_\Priv}$ that operates as follows.

\begin{itemize}
    \item Sample $(\pk,\sk,\O,\CK,\PrivGen,\H,\PrivVer)$ as in $\Gen(1^\secp,N,D,S)$, except that $\H_\PFC$ and $\H$ are sampled as random oracles.
    \item Sample $\ct^* \gets \Priv.\Enc(\pk,x)$.
    \item Run $\calA_2^{\O,\CK,\PrivGen,\H,\PrivVer}(\ct^*)$, forwarding calls to $\H_\Priv$ (which occur as part of calls to $\PrivGen$ and $\PrivVer$) to an external random oracle $\H_\Priv$.
    \item Measure $\calA_2$'s output proof $\pi^*$, parse $\pi^*$ as $(\pk_\OSS^*,c^*,\sigma^*,m^*,y^*,z^*)$ and output $a \coloneqq (\ct^*,Q^*,\pk_\OSS^*,c^*,\sigma^*)$ and $\aux \coloneqq (m^*,y^*,z^*)$.
\end{itemize}

Note that $\calA_3$ makes $p = \poly(\secp)$ total queries to $\H_\Priv$. Now, define $V$ as in Figure \ref{fig:functionalities}. Then since $\calA_1$ breaks soundness, \[\Pr\left[V(a,\H_\Priv(a),\aux) = 1 : (a,\aux) \gets \calA_3^{\H_\Priv}\right] = \nonnegl(\secp).\]

Next, since $p = \poly(\secp)$, by Theorem \ref{thm:measure-and-reprogram} there exists an algorithm $\calA_4 \coloneqq \Sim[\calA_3]$ such that

\[\Pr\left[V((\ct^*,Q^*,\pk_\OSS^*,c^*,\sigma^*),s,(m^*,y^*,z^*)) = 1 : \begin{array}{r} ((\ct^*,Q^*,\pk_\OSS^*,c^*,\sigma^*),\state) \gets \calA_4 \\ s \gets \{0,1\}^\secp \\ (m^*,y^*,z^*) \gets \calA_4(s,\state)\end{array}\right] = \nonnegl(\secp).\]

Moreover, $\calA_4$ operates as follows.

\begin{itemize}
    \item Sample $\H_\Priv$ as a $2p$-wise independent function and $(i,d) \gets (\{0,\dots,p-1\} \times \{0,1\}) \cup \{(p,0)\}$.
    \item Run $\calA_3$ for $i$ oracle queries, answering each query using the function $\H_\Priv$. 
    \item When $\calA_3$ is about to make its $(i+1)$'th oracle query, measure its query register in the standard basis to obtain $a \coloneqq (\ct^*,Q^*,\pk_\OSS^*,c^*,\sigma^*)$. In the special case that $(i,d) = (p,0)$, just measure (part of) the final output register of $\calA_3$ to obtain $a$.
    \item Receive $s$ externally.
    \item If $d = 0$, answer $\calA_3$'s $(i+1)$'th query with $\H_\Priv$. If $d=1$, answer $\calA_3$'s $(i+1)$'th query instead with the re-programmed oracle $\H_\Priv[(\ct^*,Q^*,\pk_\OSS^*,c^*,\sigma^*) \to s]$.
    \item Run $\calA_3$ until it has made all $p$ queries to $\H_\Priv$. For queries $i+2$ through $p$, answer with the re-programmed oracle $\H_\Priv[(\ct^*,Q^*,\pk_\OSS^*,c^*,\sigma^*) \to s]$.
    \item Measure $\calA_3$'s output $\aux \coloneqq (u^*,y^*,z^*)$.
\end{itemize}

Recall that $\calA_3$ is internally running $\calA_2$, who expects oracle access to $\O,\CK,\PrivGen,\H,\PrivVer$. These oracle queries will be answered by $\calA_4$. Next, we define $\calA_5$ to be the same as $\calA_4$, except that after $(\ct^*,Q^*,\pk_\OSS^*,c^*,\sigma^*)$ is measured by $\calA_4$, $\calA_2$'s queries to $\PrivVer$ are answered instead with the functionality $\PrivVer[\pk_\OSS^*,s,(\ct^*,Q^*,\pk_\OSS^*,c^*,\sigma^*)]$ from Figure \ref{fig:functionalities}. 

\begin{claim}\label{claim:A4}
\[\Pr\left[V((\ct^*,Q^*,\pk_\OSS^*,c^*,\sigma^*),s,(m^*,y^*,z^*)) = 1 : \begin{array}{r} ((\ct^*,Q^*,\pk_\OSS^*,c^*,\sigma^*),\state) \gets \calA_5 \\ s \gets \{0,1\}^\secp \\ (m^*,y^*,z^*) \gets \calA_5(s,\state)\end{array}\right] = \nonnegl(\secp).\]
\end{claim}

\begin{proof}
We can condition on $\OSS.\Ver^\O(\pk_\OSS^*,(\ct^*,Q^*,c^*),\sigma^*) = \top$, since otherwise $V$ would output 0. Then, by the security of $\OSS$ (Definition \ref{def:OSS}), once $(\pk_\OSS^*,\ct^*,Q^*,c^*,\sigma^*)$ is measured, $\calA_2$ cannot produce any query that has noticeable amplitude on any $(\ct,Q,c,\sigma)$ such that \[(\ct,Q,c,\sigma) \neq (\ct^*,Q^*,c^*,\sigma^*) ~~ \text{and} ~~ \OSS.\Ver(\pk_\OSS^*,(\ct,Q,c),\sigma) = \top.\] But after $(\pk_\OSS^*,\ct^*,Q^*,c^*,\sigma^*)$ is measured and $s$ is sampled, $\PrivVer$ and $\PrivVer[\pk_\OSS^*,s,(\ct^*,Q^*,c^*,\sigma^*)]$ can only differ on $(\ct,Q,c,\sigma)$ such that \[(\ct,Q,c,\sigma) \neq (\ct^*,Q^*,c^*,\sigma^*) ~~ \text{and} ~~ \OSS.\Ver(\pk_\OSS^*,(\ct,Q,c),\sigma) = \top.\] Thus, since $\calA_2$ only has polynomially-many queries, changing the oracle in this way can only have a negligible effect on the final probability, which completes the proof.
\end{proof}

Next, we claim the following, where $V[1]$ is defined in Figure \ref{fig:functionalities}.

\begin{claim}\label{claim:hybrid0}
\[\Pr\left[V[1]((\ct^*,Q^*,\pk_\OSS^*,c^*,\sigma^*),s,(m^*,y^*,z^*)) = 1 : \begin{array}{r} ((\ct^*,Q^*,\pk_\OSS^*,c^*,\sigma^*),\state) \gets \calA_5 \\ s \gets \{0,1\}^\secp \\ (m^*,y^*,z^*) \gets \calA_5(s,\state)\end{array}\right] = \nonnegl(\secp).\]
\end{claim}

\begin{proof}
This is a straightforward adaptation of \cite[Claim 5.14]{bartusek2023obfuscation}, reducing to the ``Incompatible standard-basis measurements'' property of $\Priv$ (Definition \ref{def:priv}).


\end{proof}

Finally, we will define a sequence of hybrids $\{\calH_\iota\}_{\iota \in [0,p]}$ based on $\calA_5$. Hybrid $\calH_\iota$ is defined as follows.

\begin{itemize}
    \item Run $((\ct^*,Q^*,\pk_\OSS^*,c^*,\sigma^*),\state) \gets \calA_5$.
    \item Sample $s \gets \{0,1\}^\secp$.
    \item Run $(m^*,y^*,z^*) \gets \calA_5(s,\state)$ with the following difference. Recall that at some point, $\calA_5$ begins using the oracle $\H_\Priv[(\ct^*,Q^*,\pk_\OSS^*,c^*,\sigma^*) \to s]$ while answering $\calA_2$'s queries. For the first $\iota$ times that $\calA_2$ queries $\PrivVer[\pk_\OSS^*,s,(\ct^*,Q^*,c^*,\sigma^*)]$ after this point, respond using the oracle $\PrivVer[\pk_\OSS^* \to \bot]$ that outputs $\bot$ on every input that includes $\pk_\OSS^*$ (and behaves normally otherwise).
    \item Output $V[1]((\ct^*,Q^*,\pk_\OSS^*,c^*,\sigma^*),s,(m^*,y^*,z^*))$.
\end{itemize}

Note that Claim \ref{claim:hybrid0} is stating exactly that $\Pr[\calH_0 = 1] = \nonnegl(\secp)$. Next, we have the following claim.

\begin{claim}\label{claim:hybridq}
$\Pr[\calH_p = 1] = \negl(\secp)$.
\end{claim}

\begin{proof}
This is a straightforward adaptation of \cite[Claim 5.15]{bartusek2023obfuscation}, reducing to the soundness of $\Priv$ (Definition \ref{def:priv}).

\end{proof}

Finally, we prove the following Claim \ref{lemma:intermediate-hybrids}. Since $p = \poly(\secp)$, this contradicts Claim \ref{claim:hybrid0} and Claim \ref{claim:hybridq}, which completes the proof.

\end{proof}

\begin{claim}\label{lemma:intermediate-hybrids}
For any $\iota \in [p]$, $\Pr[\calH_\iota = 1] \geq \Pr[\calH_{\iota - 1} = 1] - \negl(\secp)$.
\end{claim}

\begin{proof}

Throughout this proof, when we refer to ``query $\iota$'' in some hybrid, we mean the $\iota$'th query that $\calA_2$ makes to $\PrivVer[\pk_\OSS^*,s,(\ct^*,Q^*,c^*,\sigma^*)]$ after $\calA_5$ has begun using the oracle $\H_\Priv[(\ct^*,Q^*,\pk_\OSS^*,c^*,\sigma^*) \to s]$ (if such a query exists).

Now, we introduce an intermediate hybrid $\calH_{\iota-1}'$ which is the same as $\calH_{\iota-1}$ except that query $\iota$ is answered with the functionality $\PrivVer[\pk_\OSS^*,s,(\ct^*,Q^*,c^*,\sigma^*),M_0]$ defined in Figure \ref{fig:functionalities}.

So, it suffices to show that 
\begin{itemize}
    \item $\Pr[\calH_{\iota-1}' = 1] \geq \Pr[\calH_{\iota-1} = 1] - \negl(\secp)$, and 
    \item $\Pr[\calH_{\iota} = 1] \geq \Pr[\calH_{\iota-1}' = 1] - \negl(\secp)$.
\end{itemize}

We note that the only difference between the three hybrids is how query $\iota$ is answered:
\begin{itemize}
    \item In $\calH_{\iota - 1}$, query $\iota$ is answered with $\PrivVer[\pk_\OSS^*,s,(\ct^*,Q^*,c^*,\sigma^*)]$.
    \item In $\calH_{\iota - 1}'$, query $\iota$ is answered with $\PrivVer[\pk_\OSS^*,s,(\ct^*,Q^*,c^*,\sigma^*),M_0]$.
    \item In $\calH_\iota$, query $\iota$ is answered with $\PrivVer[\pk^*_\OSS \to \bot]$.
\end{itemize}

Now, the proof is completed by appealing to the following two claims.
\end{proof}

\begin{claim}
$\Pr[\calH_{\iota-1}' = 1] \geq \Pr[\calH_{\iota-1} = 1] - \negl(\secp).$
\end{claim}

\begin{proof}

This is a straightforward adaptation of \cite[Claim 5.17]{bartusek2023obfuscation}, reducing to the ``Incompatible standard-basis measurements'' property of $\Priv$ (Definition \ref{def:priv}).


\end{proof}

\begin{claim}
$\Pr[\calH_{\iota} = 1] \geq \Pr[\calH_{\iota-1}' = 1] - \negl(\secp)$
\end{claim}

\begin{proof}

This is a straightforward adaptation of \cite[Claim 5.18]{bartusek2023obfuscation}, reducing to the string binding with public decodability of $\PFC$ (Definition \ref{def:string-binding}).

\end{proof}

\begin{figure}[H]
\centering
\fbox{%
\parbox{0.96\textwidth}{
\textbf{Functionalities used in the proof of Theorem \ref{thm:QFHE-security}}

\begin{itemize}

    \item $\PrivVer[\pk_\OSS^*,s](\ct,Q,\pk_\OSS,c,\sigma,u,y,z)$
    \begin{itemize}
        \item Same as $\PrivVer$ except that it sets $\sp = s$ if $\pk_\OSS = \pk_\OSS^*$.
    \end{itemize}
    
    \item $\PrivVer[\pk_\OSS^*,s,(\ct^*,Q^*,c^*,\sigma^*)](\ct,Q,\pk_\OSS,c,\sigma,u,y,z)$
    \begin{itemize}
        \item Same as $\PrivVer[\pk_\OSS^*,s]$ except output $\bot$ if $\pk_\OSS = \pk_\OSS^* \wedge (\ct,Q,c,\sigma) \neq (\ct^*,Q^*,c^*,\sigma^*).$
    \end{itemize}
    
    \item $\PrivVer[\pk_\OSS^*,s,(\ct^*,Q^*,c^*,\sigma^*),0](\ct,Q,\pk_\OSS,c,\sigma,u,y,z)$
    \begin{itemize}
        \item Same as $\PrivVer[\pk_\OSS^*,s,(\ct^*,Q^*,c^*,\sigma^*)]$ except output $\bot$ if $\pk_\OSS = \pk_\OSS^* \wedge m_{\overline{S} \cup \overline{T}} \notin M_0$, where $S,M_0$ are defined based on $\sp = s$ as in Definition \ref{def:priv}.
    \end{itemize}

    \item $\PrivVer[\pk_\OSS^*,s,(\ct^*,Q^*,c^*,\sigma^*),1](\ct,Q,\pk_\OSS,c,\sigma,u,y,z)$
    \begin{itemize}
        \item Same as $\PrivVer[\pk_\OSS^*,s,(\ct^*,Q^*,c^*,\sigma^*)]$ except output $\bot$ if $\pk_\OSS = \pk_\OSS^* \wedge m_{\overline{S} \cup \overline{T}} \notin M_1$, where $S,M_1$ are defined based on $\sp = s$ as in Definition \ref{def:priv}.
    \end{itemize}
    

    \item $V(a,s,\aux)$:
    \begin{itemize}
        \item Parse $a \coloneqq (\ct^*,Q^*,\pk_\OSS^*,c^*,\sigma^*)$ and $\aux \coloneqq (u^*,y^*,z^*)$.
        \item Compute $\widetilde{\ct} \coloneqq \PrivVer[\pk_\OSS^*,s](\ct^*,Q^*,\pk_\OSS^*,c^*,\sigma^*,u^*,y^*,z^*)$.
        \item Output 1 iff $\widetilde{\ct} \neq \bot$ and $\Priv.\Dec(\sk,\widetilde{\ct}) \neq Q^*(x)$.
    \end{itemize}
    
    \item $V[1](a,s,\aux)$:
    \begin{itemize}
        \item Parse $a \coloneqq (\ct^*,Q^*,\pk_\OSS^*,c^*,\sigma^*)$ and $\aux \coloneqq (u^*,y^*,z^*)$.
        \item Compute $\widetilde{\ct} \coloneqq \PrivVer[\pk_\OSS^*,s,(\ct^*,Q^*,c^*,\sigma^*),1](\ct^*,Q^*,(\pk_\OSS^*,c^*,\sigma^*,u^*,y^*,z^*))$.
         \item Output 1 iff $\widetilde{\ct} \neq \bot$ and $\Priv.\Dec(\sk,\widetilde{\ct}) \neq Q^*(x)$.
    \end{itemize}

\end{itemize}
}}
\caption{Helper functionalities for the proof of Theorem \ref{thm:QFHE-security}.}\label{fig:functionalities}
\end{figure}
\section{(Succinct) Classical Obfuscation of Pseudo-deterministic Quantum Circuits}\label{sec:obfuscation}

\begin{definition}[Classical Ideal obfuscation for Pseudodeterministic Quantum Circuits]\label{def:classical-obf-pd-circuits}
  A classical ideal obfuscator for a family of pseudo-deterministic quantum circuits is a pair of QPT algorithms $(\Obf,\Eval)$ with the following syntax.
  \begin{itemize}
      \item $\Obf(1^\secp, P_Q) \to P_{\widetilde{Q}}$: $\Obf$ is a classical p.p.t.~algorithm that takes as input the security parameter $1^\secp$ and the classical description $P_Q$ of a quantum circuit $Q$, and outputs the classical description of an obfuscated circuit $P_{\widetilde{Q}}$.
      \item $\Eval(P_{\widetilde{Q}},x) \to b$: $\Eval$ is a (quantum) polynomial-time algorithm that takes as input the description $P_{\widetilde{Q}}$ of an obfuscated circuit $\widetilde{Q}$ and an input $x$, and outputs a bit $b \in \{0,1\}$.
  \end{itemize}
  A VBB obfuscator should satisfy the following properties for any pseudo-deterministic family of circuits $Q = \{Q_\secp\}_{\secp \in \bbN}$ with input length $n = n(\secp)$.
  \begin{itemize}
      \item \textbf{Correctness}: It holds with probability $1-\negl(\secp)$ over $\widetilde{Q} \gets \Obf(1^\secp, P_Q)$ that for all $x \in \{0,1\}^n$, $\Pr[\Eval(P_{\widetilde{Q}},x) \to Q(x)] = 1-\negl(\secp)$.
      \item \textbf{Security}: There exists a QPT simulator $\{\calS_\secp\}_{\secp \in \bbN}$, such that for any QPT adversary $\{\calA_\secp\}_{\secp \in \bbN}$, pseudodeterministic programs $P_Q$,
\begin{align*}
    \left\{\calA_\secp\left(\Obf(1^\secp, P_Q)\right)\right\} \approx_c \calS_\secp^{O[Q]}
\end{align*}
where $O[Q]$ is the oracle that computes the map $x \to Q(x)$.
  \end{itemize}
  In addition, the ideal obfuscator could optionally satisfy a succinctness criterion.
  \begin{itemize}
      \item \textbf{Succinctness}: The size of the obfuscation $P_{\widetilde{Q}} \leftarrow \Obf(1^\secp, P_Q)$ is polynomially bounded by the size of the description $P_Q$ of the input circuit, rather than the size of the circuit (in terms of number of gates) itself. That is, $|P_{\widetilde{Q}}| \le \poly(\lambda, |P_Q|)$, rather than $|P_{\widetilde{Q}}| \le \poly(\lambda, |Q|)$.
  \end{itemize}
\end{definition}

\subsection{Construction}\label{sec:construct-obf}
\begin{itemize}
    \item A publicly-verifiable QFHE for pseudo-deterministic circuits in the oracle model $\pvQFHE = (\Gen,\allowbreak\Enc,\allowbreak\Eval,\allowbreak\Ver,\allowbreak\Dec)$ (\Cref{def:PVQFHE}).
    \item A post-quantum input-succinct obfuscator $(\SObf, \SEval)$, which we constructed in \Cref{sec:input-succinct}.
\end{itemize}

Given the above building blocks, we give the following construction of classical obfuscation for pseudo-deterministic quantum circuits.

\begin{figure}[H]
\centering
\fbox{%
\parbox{0.96\textwidth}{%
\textbf{Classical obfuscation of pseudo-deterministic quantum circuits}

Let $U$ be the universal quantum circuit that takes as input the description $P_Q$ of a circuit of input of size $N$, where $N$ is the length of an input to $Q$. Let $D$ be the depth of $U$, and let $S$ be a bound on the size of $U$. Let $U_x$ be the quantum circuit that takes as input the description of a circuit $P_Q$ and runs the universal quantum circuit $U$ on $P_Q \| x$.
\begin{itemize}
    \item $\Obf(1^\secp, P_Q)$:
    \begin{itemize}
        \item Sample $(\pk,\sk, \PP) \gets \pvQFHE.\Gen(1^\secp,D, |P_Q|, S)$, and $\ct \gets \pvQFHE.\Enc(\pk, P_Q)$.
        \item Let $\mathsf{DK}(x,\ct',\pi)$ be the following functionality. First, run $\pvQFHE.\Ver^{\mathsf{PP}}(\ct, U_x,\ct',\pi)$. If the output was $\bot$, then output $\bot$, and otherwise output $\pvQFHE.\Dec(\sk,\ct')$.
        \item Sample $\widetilde{\mathsf{DK}} \gets \SObf(1^\secp,\mathsf{DK})$.
        \item Output $\widetilde{Q} \coloneqq \left(\ct,\PP,\widetilde{\mathsf{DK}}\right)$. 
    \end{itemize}
    \item $\Eval(P_{\widetilde{Q}},x)$:
    \begin{itemize}
        \item Parse $P_{\widetilde{Q}}$ as $(\ct,\mathsf{PP},\widetilde{\mathsf{DK}})$.
        \item Compute $(\ct',\pi) \gets \Eval^{{\mathsf{PP}}}(\ct,U_x)$.
        \item Output $b \coloneqq \widetilde{\mathsf{DK}}(x,\ct',\pi)$.
    \end{itemize}
\end{itemize}}}
\caption{Succinct classical obfuscation of pseudo-deterministic quantum circuits}\label{fig:obf}
\end{figure}

\begin{theorem}
    The construction in \Cref{fig:obf} satisfies correctness, security and succinctness definitions in \Cref{def:classical-obf-pd-circuits}.
\end{theorem}
The proofs of correctness and security are almost identical to the proof of Theorem 6.1 of \cite{bartusek2023obfuscation}.
\begin{proof}
First, correctness follows immediately from the correctness of the VBB obfuscator and the correctness of the publicly-verifiable QFHE scheme (\Cref{def:PVQFHE}). Note that even though the evaluation procedure may include measurements, an evaluator could run coherently, measure just the output bit $b$, and reverse. By the Gentle Measurement lemma, this implies the ability to run the obfuscated program on any $\poly(\secp)$ number of inputs.

Next, we show security. For any QPT adversary $\{\calA_\secp\}_{\secp \in \bbN}$, we define a simulator $\{\calS_\secp\}_{\secp \in \bbN}$ as follows, where $\{\widetilde{\calA}_\secp\}_{\secp \in \bbN}$ is the simulator for the classical obfuscation scheme $(\CObf,\CEval)$, defined based on $\{\calA_\secp\}_{\secp \in \bbN}$.
\begin{itemize}
    \item Sample $(\pk,\sk, \PP) \gets \pvQFHE.\Gen(1^\secp,D, |Q|, S)$, $\ct \gets \pvQFHE.\Enc(\pk,0^{|Q|})$.
    \item Run $\widetilde{\calA}_\secp^{\mathsf{PP},\mathsf{DK}}(\ct)$, answering $\mathsf{PP}$ calls honestly, and $\mathsf{DK}$ calls as follows.
    \begin{itemize}
        \item Take $(x,\ct',\pi)$ as input.
        \item Run $\Ver^{\mathsf{PP}}(x,\ct',\pi)$. If the output was $\bot$ then output $\bot$.
        \item Otherwise, forward $x$ to the external oracle $O[Q]$, and return the result $b = O[Q](x)$.
    \end{itemize}
    \item Output $\widetilde{\calA}_\secp$'s output.
\end{itemize}

Now, for any circuit $Q$, we define a sequence of hybrids.

\begin{itemize}
    \item Hybrid 0: Sample $\widetilde{Q} \gets \Obf(1^\secp,Q)$ and run $\calA_\secp(1^\secp,\widetilde{Q})$.
    \item Hybrid 1: Sample $(\ct,\mathsf{PP},\mathsf{DK})$ as in $\Obf(1^\secp,Q)$, and run $\widetilde{\calA}_\secp^{\mathsf{PP},\mathsf{DK}}(\ct)$.
    \item Hybrid 2: Same as hybrid 1, except that calls to $\mathsf{DK}$ are answered as in the description of $\calS_\secp$.
    \item Hybrid 3: Same as hybrid 2, except that we sample $\ct \gets \Enc(\pk,0^{|Q|})$. This is $\calS_\secp$.
\end{itemize}

We complete the proof by showing the following.

\begin{itemize}
    \item $|\Pr[1 \leftarrow \text{Hybrid }0] - \Pr[1 \leftarrow \text{Hybrid }1]| = \negl(\secp)$. This follows from the security of the classical obfuscation scheme $(\CObf,\CEval)$.
    \item $|\Pr[1 \leftarrow \text{Hybrid }1] - \Pr[1 \leftarrow \text{Hybrid }2]| = \negl(\secp)$. Suppose otherwise. Then there must exist some query made by $\widetilde{\calA}_\secp$ to $\mathsf{DK}$ with noticeable amplitude on $(x,\widetilde{\ct},\pi)$ such that $\mathsf{DK}$ does not return $\bot$ but $\Dec(\sk,\widetilde{\ct}) \neq Q(x)$. Thus, we can measure a random one of the $\poly(\secp)$ many queries made by $\widetilde{\calA}_\secp$ to obtain such an $(x,\widetilde{\ct},\pi)$, which violates the soundness of the publicly-verifiable QFHE scheme (\Cref{def:PVQFHE}).
    \item $|\Pr[1 \leftarrow \text{Hybrid }2] - \Pr[1 \leftarrow \text{Hybrid }3]| = \negl(\secp)$. Since $\sk$ is no longer used in hybrid 2 to respond to $\mathsf{DK}$ queries, this follows directly from the security of the publicly-verifiable QFHE scheme (\Cref{def:PVQFHE}).
\end{itemize}

Finally, the scheme is succinct. This follows from the compactness of the $\pvQFHE$ scheme and the input-succinct obfuscation $(\SObf, \SEval)$ scheme.
\end{proof}

\section{Acknowledgements} 
Aparna Gupte would like to thank Yael Kalai and Vinod Vaikuntanathan for valuable discussions. Aparna Gupte was supported in part by DARPA under Agreement No. HR00112020023, NSF CNS-2154149 and a Simons Investigator Award. This work was done in part while Aparna Gupte was at the Simons Institute and participating in the Challenge Institute for Quantum Computation at UC Berkeley, and while Omri Shmueli was a research fellow at the Simons Institute for the Theory of Computing at UC Berkeley.

\bibliographystyle{alpha}
\bibliography{refs}

@inproceedings{aaronson2012quantum,
  title={Quantum money from hidden subspaces},
  author={Aaronson, Scott and Christiano, Paul},
  booktitle={Proceedings of the forty-fourth annual ACM symposium on Theory of computing},
  pages={41--60},
  year={2012}
}

@inproceedings{coladangelo2021hidden,
  title={Hidden cosets and applications to unclonable cryptography},
  author={Coladangelo, Andrea and Liu, Jiahui and Liu, Qipeng and Zhandry, Mark},
  booktitle={Annual International Cryptology Conference},
  pages={556--584},
  year={2021},
  organization={Springer}
}

@article{zhandry2013note,
  title={A note on the quantum collision and set equality problems},
  author={Zhandry, Mark},
  journal={arXiv preprint arXiv:1312.1027},
  year={2013}
}

@article{shor1999polynomial,
  title={Polynomial-time algorithms for prime factorization and discrete logarithms on a quantum computer},
  author={Shor, Peter W},
  journal={SIAM review},
  volume={41},
  number={2},
  pages={303--332},
  year={1999},
  publisher={SIAM}
}

@inproceedings{bartusek2021secure,
  title={Secure quantum computation with classical communication},
  author={Bartusek, James},
  booktitle={Theory of Cryptography Conference},
  pages={1--30},
  year={2021},
  organization={Springer}
}

@inproceedings{shmueli2022public,
  title={Public-key quantum money with a classical bank},
  author={Shmueli, Omri},
  booktitle={Proceedings of the 54th Annual ACM SIGACT Symposium on Theory of Computing},
  pages={790--803},
  year={2022}
}

@inproceedings{shmueli2022semi,
  title={Semi-quantum tokenized signatures},
  author={Shmueli, Omri},
  booktitle={Annual International Cryptology Conference},
  pages={296--319},
  year={2022},
  organization={Springer}
}

@article{mahadev2020classical,
  title={Classical homomorphic encryption for quantum circuits},
  author={Mahadev, Urmila},
  journal={SIAM Journal on Computing},
  volume={52},
  number={6},
  pages={FOCS18--189},
  year={2020},
  publisher={SIAM}
}

@article{micalisnark,
  title={Computationally Sound Proofs},
  author={Micali, Silvio},
  volume={52},
  number={6},
  pages={FOCS},
  year={1994},
  publisher={SIAM}
}

@article{classicalsuccintobf,
  title={On Extractability (a.k.a. Differing-Inputs) Obfuscation},
  author={Boyle, Elette and Chung, Kai-Min and Pass, Rafael},
  pages={TCC},
  year={2014}
}

@article{huang2025obfuscation,
  title={Obfuscation of Unitary Quantum Programs},
  author={Huang, Mi-Ying and Tang, Er-Cheng},
  journal={arXiv preprint arXiv:2507.11970},
  year={2025}
}

@inproceedings{zhandrycollapse,
  title={New Constructions of Collapsing Hashes},
  author={Zhandry, Mark},
  booktitle={Annual international cryptology conference},
  pages={},
  year={2022},
  organization={Springer}
}

@inproceedings{barak2001possibility,
  title={On the (im) possibility of obfuscating programs},
  author={Barak, Boaz and Goldreich, Oded and Impagliazzo, Rusell and Rudich, Steven and Sahai, Amit and Vadhan, Salil and Yang, Ke},
  booktitle={Annual international cryptology conference},
  pages={1--18},
  year={2001},
  organization={Springer}
}

@inproceedings{goldwasser2007best,
  title={On best-possible obfuscation},
  author={Goldwasser, Shafi and Rothblum, Guy N},
  booktitle={Theory of Cryptography Conference},
  pages={194--213},
  year={2007},
  organization={Springer}
}

@inproceedings{wee2021candidate,
  title={Candidate obfuscation via oblivious LWE sampling},
  author={Wee, Hoeteck and Wichs, Daniel},
  booktitle={Annual International Conference on the Theory and Applications of Cryptographic Techniques},
  pages={127--156},
  year={2021},
  organization={Springer}
}

@article{brakerski2023candidate,
  title={Candidate iO from homomorphic encryption schemes},
  author={Brakerski, Zvika and D{\"o}ttling, Nico and Garg, Sanjam and Malavolta, Giulio},
  journal={Journal of Cryptology},
  volume={36},
  number={3},
  pages={27},
  year={2023},
  publisher={Springer}
}

@inproceedings{chen2018ggh15,
  title={GGH15 beyond permutation branching programs: proofs, attacks, and candidates},
  author={Chen, Yilei and Vaikuntanathan, Vinod and Wee, Hoeteck},
  booktitle={Annual International Cryptology Conference},
  pages={577--607},
  year={2018},
  organization={Springer}
}

@inproceedings{bartusek2018return,
  title={Return of GGH15: provable security against zeroizing attacks},
  author={Bartusek, James and Guan, Jiaxin and Ma, Fermi and Zhandry, Mark},
  booktitle={Theory of Cryptography Conference},
  pages={544--574},
  year={2018},
  organization={Springer}
}

@inproceedings{gentry2015graph,
  title={Graph-induced multilinear maps from lattices},
  author={Gentry, Craig and Gorbunov, Sergey and Halevi, Shai},
  booktitle={Theory of Cryptography Conference},
  pages={498--527},
  year={2015},
  organization={Springer}
}

@inproceedings{jain2021indistinguishability,
  title={Indistinguishability obfuscation from well-founded assumptions},
  author={Jain, Aayush and Lin, Huijia and Sahai, Amit},
  booktitle={Proceedings of the 53rd annual ACM SIGACT symposium on theory of computing},
  pages={60--73},
  year={2021}
}

@article{zhandry2021construct,
  title={How to construct quantum random functions},
  author={Zhandry, Mark},
  journal={Journal of the ACM (JACM)},
  volume={68},
  number={5},
  pages={1--43},
  year={2021},
  publisher={ACM New York, NY}
}

@article{garg2016candidate,
  title={Candidate indistinguishability obfuscation and functional encryption for all circuits},
  author={Garg, Sanjam and Gentry, Craig and Halevi, Shai and Raykova, Mariana and Sahai, Amit and Waters, Brent},
  journal={SIAM Journal on Computing},
  volume={45},
  number={3},
  pages={882--929},
  year={2016},
  publisher={SIAM}
}

@inproceedings{ananth2022pseudorandom,
  title={Pseudorandom (function-like) quantum state generators: New definitions and applications},
  author={Ananth, Prabhanjan and Gulati, Aditya and Qian, Luowen and Yuen, Henry},
  booktitle={Theory of Cryptography Conference},
  pages={237--265},
  year={2022},
  organization={Springer}
}

@inproceedings{zhandry2019record,
  title={How to record quantum queries, and applications to quantum indifferentiability},
  author={Zhandry, Mark},
  booktitle={Annual International Cryptology Conference},
  pages={239--268},
  year={2019},
  organization={Springer}
}

@incollection{diffie2022new,
  title={New directions in cryptography},
  author={Diffie, Whitfield and Hellman, Martin E},
  booktitle={Democratizing cryptography: the work of Whitfield Diffie and Martin Hellman},
  pages={365--390},
  year={2022}
}

@inproceedings{bartusek2024quantum,
  title={Quantum State Obfuscation from Classical Oracles},
  author={Bartusek, James and Brakerski, Zvika and Vaikuntanathan, Vinod},
  booktitle={Proceedings of the 56th Annual ACM Symposium on Theory of Computing},
  pages={1009--1017},
  year={2024}
}

@inproceedings{bartusek2023obfuscation,
  title={Obfuscation of pseudo-deterministic quantum circuits},
  author={Bartusek, James and Kitagawa, Fuyuki and Nishimaki, Ryo and Yamakawa, Takashi},
  booktitle={Proceedings of the 55th Annual ACM Symposium on Theory of Computing},
  pages={1567--1578},
  year={2023}
}

@article{shmueli2025one,
  title={On One-Shot Signatures, Quantum vs Classical Binding, and Obfuscating Permutations},
  author={Shmueli, Omri and Zhandry, Mark},
  journal={Cryptology ePrint Archive},
  year={2025}
}

@InProceedings{chiesa2019snark,
author="Chiesa, Alessandro
and Manohar, Peter
and Spooner, Nicholas",
editor="Hofheinz, Dennis
and Rosen, Alon",
title="Succinct Arguments in the Quantum Random Oracle Model",
booktitle="Theory of Cryptography",
year="2019",
publisher="Springer International Publishing",
address="Cham",
pages="1--29",
isbn="978-3-030-36033-7"
}

@INPROCEEDINGS{6375347,
  author={Zhandry, Mark},
  booktitle={2012 IEEE 53rd Annual Symposium on Foundations of Computer Science}, 
  title={How to Construct Quantum Random Functions}, 
  year={2012},
  volume={},
  number={},
  pages={679-687},
  doi={10.1109/FOCS.2012.37}}

@INPROCEEDINGS{zhacomporacle,
  author={Zhandry, Mark},
  booktitle={Annual International Cryptology Conference}, 
  title={How to Record Quantum Queries, and
Applications to Quantum Indifferentiability}, 
  year={2019},
  volume={},
  number={},
  pages={},
  doi={}}

@inproceedings{amos2020one,
  title={One-shot signatures and applications to hybrid quantum/classical authentication},
  author={Amos, Ryan and Georgiou, Marios and Kiayias, Aggelos and Zhandry, Mark},
  booktitle={Proceedings of the 52nd Annual ACM SIGACT Symposium on Theory of Computing},
  pages={255--268},
  year={2020}
}

@InProceedings{bra18,
author="Brakerski, Zvika",
editor="Shacham, Hovav
and Boldyreva, Alexandra",
title="Quantum FHE (Almost) As Secure As Classical",
booktitle="Advances in Cryptology -- CRYPTO 2018",
year="2018",
publisher="Springer International Publishing",
address="Cham",
pages="67--95",
isbn="978-3-319-96878-0"
}

@INPROCEEDINGS{mah18,
  author={Mahadev, Urmila},
  booktitle={2018 IEEE 59th Annual Symposium on Foundations of Computer Science (FOCS)}, 
  title={Classical Verification of Quantum Computations}, 
  year={2018},
  volume={},
  number={},
  pages={259-267},
  doi={10.1109/FOCS.2018.00033}}

@InProceedings{BM22,
  author =	{Bartusek, James and Malavolta, Giulio},
  title =	{{Indistinguishability Obfuscation of Null Quantum Circuits and Applications}},
  booktitle =	{13th Innovations in Theoretical Computer Science Conference (ITCS 2022)},
  pages =	{15:1--15:13},
  series =	{Leibniz International Proceedings in Informatics (LIPIcs)},
  ISBN =	{978-3-95977-217-4},
  ISSN =	{1868-8969},
  year =	{2022},
  volume =	{215},
  editor =	{Braverman, Mark},
  publisher =	{Schloss Dagstuhl -- Leibniz-Zentrum f{\"u}r Informatik},
  address =	{Dagstuhl, Germany},
  URL =		{https://drops.dagstuhl.de/entities/document/10.4230/LIPIcs.ITCS.2022.15},
  URN =		{urn:nbn:de:0030-drops-156115},
  doi =		{10.4230/LIPIcs.ITCS.2022.15}
}

@inproceedings{chung2022constant,
  title={Constant-round blind classical verification of quantum sampling},
  author={Chung, Kai-Min and Lee, Yi and Lin, Han-Hsuan and Wu, Xiaodi},
  booktitle={Annual International Conference on the Theory and Applications of Cryptographic Techniques},
  pages={707--736},
  year={2022},
  organization={Springer}
}

@inproceedings{bartusek2022succinct,
  title={Succinct classical verification of quantum computation},
  author={Bartusek, James and Kalai, Yael Tauman and Lombardi, Alex and Ma, Fermi and Malavolta, Giulio and Vaikuntanathan, Vinod and Vidick, Thomas and Yang, Lisa},
  booktitle={Annual International Cryptology Conference},
  pages={195--211},
  year={2022},
  organization={Springer}
}

@inproceedings{DFMS19,
author = {Don, Jelle and Fehr, Serge and Majenz, Christian and Schaffner, Christian},
title = {Security of the Fiat-Shamir Transformation in the Quantum Random-Oracle Model},
year = {2019},
isbn = {978-3-030-26950-0},
publisher = {Springer-Verlag},
address = {Berlin, Heidelberg},
url = {https://doi.org/10.1007/978-3-030-26951-7_13},
doi = {10.1007/978-3-030-26951-7_13},
booktitle = {Advances in Cryptology – CRYPTO 2019: 39th Annual International Cryptology Conference, Santa Barbara, CA, USA, August 18–22, 2019, Proceedings, Part II},
pages = {356–383},
numpages = {28},
location = {Santa Barbara, CA, USA}
}

@inproceedings{DFM20,
author = {Don, Jelle and Fehr, Serge and Majenz, Christian},
title = {The Measure-and-Reprogram Technique 2.0: Multi-round Fiat-Shamir and More},
year = {2020},
isbn = {978-3-030-56876-4},
publisher = {Springer-Verlag},
address = {Berlin, Heidelberg},
url = {https://doi.org/10.1007/978-3-030-56877-1_21},
doi = {10.1007/978-3-030-56877-1_21},
booktitle = {Advances in Cryptology – CRYPTO 2020: 40th Annual International Cryptology Conference, CRYPTO 2020, Santa Barbara, CA, USA, August 17–21, 2020, Proceedings, Part III},
pages = {602–631},
numpages = {30},
location = {Santa Barbara, CA, USA}
}

@inproceedings{GKNV25,
author = {Gunn, Sam and Kalai, Yael and Natarajan, Anand and Vill\'{a}nyi, \'{A}gi},
title = {Classical Commitments to Quantum States},
year = {2025},
isbn = {9798400715105},
publisher = {Association for Computing Machinery},
address = {New York, NY, USA},
url = {https://doi.org/10.1145/3717823.3718264},
doi = {10.1145/3717823.3718264},
booktitle = {Proceedings of the 57th Annual ACM Symposium on Theory of Computing},
pages = {234–244},
numpages = {11},
location = {Prague, Czechia},
series = {STOC '25}
}

@INPROCEEDINGS {NZ23,
author = { Natarajan, Anand and Zhang, Tina },
booktitle = { 2023 IEEE 64th Annual Symposium on Foundations of Computer Science (FOCS) },
title = {{ Bounding the Quantum Value of Compiled Nonlocal Games: From CHSH to BQP Verification }},
year = {2023},
volume = {},
ISSN = {},
pages = {1342-1348},
doi = {10.1109/FOCS57990.2023.00081},
url = {https://doi.ieeecomputersociety.org/10.1109/FOCS57990.2023.00081},
publisher = {IEEE Computer Society},
address = {Los Alamitos, CA, USA},
month =Nov}

@INPROCEEDINGS {MNZ24,
author = { Metger, Tony and Natarajan, Anand and Zhang, Tina },
booktitle = { 2024 IEEE 65th Annual Symposium on Foundations of Computer Science (FOCS) },
title = {{ Succinct Arguments for QMA from Standard Assumptions via Compiled Nonlocal Games }},
year = {2024},
volume = {},
ISSN = {},
pages = {1193-1201},
doi = {10.1109/FOCS61266.2024.00078},
url = {https://doi.ieeecomputersociety.org/10.1109/FOCS61266.2024.00078},
publisher = {IEEE Computer Society},
address = {Los Alamitos, CA, USA},
month =Oct}

@inproceedings{BGLPT15,
author = {Bitansky, Nir and Garg, Sanjam and Lin, Huijia and Pass, Rafael and Telang, Sidharth},
title = {Succinct Randomized Encodings and their Applications},
year = {2015},
isbn = {9781450335362},
publisher = {Association for Computing Machinery},
address = {New York, NY, USA},
url = {https://doi.org/10.1145/2746539.2746574},
doi = {10.1145/2746539.2746574},
booktitle = {Proceedings of the Forty-Seventh Annual ACM Symposium on Theory of Computing},
pages = {439–448},
numpages = {10},
location = {Portland, Oregon, USA},
series = {STOC '15}
}

@inproceedings{BK21,
author = {Broadbent, Anne and Kazmi, Raza Ali},
title = {Constructions for Quantum Indistinguishability Obfuscation},
year = {2021},
isbn = {978-3-030-88237-2},
publisher = {Springer-Verlag},
address = {Berlin, Heidelberg},
url = {https://doi.org/10.1007/978-3-030-88238-9_2},
doi = {10.1007/978-3-030-88238-9_2},
booktitle = {Progress in Cryptology – LATINCRYPT 2021: 7th International Conference on Cryptology and Information Security in Latin America, Bogot\'{a}, Colombia, October 6–8, 2021, Proceedings},
pages = {24–43},
numpages = {20},
location = {Bogot\'{a}, Colombia}
}

@inproceedings{SP20,
author = {Agrawal, Shweta and Pellet-Mary, Alice},
title = {Indistinguishability Obfuscation Without Maps: Attacks and Fixes for Noisy Linear FE},
year = {2020},
isbn = {978-3-030-45720-4},
publisher = {Springer-Verlag},
address = {Berlin, Heidelberg},
url = {https://doi.org/10.1007/978-3-030-45721-1_5},
doi = {10.1007/978-3-030-45721-1_5},
booktitle = {Advances in Cryptology – EUROCRYPT 2020: 39th Annual International Conference on the Theory and Applications of Cryptographic Techniques, Zagreb, Croatia, May 10–14, 2020, Proceedings, Part I},
pages = {110–140},
numpages = {31},
location = {Zagreb, Croatia}
}
\appendix
\section{Remaining proofs from \Cref{sec:pfc}}\label{sec:appendix-proofs-pfc}

\subsection{Correctness}
\begin{proof}[Proof of \Cref{thm:pfc-correctness}]
This proof is almost identical to the proof of Theorem 4.6 of \cite{bartusek2023obfuscation}, but we repeat it here for completeness. The only difference is that we are using one-shot signatures instead of signature tokens; this is what allows us to make the commitment keys entirely classical.

We will assume that the one-shot signature schemes $\OSS, \overline{\OSS}$ are perfectly correct. If they are statistically correct, then the Pauli functional commitment scheme we obtain is still correct, because we allow for a negligible statistical distance.

We will first show that the map applied by $\Com^{\CK}(b)$ in the case that the measurement of the first qubit of $\calS'$ is $1-b$ successfully takes $\ket{S_{y,1-b}} \to \ket{S_{y,b}}$. Since we are assuming perfect correctness from $\overline{\OSS}$, it suffices to show that for any balanced affine subspace $\ket{S_y}$, \[H^{\otimes n}\mathsf{Phase}^{O[S_y^\bot]}H^{\otimes n}\ket{S_{y,1-b}} \to \ket{S_{y,b}},\] where $\mathsf{Phase}^{O[S_y^\bot]}$ is the map $\ket{s} \to (-1)^{O[S_y^\bot](s)}\ket{s}$, and $O[S_y^\bot]$ is the oracle that outputs 0 if $s \in S^\bot$ and 1 if $s \notin S^\bot$. This was also shown in \cite{amos2020one}.

Let $S_y^*$ be the subspace parallel to affine subspace $S_y$, and let $S_{y,b}^*$ be the coset of $S_y^*$ whose first coordinate is $b$.
We will use the facts that $S_{y,1}^* = S_{y,0}^* + w$ for some $w$, and that $S_{y, 0} = S^*_{y,0} + v_0$ and $S_{y, 1} = S_{y,1}^* + v_1$ for some $v_0,v_1$ such that $v_0 + v_1 = w$. Also note that for any $s \in S_y^\bot$, $s \cdot w = 0$, and for any $s \in ({S^*_{y,0}})^\bot \setminus (S^*_y)^\bot$, $s \cdot w = 1$.

\begin{align*}
    H^{\otimes n}&\mathsf{Ph}^{O[S_y^\bot]}H^{\otimes n}\ket{S_{y,1-b}} \\ 
    &= H^{\otimes n} \mathsf{Ph}^{O[S_y^\bot]}H^{\otimes n}\frac{1}{\sqrt{2^{n/2 - 1}}}\left(\sum_{s \in S^*_{y,0}} \ket{s + v_{1-b}}\right) \\ 
    &= H^{\otimes n}\mathsf{Ph}^{O[S_y^\bot]}\frac{1}{\sqrt{2^{n/2+1}}}\left(\sum_{s \in (S_{y,0}^*)^\bot}(-1)^{s \cdot v_{1-b}}\ket{s}\right) \\ 
    &= H^{\otimes n}\mathsf{Ph}^{O[S_y^\bot]}\frac{1}{\sqrt{2^{n/2+1}}}\left(\sum_{s \in {(S_y^*)}^\bot} (-1)^{s \cdot w + s \cdot v_{b}}\ket{s} + \sum_{s \in (S_{y,0}^*)^\bot \setminus (S_y^*)^\bot} (-1)^{s \cdot w + s \cdot v_{b}} \ket{s}\right) \\ 
    &= H^{\otimes n}\mathsf{Ph}^{O[S_y^\bot]}\frac{1}{\sqrt{2^{n/2+1}}}\left(\sum_{s \in (S_y^*)^\bot} (-1)^{s \cdot v_{b}}\ket{s} + \sum_{s \in (S_{y,0}^*)^\bot \setminus (S_y^*)^\bot} (-1)^{1 + s \cdot v_{b}} \ket{s}\right) \\ 
    &= H^{\otimes n} \frac{1}{\sqrt{2^{n/2+1}}}\left(\sum_{s \in S^\bot} (-1)^{s \cdot v_b}\ket{s} + \sum_{s \in (S_{y,0}^*)^\bot \setminus (S_y^*)^\bot} (-1)^{s \cdot v_b} \ket{s}\right) \\ 
    &= H^{\otimes n}\frac{1}{\sqrt{2^{n/2+1}}}\left(\sum_{s \in (S_{y,0}^*)^\bot}(-1)^{s \cdot v_b}\ket{s}\right) \\ 
    &= \ket{S_{y,b}}.
\end{align*}

Now consider applying $\Com^{\CK}$ to a pure state. By the correctness of $\OSS$, $\Com$ can produce a valid signature $\sigma_0$ of $0$, and is able to access the $\overline{\O} = (P, P^{-1}, D)$ oracles of $\OSS$ through $\CK_{P}(\vk, \cdot)$, $\CK_{P^{-1}}(\vk, \cdot, \cdot)$ and $\CK_D(\vk, \sigma_0)$. These oracles $\overline{O} = (P, P^{-1}, D)$ are an instance of $\overline{\OSS}$ generated using randomness determined by $\RO(\vk)$. By the correctness of $\overline{\OSS}.\Gen^{P, P^{-1}, D}$, $\Com$ is able to generate a $(\overline{vk}, \ket{\overline{\sk}})$ pair $(y, \ket{S_y})$. 

Applying $\Com^\CK$ to pure state $\ket{\psi} = \alpha \ket{0} + \beta\ket{1}$
produces (up to negligible trace distance) the state \[\ket{\psi_\Com} = \alpha_0\ket{0}\ket{S_{y,0}} + \alpha_1\ket{1}\ket{S_{y,1}},\] and a signature $c$ on the bit 1.

We claim that measuring and decoding $\ket{\psi_\Com}$ in the standard (resp. Hadamard) basis produces the same distribution as directly measuring $\ket{\psi}$ in the standard (resp. Hadamard) basis. As a mixed state is a probability distribution over pure states, this will complete the proof of correctness.

First, it is immediate that measuring $\ket{\psi_\Com}$ in the standard basis produces a bit $b$ with probability $|\alpha_b|^2$ along with a vector $s$ such that $s \in (S+v)_b$. 

Next, note that applying Hadamard to each qubit of $\ket{\psi_\Com}$ except the first results in the state 

\[\alpha_0 \ket{0} \left(\sum_{s \in (S_{y,0}^*)^\bot}(-1)^{s \cdot v_0}\ket{s}\right) + \alpha_1 \ket{1} \left(\sum_{s \in (S_{y,0}^*)^\bot}(-1)^{s \cdot v_1}\ket{s}\right),\] and thus, measuring each of these qubits (except the first) in the Hadamard basis produces a vector $s$ and a single-qubit state \[(-1)^{s \cdot v_0}\alpha_0\ket{0} + (-1)^{s \cdot v_1} \alpha_1\ket{1} = \alpha_0 \ket{0} + (-1)^{s \cdot w}\alpha_1 \ket{1}.\]
So, measuring this qubit in the Hadamard basis is equivalent to measuring $\ket{\psi}$ in the Hadamard basis and masking the result with $s \cdot w$. Recalling that $s \cdot w = 0$ if $s \in (S_{y}^*)^\bot$ and $s \cdot w = 1$ if $s \in (S_{y,0}^*)^\bot \setminus (S_y^*)^\bot$ completes the proof of correctness, since $S_y^\bot = (S_y^*)^\bot$ and $S_{y,0}^\bot = (S_{y,0}^*)^\bot$
\end{proof}

\subsection{Binding}
\begin{lemma}\label{lem:collapsing-single-bit}
    The collapse-binding definition of \Cref{def:collapse-binding} implies the single-bit binding definition of \Cref{def:binding}.
\end{lemma}
\begin{proof}
We will first define the following notion of binding, and then establish that the collapse-binding definition of  \Cref{def:collapse-binding} implies this collapse binding definition.
\begin{definition}[Single bit collapse binding]\label{def:singlebitcollapsebinding}
For any adversary $\mathcal{A} := \{(\calA_1, \calA_2)_\secp\}_{\secp \in \mathbb{N}}$, where each of $\calA_1$ and $\calA_2$ are oracle-aided quantum operations that make at most $\poly(\secp)$ oracle queries, define the single-bit collapse binding experiment $\mathsf{SingleBitCollapseBindingExpt}^{\mathcal{A}}(1^\secp)$ as follows.
\begin{itemize}
  \item Sample $(\CK, \dk) \leftarrow \mathsf{Gen}(1^\secp)$.

  \item Run $\calA_1^{\mathsf{CK}, \mathsf{DecZ}[\dk], \mathsf{DecX}[\dk]}(1^\lambda)$ until it outputs a commitment $c$ and a state on registers $(\mathcal{B}, \mathcal{U}, \mathcal{R})$. Here $\calB$ is a single-qubit register, $\calU$ is the opening register and $\calR$ holds the internal state of $\calA$.

  \item Sample $b \leftarrow \{0,1\}$. If $b = 0$, do nothing, and if $b=1$ measure $(\mathcal{B}, \mathcal{U})$ with $\{\Pi_{\dk, c,0}, \Pi_{\dk, c,1}\}$

  \item Run $\calA_2^{\mathsf{CK}, \mathsf{DecZ}[\dk]}(\mathcal{B}, \mathcal{U}, \mathcal{R})$ until it outputs a bit $b'$. 
  
  \item The experiment outputs 1 if $b = b'$.
\end{itemize}

We say that adversary $\mathcal{A}$ is \emph{valid} if the state on $(\mathcal{B}, \mathcal{U})$ output by $\calA_1$ is in the image of $|0\rangle\langle0| \otimes \Pi_{\dk, c,0} + |1\rangle\langle1| \otimes\Pi_{\dk, c,1}$. Then, we say that a Pauli functional commitment $(\mathsf{Gen}, \mathsf{Com}, \mathsf{OpenZ}, \mathsf{OpenX}, \mathsf{DecZ}, \mathsf{DecX})$ satisfies publicly-decodable single bit collapse-binding if it holds that for all valid adversaries $\mathcal{A}$,
\[
\left| \Pr\left[ \mathsf{SingleBitCollapseBindingExpt}^{\mathcal{A}}(1^\lambda) = 1 \right] - \frac{1}{2} \right| = \mathsf{negl}(\lambda).
\]
\end{definition}
\begin{lemma}
    The collapse binding definition of Definition \ref{def:collapse-binding} implies the single bit collapse binding definition of Definition \ref{def:singlebitcollapsebinding}.
\end{lemma}
\begin{proof}
We will establish the proof of this claim through a sequence of hybrids.
\paragraph{Hybrid 0:} This hybrid describes the following experiment 
\begin{itemize}
  \item Sample $(\CK, \dk) \leftarrow \mathsf{Gen}(1^\secp)$.

  \item Run $\calA_1^{\mathsf{CK}, \mathsf{DecZ}[\dk], \mathsf{DecX}[\dk]}(1^\lambda)$ until it outputs a commitment $c$ and a state on registers $(\mathcal{B}, \mathcal{U}, \mathcal{R})$. Here $\calB$ is a single-qubit register, $\calU$ is the opening register and $\calR$ holds the internal state of $\calA$.
\item Measure  $(\mathcal{B}, \mathcal{U})$ with $\{\Pi_{dk,c,0}, \Pi_{dk,c,1}\}$.

  \item Run $\calA_2^{\mathsf{CK}, \mathsf{DecZ}[\dk]}(\mathcal{B}, \mathcal{U}, \mathcal{R})$ until it outputs a bit $b'$.

\end{itemize}
\paragraph{Hybrid 1}
This hybrid is exactly the same as the previous hybrid, except that the computational basis measurement is performed instead of the single bit measurement $\{\Pi_{dk,c,0}, \Pi_{dk,c,1}\}$
\begin{itemize}
  \item Sample $(\CK, \dk) \leftarrow \mathsf{Gen}(1^\secp)$.

  \item Run $\calA_1^{\mathsf{CK}, \mathsf{DecZ}[\dk], \mathsf{DecX}[\dk]}(1^\lambda)$ until it outputs a commitment $c$ and a state on registers $(\mathcal{B}, \mathcal{U}, \mathcal{R})$. Here $\calB$ is a single-qubit register, $\calU$ is the opening register and $\calR$ holds the internal state of $\calA$.
\item Measure  $(\mathcal{B}, \mathcal{U})$ in the computational basis.

  \item Run $\calA_2^{\mathsf{CK}, \mathsf{DecZ}[\dk]}(\mathcal{B}, \mathcal{U}, \mathcal{R})$ until it outputs a bit $b'$.

\end{itemize}
\begin{claim}
If $(\mathsf{Gen}, \mathsf{Com}, \mathsf{OpenZ}, \mathsf{OpenX}, \mathsf{DecZ}, \mathsf{DecX})$ satisfies the collapse-binding definition of Definition \ref{def:collapse-binding}, then \[
\Big|\Pr[1\leftarrow \textbf{Hybrid 0}]- \Pr[1\leftarrow \textbf{Hybrid 1}]
  \Big|\leq \text{negl}(
  \lambda)\]  
\end{claim}
\begin{proof}
    Assume for contradiction that there exists an adversary $(\mathcal{A}_1, \mathcal{A}_2)$ such that $\Big|\Pr[1\leftarrow \textbf{Hybrid 1}]-\Pr[1\leftarrow \textbf{Hybrid 0}]
  \Big|\geq  \epsilon(
  \lambda)$, where $\epsilon(.)$ is some non negligible function. Then, we can construct an adversary $(\mathcal{B}_1, \mathcal{B}_2)$ that breaks the collapse binding property of  $(\mathsf{Gen}, \mathsf{Com}, \mathsf{OpenZ}, \mathsf{OpenX}, \mathsf{DecZ}, \mathsf{DecX})$. 
  \paragraph{Description of $(\mathcal{B}_1, \mathcal{B}_2)$:}
  \begin{itemize}
    \item $\calB_1$ runs $\calA_1^{\mathsf{CK}, \mathsf{DecZ}[\dk], \mathsf{DecX}[\dk]}(1^\lambda)$ until it outputs a commitment $c$ and a state on registers $(\mathcal{B}, \mathcal{U}, \mathcal{R})$. Here $\calB$ is a single-qubit register, $\calU$ is the opening register and $\calR$ holds the internal state of $\calA$.
\item $\mathcal{B}_1$ measures  $(\mathcal{B}, \mathcal{U})$ with $\{\Pi_{dk,c,0}, \Pi_{dk,c,1}\}$. 
\item Sample $b \leftarrow \{0,1\}$. If $b=0$ do nothing, and if $b=1$, measure $(\mathcal{B}, \mathcal{U})$ in the computational basis. 
\item  $\calB_2$ runs $\calA_2^{\mathsf{CK}, \mathsf{DecZ}[\dk]}(\mathcal{B}, \mathcal{U}, \mathcal{R})$ until it outputs $b'$.
\item  $\calB_2$ outputs $b'$.
\item The experiment outputs 1 if $b'=b$.
  
  \end{itemize}
  Observe that when $b=0$, this experiment is identical to $\textbf{Hybrid 0}$, and when $b=1$, the experiment is identical to $\textbf{Hybrid 1}$. Thus, $(\mathcal{B}_1, \mathcal{B}_2)$ breaks collapse binding (Definition \ref{def:collapse-binding}) with non negligible advantage. 
\end{proof}
\paragraph{Hybrid 2:} This is exactly the same as the previous hybrid, except no measurement is performed.
\begin{itemize}
    \item Sample $(\CK, \dk) \leftarrow \mathsf{Gen}(1^\secp)$.

  \item Run $\calA_1^{\mathsf{CK}, \mathsf{DecZ}[\dk], \mathsf{DecX}[\dk]}(1^\lambda)$ until it outputs a commitment $c$ and a state on registers $(\mathcal{B}, \mathcal{U}, \mathcal{R})$. Here $\calB$ is a single-qubit register, $\calU$ is the opening register and $\calR$ holds the internal state of $\calA$.

  \item Run $\calA_2^{\mathsf{CK}, \mathsf{DecZ}[\dk]}(\mathcal{B}, \mathcal{U}, \mathcal{R})$ until it outputs a bit $b'$. 
  
\end{itemize}
\begin{claim}
If $(\mathsf{Gen}, \mathsf{Com}, \mathsf{OpenZ}, \mathsf{OpenX}, \mathsf{DecZ}, \mathsf{DecX})$ satisfies the collapse-binding definition of Definition \ref{def:collapse-binding}, then \[
\Big|\Pr[1\leftarrow \textbf{Hybrid 2}]- \Pr[1\leftarrow \textbf{Hybrid 1}]
  \Big|\leq \text{negl}(
  \lambda)\]  
\end{claim}
\begin{proof}
  This follows directly from the collapse binding definition of Definition \ref{def:collapse-binding}. 
\end{proof}
\end{proof}
Now, Lemma A.1 of \cite{bartusek2023obfuscation} shows that single-bit collapse binding implies a notion of binding called ``unique-message binding'', and also shows that this notion implies string-binding. String-binding implies single-bit binding, and therefore we are done.

\end{proof}

\end{document}